\documentclass[
twocolumn,
superscriptaddress
]{revtex4-1}

\usepackage{amsmath,amsfonts, amssymb, amsthm}
\usepackage{bbm, physics, mathtools}
\usepackage{graphicx, subfigure}
\usepackage{enumerate, multirow, makecell}
\usepackage{xcolor, hyperref}

\hypersetup{
	colorlinks=true,  
	linkcolor=blue,   
	citecolor=blue,   
	urlcolor=blue     
}

\newtheorem{theorem}{Theorem}
\newtheorem{proposition}[theorem]{Proposition}
\newtheorem{definition}[theorem]{Definition}

\newtheorem{innercustomthm}{Theorem}
\newenvironment{customthm}[1]
  {\renewcommand\theinnercustomthm{#1}\innercustomthm}
  {\endinnercustomthm}

\newcommand{\ent}[2]{S\left( #1 \middle\vert\middle\vert #2 \right)}
\def\>{\rangle}
\def\<{\langle}

\def\id{\mathbbm{1}}
\renewcommand{\tr}{{\rm{tr}}}

\def\bmeta{\boldsymbol{\eta}}
\def\bmo{\boldsymbol{0}}
\def\bma{\boldsymbol{a}}
\let\ring\r

\def\x{\boldsymbol{x}}
\def\y{\boldsymbol{y}}
\def\z{\boldsymbol{z}}
\def\r{\boldsymbol{r}}
\def\w{\boldsymbol{w}}
\def\p{\boldsymbol{p}}
\def\q{\boldsymbol{q}}
\def\d{\boldsymbol{d}}
\def\m{\boldsymbol{m}}
\def\k{\boldsymbol{k}}
\def\q{\boldsymbol{q}}

\def\A{{\cal A}}
\def\Z{{\cal Z}}
\def\H{{\cal H}}
\def\E{{\cal E}}
\def\J{{\cal J}}
\def\R{{\cal R}}
\def\D{{\cal D}}
\def\M{{\cal M}}
\def\F{{\cal F}}
\renewcommand{\O}{{\cal O}}
\renewcommand{\P}{{\cal P}}

\begin{document}

\begin{abstract}
Magic states are key ingredients in schemes to realize universal fault-tolerant quantum computation.
Theories of magic states attempt to quantify this computational element via monotones and determine how these states may be efficiently transformed into useful forms. Here, we develop a statistical mechanical framework based on majorization to describe Wigner negative magic states for qudits of odd prime dimension processed under Clifford circuits. We show that majorization allows us to both quantify disorder in the Wigner representation and derive upper bounds for magic distillation. These bounds are shown to be tighter than other bounds, such as from mana and thauma, and can be used to incorporate hardware physics, such as temperature dependence and system Hamiltonians. We also show that a subset of single-shot R\'{e}nyi entropies remain well-defined on quasi-distributions, are fully meaningful in terms of data processing and can acquire negative values that signal magic. We find that the mana of a magic state is the measure of divergence of these R\'{e}nyi entropies as one approaches the Shannon entropy for Wigner distributions, and discuss how distillation lower bounds could be obtained in this setting.  This use of majorization for quasi-distributions could find application in other studies of non-classicality, and raises nontrivial questions in the context of classical statistical mechanics.
\end{abstract}

\preprint{APS/123-QED}

\title{Constraints on magic state protocols from the statistical mechanics of Wigner negativity}

\author{Nikolaos Koukoulekidis}
	\email{nk2314@imperial.ac.uk}
	\affiliation{Department of Physics, Imperial College London, London SW7 2AZ, UK}
\author{David Jennings}
	\affiliation{Department of Physics, Imperial College London, London SW7 2AZ, UK}
	\affiliation{School of Physics and Astronomy, University of Leeds, Leeds, LS2 9JT, UK}

\date{\today}
\maketitle

\section*{Introduction}

There has been rapid progress towards the goal of a fault-tolerant quantum computer~\cite{campbell_roads_2017, cit:raussendorf, Raussendorf_2013, Nickerson_2014, Nikahd_2017, chao_2018, lin_pieceable_2020, Lin_2020, chamberland_2020, Bourassa_2021}. Nevertheless, many challenges remain and there is increasing need for theory to take into account physical limitations of the hardware involved. The surface code~\cite{Bravyi_1998, Freedman_2001, Dennis_2002, Raussendorf_2007} is a leading framework for fault-tolerance with very high error thresholds. Within this scheme, Clifford unitaries can be implemented in a robust, fault-tolerant way. However, due to the Eastin-Knill theorem~\cite{Eastin_2009}, we also know that it is impossible to have a universal set of transversal gates, and although Clifford unitaries can be realized transversally~\cite{Calderbank_1996, Steane_1996}, one needs to find ways around the Eastin-Knill restriction. This can be achieved by injecting in quantum states, called magic states, which promote the Clifford group to universal quantum computing~\cite{cit:bravyi}. The obstacle to this is that these states are invariably noisy, so protocols involving stabilizer operations must be employed to purify many copies of the magic states and improve the overall performance of the induced quantum gates~\cite{cit:bravyi, Jones_2013, Ogorman_2017}. A key question then arises about the overhead on purifying many copies of a magic state into less noisy forms. 

To address this, concrete distillation protocols have been developed, such as the Bravyi-Haah qubit protocol that provides a quadratic reduction in noise per cycle~\cite{Bravyi2012}. Such distillation rates have been improved in more recent works~\cite{Jones_2013, haah2017magic, Hastings2018, Litinski_2019} and there have been proposals of protocols for qudits of odd dimension~\cite{CampbellAnwar_2012, Anwar_2012, Dawkins_2015, Krishna2019, cit:prakash} as well as of protocols within a full architectural analysis~\cite{chamberland_very_2020, Chamberland2019faulttolerantmagic}. There is also analysis of magic protocols from the perspective of magic monotones, which provide upper bounds on distillation rates and find application in the analysis of gate synthesis~\cite{Campbell_2017, Howard_2017, Prakash_2018, Seddon_2021,leone2021renyi}. These frameworks for magic view magic states as resource states with respect to a natural class of quantum operations that are considered cheap, or free, such as stabilizer operations~\cite{Gour_2019, cit:ahmadi, cit:seddon, Wang_2019}. Beyond quantum computing, recent work has also looked at how magic can be used to analyse many-body physics~\cite{Sarkar_2020} and conformal field theories~\cite{White_2021} as well its connection to contextuality~\cite{Vega_2017, cit:howard2, Zurel_2020, Okay_2021, schmid2020structure, Schmid_2021}.

Work towards fault-tolerance is increasingly bridging the gap between abstract theory and experiment. Extensive work is being done on error mitigation~\cite{jones_2012, Li_2017, Temme_2017, Endo_2018, McClean_2017}, and the incorporation of hardware physics into the theoretical models~\cite{Kandala_2019, holmes_resource_2019, Colless_2018, song2018quantum, Bravyi_2021}. For example, the XZZX code~\cite{bonilla_ataides_xzzx_2021} is a variant of the surface code that incorporates noise bias explicitly and has been shown to attain the hashing bound of random codes~\cite{Bennett_1996}. In this work, we develop a framework to analyse magic state distillation protocols where explicit physics of the system is incorporated in the distillation bounds. This is achieved by considering how a given magic state protocol transforms a pair of quantum states -- one being a noisy magic state, and the other a stabilizer state that is distinguished by the physics of the system (for example a state at a characteristic temperature, or the maximally mixed state) and acts as a reference state for the protocol. The magic distillation bounds can then be expressed in terms of the physics of the reference state (e.g.~free energy changes in the case of temperature). Such bounds are of potential interest in assessing how physical features such as temperature, noise biases or fixed-point structure associated with restricted gate-sets constrain distillation protocols~\cite{Tuckett_2019, Aliferis_2008, Stephens_2013, Li_2015, Babbush_2018, Guillaud_2019, Fowler_2019}.

The approach we take is most closely aligned with resource theories of magic, although it differs in key ways. We obtain distillation upper bounds without the use of monotones, but instead use tools from statistical mechanics and recent work in single-shot resource theories~\cite{cit:janzing, cit:horodecki2013, Brandao_2015, Vinjanampathy_2016, Goold_2016, cit:lostaglio, cit:gour}. Our analysis relies on the discrete Wigner representation of quantum systems, in which all states and operations can be described on a discrete phase space~\cite{Ferrie_2008, Okay_2021}. Crucially, we focus on magic states with negativity in their Wigner representation, which is known to be a necessary condition for universality in the state-injection model~\cite{cit:veitch, cit:mari, cit:gottesman, cit:knill, Campbell_2011}. However, taking a statistical mechanical perspective in this context raises a problem: the standard Boltzmann entropy is not defined for quasi-distributions. We circumvent this obstacle by making use of a more fundamental tool -- majorization theory~\cite{cit:marshall, Veinott_1971, Ruch_1976}. To our knowledge majorization of quasi-distributions has not been considered in quantum physics before, and therefore these methods could find application in other studies of non-classicality in quantum systems~\cite{Fine_1982, Allahverdyan_2018, arvidsson_2020, halpern_2018, Lostaglio_2018, Levy_2020}.

\begin{figure}[t]
    \centering
    \includegraphics[scale=0.38]{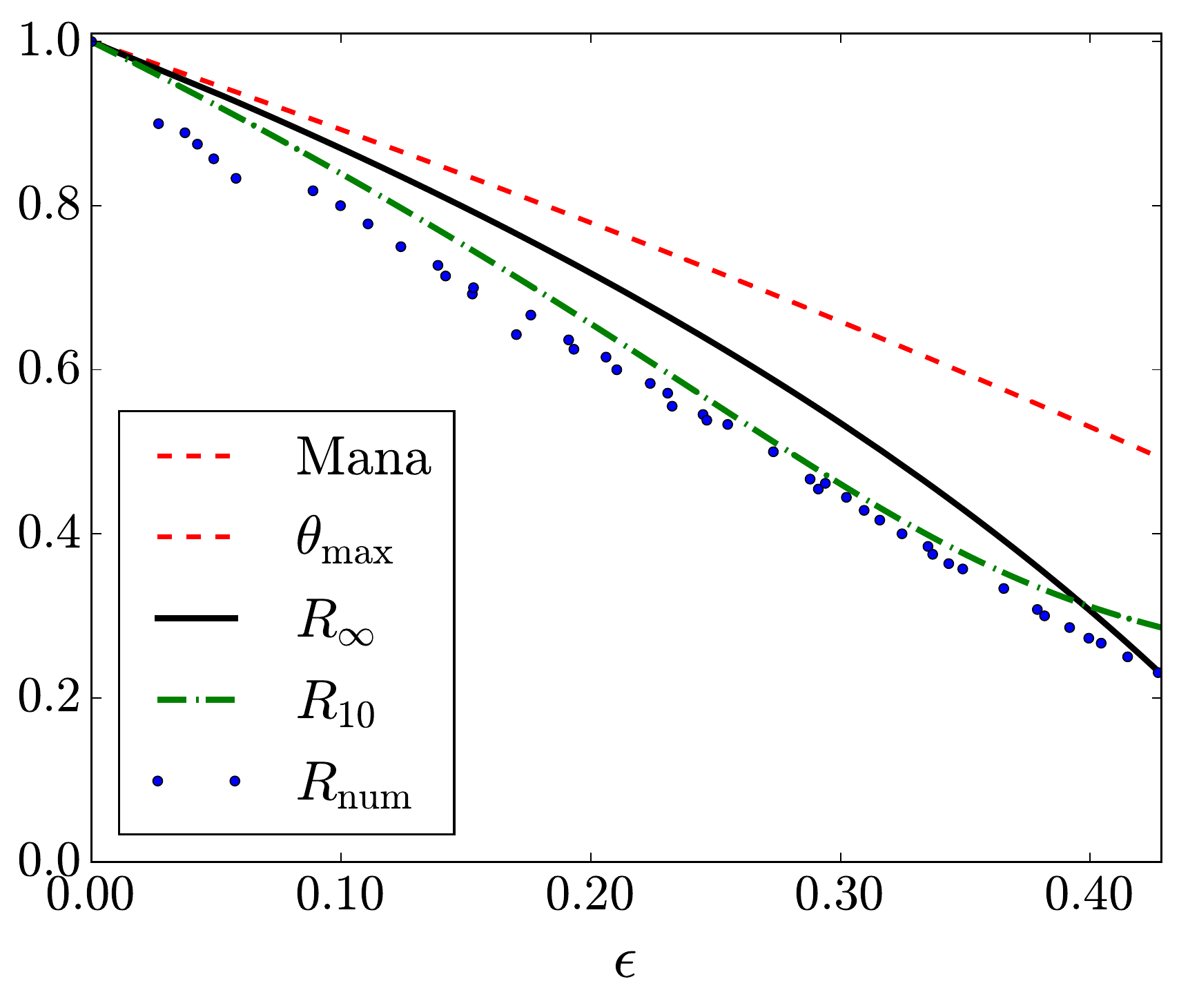}
    \caption{\textbf{Magic distillation bounds from majorization of Wigner quasi-distributions.} The above plot shows distillation upper bounds for qutrits obtained from majorization on unital protocols. The constraints are plotted as a function of the depolarizing error $\epsilon$ for a noisy Strange state (Eq.~(\ref{eq:noisysn})). The bound $R_\infty$ follows from a Lorenz curve analysis presented before Eq.~(\ref{eq:unital-bound}), while the bound $R_{10}$ comes from single-shot entropies on the Wigner quasi-distributions, as discussed in the derivation of Eq.~(\ref{eq:Ra}). Numerical bounds $R_{\rm{num}}$ are obtained by using the full set of majorization constraints. Bounds $R_\infty$, $R_{10}$, $R_{\rm{num}}$ all provide stricter constraints than mana~\cite{cit:veitch2} and max--thauma $\theta_{\rm{max}}$~\cite{Wang_2020} bounds for noisy Strange states.
    }
    \label{fig:distill_bounds}
\end{figure}

We begin by describing how stabilizer operations correspond to stochastic maps on a discrete phase space, and how majorization can be applied to constrain magic state transformations. We then consider families of magic protocols with increasing complexity. The simplest case we consider is for magic protocols that lead to unital channels. These are a subset of protocols that leave invariant some distinguished state or have some equilibrium fixed-point structure. We then consider protocols that are non-equilibrium processes, but only generate sub-linear correlations that enable a simple description in the thermodynamic limit before giving entropic analysis that can be applied generally.

The relations between the sets of protocols considered are as follows:
\begin{align}
	\mbox{Unital} &\subset \mbox{ local fixed-points} \nonumber\\ &\subset \mbox{ sub-linear correlations} \subset \mbox{ general protocols.} \nonumber
\end{align}

We provide explicit bounds for magic protocols that generate unital channels (Eq.~(\ref{eq:unital-bound})), as well as bounds that incorporate the temperature and Hamiltonian of the system (Theorem~\ref{thm:free-energy}) and discuss extensions to more general scenarios. These bounds are shown to be tighter than other bounds such as those that come from mana and thauma (Fig.~\ref{fig:distill_bounds}). 

We find that the analysis in the presence of negativity displays a range of features that do not appear in classical statistical mechanics, and leads to a picture of Wigner negativity in a quantum circuit being described as non-classical free energy that is processed under stochastic dynamics (Fig.~\ref{fig:sketch}). This is demonstrated by non-monotonic Lorenz curves (Fig.~\ref{fig:lcs}), and by single-shot entropies for general Wigner quasi-distributions that are well-defined, obey data-processing under stabilizer operations and can take on negative values (Eq.~(\ref{eq:H})).

\begin{figure}[t]
    \centering
        \includegraphics[scale=0.3]{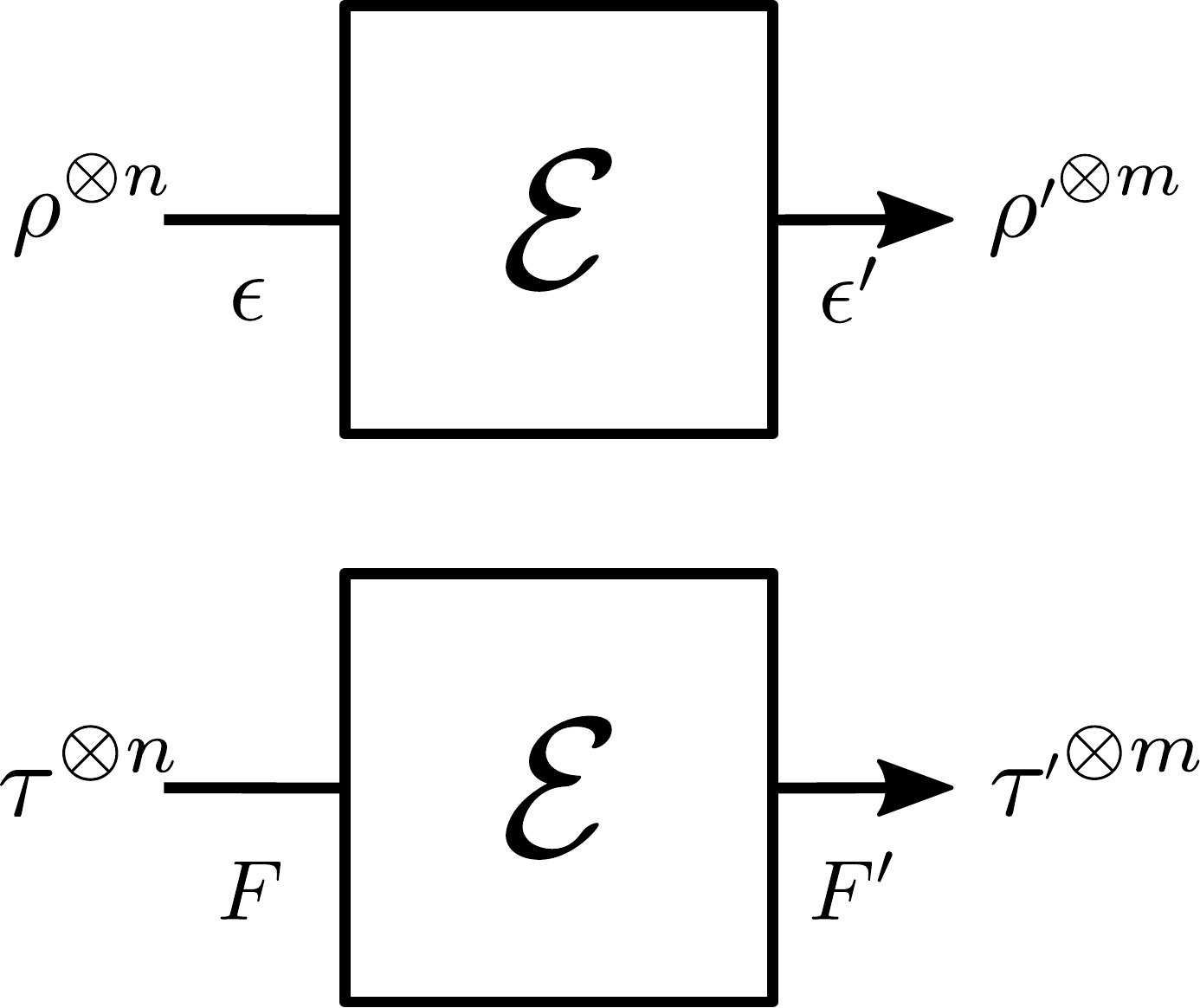}
    \caption{\textbf{Probing magic protocols with reference states.} 
	A magic protocol (top figure) converts $n$ copies of some noisy magic state $\rho$ with error parameter $\epsilon$ to $m$ copies of some less noisy state $\rho'$ with $\epsilon' < \epsilon$. The physics of the protocol can be analysed by considering how it would transform some distinguished reference stabilizer state $\tau^{\otimes n}$ (bottom figure). For the case where this reference state is a thermal state at some temperature $T$ our magic distillation bounds are functions of $T$, the error parameters $\epsilon, \epsilon'$, and free energies $F, F'$. In addition, a quantity $\phi$ (Eq.~(\ref{eq:phi})) relates the computational and energy bases, and corresponds to the degree to which the Wigner negativity of $\rho$ can be given a sharp energy via the system Hamiltonian.}
    \label{fig:sketch}
\end{figure}

We show that the mana of a quantum state arises as a residue in the $\alpha$--parameter for R\'{e}nyi entropies, and gives the rate of divergence to $-\infty$ as we approach the Shannon entropy of a Wigner distribution (Eq.~(\ref{eq:mana-divergence})).
Finally, in the Discussion section, we discuss how our approach could be extended to construct explicit lower bounds on distillation that take into account physical constraints.

\section*{Results}

\subsection*{Phase space representation of quantum mechanics}

Central to our construction is the representation of any quantum state or quantum operation on a system of dimension $d$ in terms of quasi-probability representations on a discrete phase space~\cite{Gross2006, Ferrie_2008}. This construction is a discrete version of Wigner representations in quantum optics~\cite{Wigner_1932, Vourdas_2004, Kenfack_2004}.

Consider a $d$--dimensional quantum system with Hilbert space $\H_d$, and let $\{ |0\>, |1\>, \dots , |d-1\>\}$ denote the standard computational basis, defined over $\mathbb{Z}_d = \{ 0, 1, \dots,d-1 \}$. On this space, generalised Pauli matrices $X, Z$ can be defined by their respective roles as shift and phase operators, acting on the basis states as follows,
\begin{align}
    X \ket{k} &= \ket{k + 1} \label{eq:xpauli}\\
	Z \ket{k} &= \omega^k \ket{k}. \label{eq:zpauli}
\end{align}
Here $\omega \coloneqq e^{2\pi i/d}$ is the $d$-th root of unity and addition is taken modulo $d$. From these we can construct a phase space $\P_d \coloneqq \mathbb{Z}_d \times \mathbb{Z}_d$ that provides a complete representation of the quantum system. Given a point $\z \coloneqq (q, p)$ we define a displacement operator, 
\begin{equation}\label{eq:ddef}
    D_{\z} \coloneqq \tau^{q p} X^{q} Z^{p},\ 
\end{equation}
where the phase factor $\tau \coloneqq -\omega^{1/2}$ ensures unitarity. We assume going forward that $d$ is an odd prime~\cite{Appleby_2005}. The qubit case $d=2$ is known to have obstacles to using a Wigner representation to quantify magic in terms of negativity~\cite{Mermin_1990, cit:howard2}, although recent work~\cite{Raussendorf_qubits} has shown how to extend the representation to include qubits by enlarging the phase space. For a composite system with dimension $d_{\rm tot} = d_1 \dots d_n$ and Hilbert space $\H_{d_1} \otimes \dots \otimes \H_{d_n}$, we define displacement operators as
\begin{equation}\label{eq:composited}
    D_{\z} \coloneqq D_{(q_1, p_1)} \otimes \dots \otimes D_{(q_n, p_n)},
\end{equation}
where now we have
\begin{align*}
	\z &\coloneqq (q_1, \dots, q_n, p_1, \dots, p_n) \in \P_{d_1} \times \dots \times \P_{d_n},
\end{align*}
to denote the phase space point for the composite system. 
For simplicity, we assume $n$ copies of a $d$--dimensional system $d_1=d_2 = \cdots = d$, and therefore we have that $\z \in \mathbb{Z}_d^{2n}$.

The displacement operators form the Heisenberg-Weyl group~\cite{Folland_1989, Bengtsson_2006} under matrix multiplication modulo phases,
\begin{equation}\label{eq:gp}
    {\rm{HW}}_d^n \coloneqq \{ \tau^k D_{\z}: k \in \mathbb{Z}_d, \z \in \P_d^n\}.
\end{equation}
The Clifford group $ {\cal C}_d^n $ is then defined as the set of unitaries that normalise the Heisenberg-Weyl group~\cite{Appleby_2005}. We may define the pure stabilizer states as those states obtained by acting on $|0\>$ with Clifford unitaries~\cite{Gross2006} and, finally, the full set of stabilizer states as the convex hull of all pure stabilizers, namely all probabilistic mixtures of states of the form $U|0\>\<0|U^\dagger$ where $U$ is Clifford. 

\subsection*{Wigner representation for quantum states and quantum operations}

In order to provide a complete decomposition of arbitrary quantum states and quantum operations we now define a basis of Hermitian observables that transforms covariantly under the action of the Clifford group. At every point $\z \in \P_d$ we define the phase-point operator
\begin{equation}\label{eq:ax}
	A_{\z} \coloneqq \frac{1}{d} \sum_{\y \in \P_d} \omega^{\eta(\z, \y)} D_{\y}, 
\end{equation}
where $\eta(\z, \y)$ is the symplectic inner product between any two points $\z,\y \in \P_d$, and is given explicitly by
\begin{equation}
	\eta(\z, \y) \coloneqq \y^T \begin{pmatrix}
		0  & \id \\
		-\id & 0 \\
	\end{pmatrix} \z,
\end{equation}
where $0, \id$ denote the $n\times n$ zero and identity matrices.

The phase-point operators form an orthogonal operator basis with respect to the Hilbert-Schmidt inner product, as discussed in Supplementary Note~1.
Therefore, any quantum state $\rho$ can be expressed as a linear combination of them, $\rho = \sum_{\z \in \P_d} W_\rho(\z) A_{\z}$, where the coefficients $W_\rho(\z)$ give the Wigner distribution of state $\rho$,
\begin{equation}\label{eq:wstate}
    W_\rho(\z) \coloneqq \frac{1}{d}\tr[A_{\z} \rho].
\end{equation}
For any quantum state $\rho$, the Wigner distribution $W_\rho(\z)$ is readily seen to be a $d^2$-dimensional quasi-distribution over $\P_d$. More precisely, $W_\rho(\z)$ is a real-valued function on $\P_d$ with the normalisation property $\sum_{\z} W_\rho(\z) = 1$ (see Supplementary Note~1 for details). Moreover, the above Wigner representation has recently been shown to be the only non-contextual quasi-probability representation of stabilizer theories in odd prime dimensions~\cite{schmid2021noncontextual}.

Any quantum channel $\E$ also admits a Wigner representation. If $\E$ maps some quantum system $A$ of dimension $d_A$ to a quantum system $B$ of dimension $d_B$ and $\J(\E) = (id \otimes \E) (|\phi^+\> \<\phi^+|)$, with $|\phi^+\> = d_A^{-1/2} \sum_k |kk\>$, is its associated Choi state~\cite{watrous_2018}, then we can define
\begin{equation}\label{Wigner-Choi}
	W_{\E}(\y |\z) \coloneqq d_A^2 W_{\J(\E)}(\bar{\z} \oplus \y),
\end{equation}
where $\bar{\z} = (q_1, \dots, q_n, -p_1, \dots, -p_n)$ can be viewed as the time-reversed version of $\z$ in the discrete phase space, where momenta are reflected while position coordinates remain unchanged. 

\subsection*{Magic theories for quantum computation}
A number of magic theories exist, where magic states are computational resource states with respect to a class of quantum operations that are considered free~\cite{Gour_2019}. One natural class of free operations are those obtained from Clifford unitaries, Pauli measurements and the ability to discard quantum systems. However, there are also other candidates~\cite{cit:ahmadi, cit:seddon, Wang_2019}.

In any theory of magic, one route to bounding distillation rates is through a magic monotone. A magic monotone is a real-valued function of any quantum state $M(\rho)$ that is monotonically non-increasing under the free operations of the magic theory. More precisely $M(\sigma) \le M(\rho)$ whenever it is possible to convert $\rho$ into $\sigma$ using free operations.

One prominent magic monotone is the mana of a state~\cite{cit:veitch2}, defined as
\begin{equation}
    \M(\rho) \coloneqq \log (2\hspace{1pt}{\rm sn}(\rho)+1),
\end{equation}
where the sum-negativity~\cite{cit:veitch2} is the sum of the negative components in $W_\rho$,
\begin{equation}
    {\rm sn}(\rho) \coloneqq \sum\limits_{\z: W_\rho(\z) < 0} \abs{W_\rho(\z)}.
\end{equation}
Using the fact that $W_\rho(\z)$ is a normalized quasi-distribution, we can also write $\M(\rho) = \sum_{\z} |W_\rho(\z)|$.

Mana is an additive magic monotone, and the fact that it is non-increasing under free operations provides a constraint on magic state interconvertibility.

In this work we develop bounds that apply to any reasonable magic theory. The central idea is to apply majorization theory to the quasi-distributions that arise for magic states. In the next section we explain how majorization relates to the free operations in a theory of magic.

\subsection*{Stochastic representation of magic protocols}

Within the Wigner representation for odd prime dimension, it is well-known that all positively represented states used in Clifford circuits admit an efficient classical simulation~\cite{cit:mari}, so negativity is a necessary resource for universal fault-tolerant quantum computing~\cite{cit:veitch}. Therefore, the free states in any magic theory are required to be a subset of
\begin{equation}
    \F \coloneqq \{ \rho: W_\rho(\z) \geq 0 \text{ for all } \z \in \P_d\}.
\end{equation}
Our focus is on states with negativity, so the particular choice of free states is not critical for our analysis. The remaining component that defines any magic theory is the set of free quantum operations. The most basic assumption we require on free operations is that they send any free state to another free state.

Any magic state protocol will correspond to a quantum channel $\E$, so from Eq.~(\ref{Wigner-Choi}) it admits a Wigner representation $W_{\E}(\y | \z)$ that acts as a transition matrix mapping phase space points $\z \rightarrow \y$. The representation obeys the relation 
\begin{equation}
	W_{\E(\rho)} (\y) = \sum_{\z \in \P_d} W_{\E}( \y | \z) W_\rho(\z),
\end{equation}
for any $\rho$. Since the magic protocol sends free states to free states, $\E$ sends all positively represented quantum state to other positively represented quantum states. In such cases, it can be shown~\cite{Wang_2019} that if $\E$ is a free operation then the associated Wigner representation $W_{\E}(\y |\z)$ must form a stochastic matrix. In particular, all stabilizer operations correspond to stochastic matrices in the Wigner representation. 

We note however, that not all stochastic maps on the phase space correspond to valid quantum operations. The reason is that the maps must also respect the symplectic structure of the phase space, which is an additional non-trivial constraint.

In what follows, we shall assume that we have a magic theory $\R = (\F, \O)$ in which the free states $\F$ are represented by non-negative Wigner functions, while the free operations $\O$ are stochastic maps in the Wigner representation. Our analysis of magic in this phase space setting makes use of majorization theory, which we describe in the next section.

\subsection*{Quantifying disorder without entropies}

Majorisation~\cite{cit:marshall, Blackwell_1953} is a collection of powerful tools that has found many applications in quantum information theory~\cite{Nielsen_1999, cit:cwiklinski, cit:lostaglio2, cit:gour, cit:gour2, Horodecki_2003, Puchala_2013, Vallejos_2021}.
It describes the disorder of distributions that undergo stochastic transformations, and in its simplest form defines a pre-order on probability distributions. Given two distributions $\p= (p_1, \dots, p_n)$ and $\q = (q_1, \dots, q_n)$ over $n$ outcomes, we say that $\p$ majorizes $\q$, denoted $\p \succ \q$, if there exists a bistochastic map $A = (A_{ij})$ such that $A\p = \q$, where bistochastic means that $A_{ij} \geq 0$ and $\sum_i A_{ij} = \sum_j A_{ij} = 1$. It can be shown~\cite{cit:marshall} that the condition $ \p \succ \q$ over probability distributions is equivalent to $n-1$ inequalities, which can be checked efficiently.

There is a natural generalisation, which is called $\d$--majorization~\cite{Veinott_1971}, or in the context of thermodynamics, thermo-majorization~\cite{cit:horodecki2013}. For a fixed probability distribution $\r = (r_1, \dots, r_n)$ with positive components, we define majorization relative to $\r$ as $\p \succ_{\r} \q$, if and only if there exists a stochastic map $A$ such that $A\r = \r$ and $A \p = \q$. The original majorization condition between probability distributions corresponds to the case $\r = (1/n, \dots, 1/n)$. 

In fact, we can further generalize to relative majorization~\cite{Blackwell_1953, Ruch_1976, ruch_mixing_1978, Renes_2016, Buscemi_2017, Rethinasamy_2020}, defined as an ordering between pairs of vectors and write 
\begin{equation}
	(\p,\r) \succ (\q, \r')
\end{equation}
if and only if there is a stochastic map $A$ such that $A\r = \r'$ and $A\p = \q$. We retrieve $\d$--majorization when $\r = \r'$.

\subsection*{Quasi-probability majorization and non-monotonic Lorenz curves}

Our analysis applies majorization to magic states at the level of the associated Wigner distributions. Since these are in general quasi-distributions, it is important to check how majorization is computed for these cases and what differences quasi-distributions bring over genuine probability distributions.

We also make use of the notion of a Lorenz curve of a vector $\w \in \mathbb{R}^n$ relative to some other vector $\r \in \mathbb{R}^n$. Given a vector $\w$ we define $\w^\downarrow$ to be the re-arrangement of the components of $\w$ into decreasing order. Given two $n$--component vectors $\w$ and $\r$, we first define $\widetilde{\w}=(\widetilde{w}_i)$, where $\widetilde{w}_i \coloneqq w_i/r_i$, as the vector of component-wise ratios between $\w$ and $\r$.
We can now define the Lorenz curve of $\w$ relative to $\r$, denoted $L_{\w|\r}(x)$, as the piece-wise linear function that passes through $(0,0)$ and the $n$ points
\begin{equation}\label{eq:lc}
	(x_k,L_{\w|\r}(x_k)) =\left( \frac{1}{R}\sum_{i=1}^k r_{\pi(i)}, \sum_{i=1}^k w_{\pi(i)} \right),
\end{equation}
where $R\coloneqq \sum_{i=1}^n r_i$ and $\pi$ is the permutation on $n$ objects mapping $\widetilde{\w}$ to $\widetilde{\w}^{\downarrow}$. The form of this requires that $\r$ has no zero components, which we shall assume without loss of generality as the rank of a quantum state is not operationally meaningful.

In the usual case where $\w$ and $\r$ are both probability distributions the Lorenz curve is defined on the interval $[0,1]$, and rises monotonically until it reaches the value $1$ at $x=1$. The value $L_{\w|\r}(x) = 1$ is a global maximum. Moreover, if $\w, \w', \r, \r'$ are all valid probability distributions with $\r,\r'$ having positive components, then
\begin{equation*}
	(\w, \r) \succ (\w', \r') \mbox{ if and only if } L_{\w |\r}(x) \ge L_{\w' |\r'}(x),
\end{equation*}
for all $x \in [0,1]$. 
This result is shown in~\cite{ruch_mixing_1978} and we reproduce it in Supplementary Note~2.
It provides a simple way of computing whether relative majorization holds between pairs of probability distributions.

However, if $\w$ is a quasi-probability distribution with negative values, and $\r$ a regular probability distribution things are different. Now the Lorenz curve is no longer monotonically increasing, but is a concave function that breaks through the $L_{\w|\r}(x) = 1$ barrier at an interior point and attains some non-trivial maximum $L_\star$ above the value $1$, before decreasing monotonically to $L_{\w|\r}(x)= 1$ at the end-point $x=1$. See Fig.~\ref{fig:lcs} for examples of non-monotonic Lorenz curves for quasi-distributions. Within quantum theory, the breaking of this barrier is associated with the degree of non-classicality.

\begin{figure}
    \centering
    \includegraphics[scale=0.35]{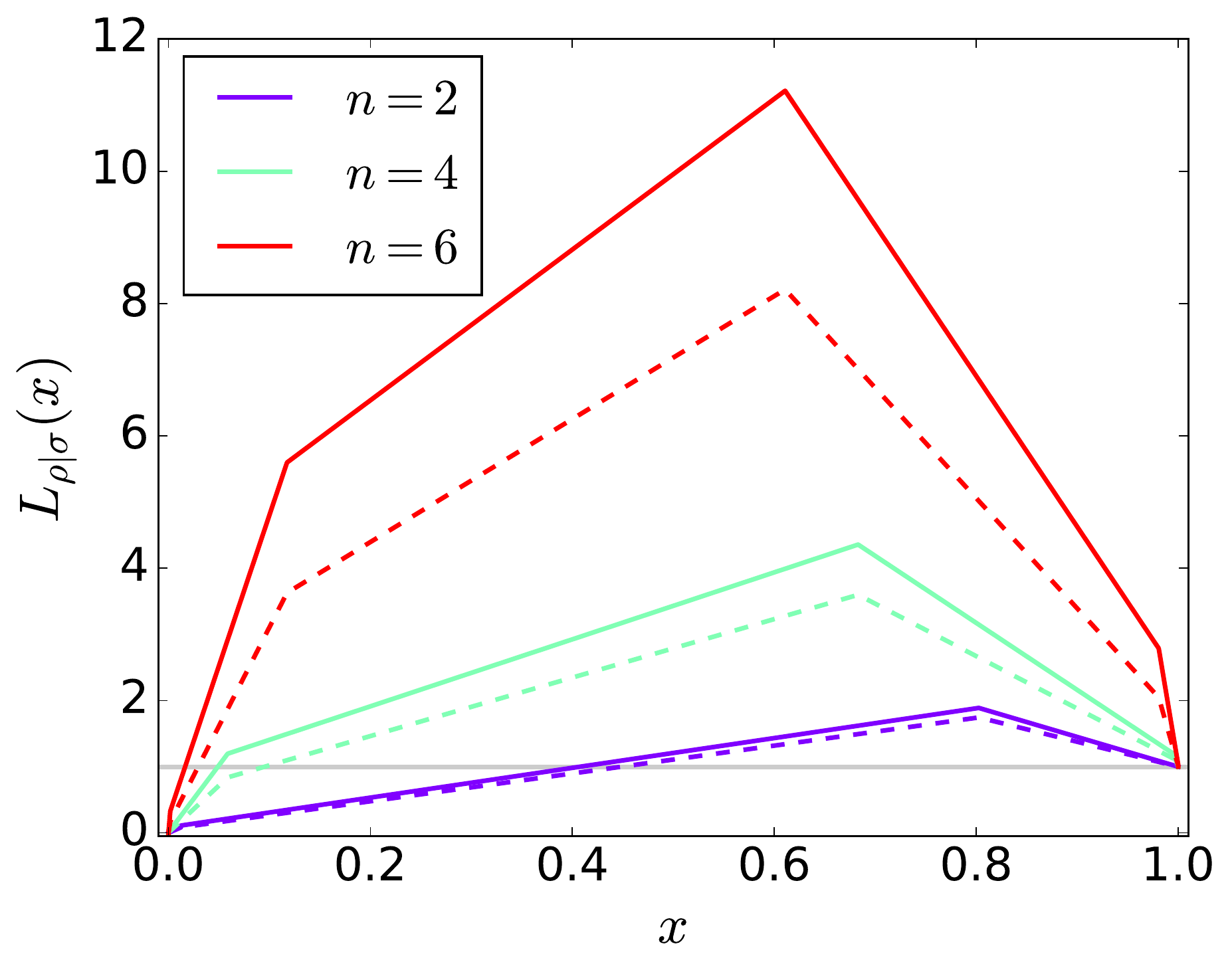}
    \caption{\textbf{A `heretical' family of Lorenz curves.} Traditionally, Lorenz curves for distributions are monotone increasing cumulative functions that reach a maximum value of $1$. In contrast, Lorenz curves for magic states break through the value of $L(x)=1$ due to the presence of negativity in the associated quasi-distributions. The above family of curves correspond to multiple copies of noisy Strange states $\rho\coloneqq\rho_{\rm{S}}(\epsilon)^{\otimes n}$ (Eq.~(\ref{eq:noisysn})) for $n=2,4,6$ within $\R_{\sigma}$, where $\sigma = \id/3$. Solid lines represent pure Strange states, while dashed lines represent $\epsilon$-noisy Strange states with depolarising error $\epsilon = 0.1$.
    }
    \label{fig:lcs}
\end{figure}

Relative majorization is usually considered for probability distributions, so we need to verify that the same Lorenz curve conditions apply to quasi-distributions before proceeding with our analysis. This can be done from first principles, but a simpler way to see that this holds is to reduce to the problem of genuine distributions by first `masking' the negativity in the quasi-distribution $\w$ using the reference distribution $\r$ and then applying the conditions for relative majorization. This negativity masking then gives the following result.

\begin{theorem}\label{thm:lcquasi}
	Let $\w \in \mathbb{R}^n,\w'\in \mathbb{R}^m$ be quasi-distributions, and let $\r \in \mathbb{R}^n, \r' \in \mathbb{R}^m$ be probability distributions with non-zero components. Then, $(\w, \r) \succ (\w', \r')$ if and only if $L_{\w|\r}(x) \ge L_{\w' |\r'}(x)$ for all $x \in [0,1]$.
\end{theorem}
\begin{proof}
	Since the components of $\r$ are strictly positive, there always exists an $\epsilon >0$ such that $\w_\epsilon \coloneqq \epsilon \w + (1-\epsilon) \r$ is a genuine probability distribution. A similar result holds for $\w'$ and $\r'$ and we choose $\epsilon$ sufficiently small so that both $\w_\epsilon$ and $\w_\epsilon'$ are probability distributions. We now have that $(\w, \r) \succ (\w', \r')$ if and only if $(\w_\epsilon , \r) \succ (\w_\epsilon', \r')$. This equivalence holds because there exists a stochastic map $A$ such that $A \w = \w'$ and $A \r = \r'$ if and only if 
\begin{equation}
A[\epsilon \w + (1-\epsilon) \r] = \epsilon \w' + (1-\epsilon) \r'\mbox{ and } A\r = \r'.
\end{equation}
In terms of a Lorenz curve condition we have that $(\w_\epsilon , \r) \succ (\w_\epsilon', \r')$ if and only if $L_{\w_\epsilon |\r} (x) \ge L_{\w'_\epsilon |\r'} (x)$ for all $x \in [0,1]$. 
Additionally, we show in Supplementary Note~2 that the Lorenz curve for any quasi-distribution obeys the relation
\begin{equation}
L_{a \w + b \r | \r} (x) = a L_{\w|\r} (x) + bx,
\end{equation}
for any $a >0$ and $b \in \mathbb{R}$. 
This relation implies that $L_{\w_\epsilon |\r} (x) \ge L_{\w'_\epsilon |\r'} (x)$ if and only if $\epsilon L_{\w |\r} (x) + (1-\epsilon) x \ge \epsilon L_{\w' |\r'} (x) + (1-\epsilon) x$. 
Finally, the $(1-\epsilon)x$ terms cancel on both sides and we get the required relative majorization conditions that $(\w, \r) \succ (\w', \r')$ if and only if $L_{\w | \r} (x) \ge L_{\w' | \r'}(x)$ for all $x \in [0,1]$, as required.
\end{proof}

With this theory in place, we turn to magic distillation protocols.

\subsection*{Magic distillation bounds with fixed-points}

We can now consider how majorization constrains magic distillation rates starting in this section with an abstract setting based on the fixed-point structure of the physical operations accessible to us. We show that any magic theory can always be decomposed in terms of `sub-theories' with a particular fixed-point structure, which allows the use of $\d$--majorization. 
Later, we drop the requirement of fixed-points and instead use the full relative majorization, before giving general entropic constraints.
 
 We begin by defining the following sub-theory of any magic theory $\R$.
\begin{definition}\label{def:sigmafrag}
   Given a theory of magic $\R = (\F, \O)$, we define the sub-theory $\R_\sigma = (\F, \O_\sigma)$, in which the free operations are 
   \begin{equation}
        \O_\sigma \coloneqq \{ \E \in \O: \E(\sigma) = \sigma \},
    \end{equation}
namely those free operations that leave the distinguished state $\sigma \in \F$ invariant.
\end{definition}
This gives a simple way to break up any theory into smaller, more manageable parts, and where the union over all sub-theories returns the parent theory, so we are not discarding any information by breaking up the theory in this way. We make this precise as follows.
\begin{theorem}\label{thm:frag}
    Let $\R = (\F, \O)$ be a theory of magic.
Then a transformation $\rho_1 \rightarrow \rho_2$ is possible in $\R$ if and only if the transformation $\rho_1 \rightarrow \rho_2$ is possible in at least one sub-theory $\R_\sigma$.
\end{theorem}
\noindent The proof of this is straightforward.
\begin{proof}
   Suppose the interconversion is possible in a $\R_\sigma$ via some $\E \in \O_\sigma$. But since $\O_\sigma \subseteq \O$ it is also possible in $\R$. Conversely, suppose the transformation is possible in $\R$ via some $\E$ in $\O$. The free states $\F$ are a closed, bounded set and moreover the image of $\F$ under the map $\E$ is in $\F$. By the Brouwer fixed-point theorem~\cite{cit:brouwer}, this mapping must therefore have a fixed point $\sigma \in \F$, so $\E \in \O_\sigma$ and the interconversion is possible in $\R_\sigma$.
\end{proof}
We can compare with resource monotones, which are also called resource measures. A complete set of measures $\{\M_i\}_i$ is such that $\rho \rightarrow \sigma$ if and only if $\M_i(\rho) \ge \M_i(\sigma)$ for all $i$. The above set of sub-theories can therefore be viewed as a complete set of ``co-measures'' for the theory, where the parent resource theory's complex pre-order of states is mapped to a simpler pre-order for a particular $\R_\sigma$. 
In the next section we show that a given $\R_\sigma$ can be approximated by a majorization pre-order, and from this we compute magic distillation bounds.

\subsection*{Majorisation of Wigner distributions with fixed-point structure}

We consider some $\sigma \in \F$, and its corresponding $\R_\sigma$. We are interested in the ability to transform many copies of some noisy magic state $\rho$ towards a more pure form of magic. The state $\rho$ is assumed to have negativity in the Wigner representation, so $W_\rho(\z) < 0$ for some regions of $\z \in \P_d$. The state $\sigma$ is assumed to have a Wigner distribution with $W_\sigma(\z) > 0$ for all $\z$ in the phase space.
This is justified because a non-full rank state $\sigma$ can be handled as a limiting case in which we first add an infinitesimal fraction of depolarising noise $\epsilon (\id/d)$ and then take $\epsilon \rightarrow 0$.

The free operations within the magic theory $\R$ are represented by stochastic maps, and within $\R_\sigma$ by stochastic maps that leave $W_\sigma(\z)$ invariant. Therefore, a necessary condition for magic state transformations $\rho_1 \rightarrow \rho_2$ within $\R_\sigma$ will be that 
\begin{equation}
	W_{\rho_1} \succ_{W_{\sigma}} W_{\rho_2},
\end{equation}
or put another way, that the quasi-distribution $W_{\rho_1}$ is more ordered than the quasi-distribution $W_{\rho_2}$ relative to $W_\sigma$. To simplify notation, we denote by $L_{\rho | \sigma}(x)$ the Lorenz curve $L_{W_{\rho} | W_{\sigma}} (x)$, and therefore have that
\begin{equation}
\rho_1 \rightarrow \rho_2 \mbox{ within }\R_\sigma \mbox{ implies } L_{\rho_1 |\sigma} (x) \ge L_{\rho_2 |\sigma} (x),
\end{equation}
for all $x \in [0,1]$, which restricts the transformations that are possible.

Note that this is not a single numerical constraint, but a family of constraints. For $n$ copies of a qudit system of dimension $d$ the number of terms in $W_{\rho}$ is $d^{2n}$, so imposing the Lorenz curve condition corresponds to exponentially many constraints.

Before computing explicit examples, we state some generic aspects of Lorenz curves for magic states, which allow us to interpret previous magic monotones as features of the curves. The first result gives a simple way to see that the sum-negativity/mana of a magic state is a monotone~\cite{cit:veitch2}.

\begin{theorem} Majorization in any $\R_\sigma$ implies the monotonicity of sum-negativity/mana. 
\end{theorem}
\begin{proof}
	The sum-negativity of a magic state $\rho$ can be written as $sn(\rho) =\frac{1}{2} (\sum_{\z} |W_\rho(\z) | - 1)$.
We make use of the $L_1$-norm form of relative majorization, (see Supplementary Note~2), which states that $\p \succ_{\r} \q$ if and only if
	\begin{equation}
\sum_k | p_k - r_k t | \geq \sum_k | q_k - r_k t |,
\end{equation}
for all $t\in \mathbb{R}$. Choosing $t=0$ we get the single condition that $\sum_k |p_k| \ge |\sum_k |q_k|$, independent of the choice of $\r$. Applying this to the Wigner quasi-distributions of two quantum states immediately gives the result.
\end{proof}

If we have a magic state $\rho$ that has negativity in its Wigner representation, then, as discussed, its Lorenz curve over-shoots $1$ and reaches a non-trivial maximum $L_\star$ that depends on the particular state. There is an simple relation between $L_\star$ and sum-negativity/mana, which is provided by the following result. 
\begin{theorem}\label{lem:lcmax}
	Given a magic state $\rho$, the maximum $L_\star$ of its Lorenz curve $L_{\rho|\sigma}(x)$ is independent of the $\R_\sigma$ and equal to $1+{\rm sn}(\rho)$. Moreover, the majorization constraint is stronger than mana in every $\R_\sigma$.
\end{theorem}
The proof of this is given in Supplementary Note~2. Using the Lorenz curve perspective, it is also simple to construct magic monotones. For example, within a given $\R_\sigma$, the area above the horizontal line $y = 1$ is a magic monotone.
\begin{theorem}
Given a magic state $\rho$ and a free state $\sigma$, let $\A_\sigma(\rho)$ be the area of the region $\{(x, y): 1 \leq y \leq L_{\rho | \sigma}(x)\}$. Then $\A_\sigma(\rho)$ is a magic monotone for $\R_\sigma$.
\end{theorem}
\begin{proof}
Consider the transformation $\rho_1 \rightarrow \rho_2$ within $\R_\sigma$. We have that $L_{\rho_1|\sigma}(x)$ is never below the curve $L_{\rho_2|\sigma}(x)$, and therefore the region above $L=1$ for $\rho_2$ is a subset of the corresponding region for $\rho_1$. Thus, $\A_\sigma(\rho_1) \geq \A_\sigma(\rho_2)$, so $\A_\sigma$ is a magic monotone in $\R_\sigma$.
\end{proof}
Note though, in contrast to mana, the area monotone is specific to $\R_\sigma$, and its value will vary as we change $\sigma$. Therefore its monotonicity depends on the physics of the fixed point and provides a means to analyse magic distillation only under free operations that leave $\sigma$ invariant. One can also define an area monotone in the more general setting of relative majorization, where one drops the fixed point emphasis.

\subsection*{Magic distillation bounds for unital protocols}

We now construct magic state distillation bounds for protocols that generate unital channels, meaning that $\E(\id/d) = \id/d$.

Our approach works for any odd dimension, but for simplicity we consider qutrit magic states ($d=3$), for which a canonical magic state exists with a simple Wigner distribution. This is the Strange state $|S\>$, which is given by
\begin{equation}
|S\> \coloneqq \frac{1}{\sqrt{2}} (|1\> + |2\>),
\end{equation}
in the computational basis. Its distribution $W_{|S\>\<S|}(\z)$ has a single negative value of $-1/3$ at $\z =(0,0)$ and the positive value $1/6$ at all other points. We define the $\epsilon$--noisy Strange state as
\begin{equation}\label{eq:noisysn}
    \rho_{\rm{S}}(\epsilon)\coloneqq (1 - \epsilon) \ket{\rm{S}}\bra{\rm{S}} + \epsilon \frac{1}{3}\id,
\end{equation}
where $\epsilon$ is the depolarizing error parameter. Any magic state $\rho$ can be processed via Clifford operations~\cite{cit:prakash,cit:prakash2} and put into this canonical form for some $\epsilon \ge 0$.

The Wigner distribution of the single-copy, $\epsilon$--noisy Strange state  is given by
\begin{equation}
	W_{\rho_{\rm{S}}(\epsilon)}(\z) = (1-\epsilon)W_{|\rm{S}\>\<\rm{S}|}(\z) + \epsilon W_{\frac{1}{3}\id}(\z).
\end{equation}
with a single negative component
\begin{equation}
	- v(\epsilon) \coloneqq - \left( \frac{1}{3} -\frac{4}{9}\epsilon \right)
\end{equation} 
located at the origin $\z = \bmo$ and positive components
\begin{equation}
	u(\epsilon) \coloneqq \frac{1}{6} -\frac{1}{18}\epsilon
\end{equation}
at the 8 phase space points $\z \ne \bmo$. We assume that $\epsilon < 3/4$ to ensure the presence of negativity in the Wigner distribution. 

We now consider the problem of purifying $n$ copies of a noisy Strange state $\rho_{\rm{S}}(\epsilon)^{\otimes n}$ into a smaller number of copies $n'$ of a less noisy Strange state $\rho_{\rm{S}}(\epsilon')^{\otimes n'}$, with $\epsilon' < \epsilon$ and $n' \leq n$. In particular, we compute the Lorenz curves for the transformation
\begin{equation}\label{eq:sudist}
	\rho_{\rm{S}}(\epsilon)^{\otimes n} \longrightarrow \rho_{\rm{S}}(\epsilon')^{\otimes n'},
\end{equation}
and use them to bound the magic distillation rate $R(\epsilon, \epsilon') \coloneqq n'/n$.

Due to negativity, the ordering of the components of the rescaled Wigner distribution $W_{\rho_n |\sigma_n}(\z) \coloneqq W_{\rho_n}(\z)/W_{\sigma_n}(\z)$, for $\rho_n = \rho_S(\epsilon)^{\otimes n}$, and $\sigma_n = \sigma^{\otimes n}= (\id/3)^{\otimes n}$, depends on whether $-v(\epsilon)$ is raised to an even or odd power. Treating the Wigner distributions as vectors, the component values $w_i$ and associated multiplicities $m_i$ in the $n$--copy case are found to be 
\begin{align}
	m_i &= 8^{i}\binom{n}{i}, \\
	w_i &= u^{i}(-v)^{n-i}, \label{eq:wigu}
\end{align}
where index $i$ runs through $0, \dots, n$, and we assume for simplicity that $n$ is even and $\epsilon \leq 3/7$, which implies that $v \geq u$.
All details on the analysis are provided in Supplementary Note~3.

The Lorenz curve $L_{\rho_n|\sigma_n}(x)$ reaches a maximum value of
\begin{align}\label{eq:lcsu_max}
	L_\star &\coloneqq L_{\rho_n |\sigma_n} (x_\star) = \sum_{i = 0}^{n/2} m_{2i} w_{2i} \nonumber\\
	&= \frac{1}{2} + \frac{1}{2}\left(\frac{15 - 8\epsilon}{9}\right)^n > 1,
\end{align}
which occurs at $x=x_\star$ given by
\begin{equation}
	x_\star = \frac{1}{2} + \frac{1}{2}\left(\frac{7}{9}\right)^n.
\end{equation}

Distillation bounds can be obtained from any part of the Lorenz curve by imposing the constraint that the Lorenz curve of the input state never dips below the curve of the output. One simple constraint can be obtained by considering the initial slopes of the two Lorenz curves. 

The coordinates of the first point after the origin for $\rho_S(\epsilon)^{\otimes n}$ are given by $(1/9^n, v(\epsilon)^n)$, while for $\rho_S(\epsilon')^{\otimes n'}$ by $(1/9^{n'}, v(\epsilon')^{n'})$, and by requiring that the initial slope for the input state Lorenz curve is larger than that of the output state (Proposition~\ref{prop:first_elb} in Supplementary Note~2) we find that
\begin{equation}\label{eq:unital-bound}
	R(\epsilon, \epsilon') \leq R_\infty \coloneqq\frac{\log (3-4\epsilon)}{\log (3-4\epsilon')}.
\end{equation}
The choice of denoting the bound by $R_\infty$ is explained in the derivation of Eq.~(\ref{eq:Ra}) in relation to single-shot entropies. For the limiting case of pure magic states on the output ($\epsilon'=0$), this simplifies to
\begin{equation}
	R \leq 1 + \frac{\log (1 - \frac{4}{3} \epsilon)}{\log 3}.
\end{equation}
We compare this to other known distillation bounds coming from mana and thauma. The mana bound can be directly calculated as
\begin{equation}
	R \leq \frac{\M(\rho_{\rm{S}}(\epsilon))}{\M(\rho_{\rm{S}}(0))} = 1 + \frac{\log \left(1 - \frac{8}{15}\epsilon \right)}{\log \frac{5}{3}}.
\end{equation}
The max--thauma~\cite{Wang_2020} is defined as
\begin{equation}
	\theta_{\rm{max}}(\rho) \coloneqq \log \min{\{2\hspace{1pt}{\rm sn}(V)+1 : V \geq \rho\}}
\end{equation}
and can be calculated numerically via a semi-definite program. For the noisy Strange state, the max--thauma bound coincides with the mana bound, and we find that they are both looser than the majorization bound $R_\infty$ as shown in Fig.~\ref{fig:distill_bounds}. 

This figure includes numerical estimates of the optimal majorization bounds $R_{\rm{num}}$ obtained by considering the full Lorenz curve, which show that the bound $R_\infty$ can be further improved. While this is an improvement on prior results,  all known distillation protocols have rates much lower than these upper bounds.  It remains a major open question to determine what are the best possible rates that can be achieved.

\subsection*{Temperature-dependent bounds for magic distillation protocols}

We now show how relative majorization can incorporate the system's Hamiltonian into magic distillation bounds. The way we do this is by considering how the protocol would disturb a reference equilibrium state $\tau$ to some different state $\tau'$, if the protocol had been applied to this reference state instead of the actual $n$-copy magic state. 

We consider a magic distillation protocol on multiple identical qudits in a noisy magic state $\rho(\epsilon)$, with noise parameter $\epsilon$, sending 
\begin{equation}
\rho(\epsilon)^{\otimes n} \longrightarrow \E(\rho(\epsilon)^{\otimes n}) =\rho(\epsilon')^{\otimes m}
\end{equation}
with $n \gg 1$, $\epsilon' <\epsilon$ and $\E$ denoting the quantum channel induced by the protocol. We also assume each qudit has a Hamiltonian $H$ and neglect interaction terms. We choose some temperature $T = (k\beta)^{-1}$ where $k$ is Boltzmann's constant and $\beta$ the inverse temperature. The reference equilibrium state of the $n$ qudits is given by
\begin{equation}
\tau^{\otimes n} = \left ( \frac{e^{-\beta H}}{\Z} \right )  ^{\otimes n}.
\end{equation}
We may also assume that the reference state $\tau$ is not a magic state, and has a strictly positive Wigner distribution $W_\tau(\z)$.

A given magic protocol on the $n$ qudits will correspond to a quantum channel $\E$. We also assume for simplicity that $U_\pi\E(X) U_\pi^\dagger = \E(X)$ for all $X$ and any permutation $U_\pi$ of the $m$ output subsystems. This is justified because the protocol on the input magic state results in state $\rho(\epsilon')^{\otimes m}$, which is invariant under permutations. Therefore, we are always allowed to symmetrize the output $\E(X)$ by performing a group average over the permutation group for the output systems without changing the performance of the distillation protocol on $\rho(\epsilon)^{\otimes n}$. Thus, we can assume that $\E$ always outputs a symmetric state in general. This means that $\E(\tau^{\otimes n})$ is a symmetric state on $m$ subsystems, so by the quantum de Finetti theorem~\cite{hudson_locally_1976, christandl_2007} we have that for $m \gg 1$
\begin{equation}
\E(\tau) \approx \int d\mu(x) \tau_x^{\prime \otimes m},
\end{equation}
where $d\mu(x)$ is a probability measure over a set $\{\tau'_x\}$ of single qudit states.

\subsection*{Free energy and sub-linear correlations in the thermodynamic limit}

To keep our analysis simple we make the following physical assumption. We assume that in the asymptotic/thermodynamic limit $n,m \rightarrow \infty$ the correlations generated on the reference equilibrium state are negligible. This implies that $d\mu(x)$ is peaked on a particular state $\tau'$, and $\E(\tau^{\otimes n}) \approx \tau'^{\otimes m}$. This scenario occurs in the context of traditional thermodynamics, and states that the output system is well-described by intrinsic variables that do not scale in $m$, and correlations are sub-linear in $m$. In particular, this allows us to compute a free energy per qudit of the output state. 

Non-trivial correlations in the thermodynamic limit can also be considered, but leads to a more complex analysis within our majorization framework of Wigner distributions. In this direction, we highlight recent work in majorization in which stochastic independence and correlations are analysed. It has been shown~\cite{muller_2015} that stochastic independence (no correlations) of independent distributions can be viewed as a resource in an extension of catalytic majorization, and leads to a single-shot operational interpretation of the Shannon entropy~\cite{muller_2016, muller_2019}.

Our bound depends on the von Neumann entropy $S(\rho)$ of a state, and the free energy at a particular temperature.
For a state $\tau = e^{-\beta H}/\Z$, the Helmholtz free energy $F$ is given by
\begin{equation}
	F \coloneqq \tr[H \tau] - \beta^{-1}S(\tau) = -\beta^{-1}\log \Z,
\end{equation}
which is obtained from the internal energy via a Legendre transform~\cite{Pathria_1997}.

The protocol transforms the equilibrium as $\tau^{\otimes n} \rightarrow \E(\tau^{\otimes n}) = \tau^{\prime \otimes m}$.  The protocol does not generate magic on its own, so we assume that the output state $\tau'$ is also a Wigner-positive state. However, this is generally a non-equilibrium state for the system. Despite this it is useful to associate an effective Hamiltonian $H'$ to the output state by considering the change $H \rightarrow H'$ such that equilibrium is restored at the reference temperature $T$. This Hamiltonian is defined by the expression $\tau' = e^{-\beta H'}/\Z'$, and has free energy $F' = -\beta^{-1} \log \Z'$.

The magic state protocol is now considered by how it transforms the pair of quantum states $(\rho^{\otimes n}, \tau^{\otimes n})$ and is then constrained by the relative majorization condition $(W_{\rho}, W_\tau) \succ (W_{\rho'}, W_{\tau'})$ that holds due to the protocol being a stochastic map in the Wigner representation.

In Supplementary Note~4 we analyse the Lorenz curves of this condition and obtain the following magic distillation bound that combines computational measures $\epsilon,\epsilon'$, with terms that depend on the Hamiltonian and reference temperature $T$ of the physical system. We state the result for the case of qutrits, but a similar analysis works for general odd-dimension qudits.

\begin{theorem}\label{thm:free-energy}
	Consider a magic distillation protocol on qutrits that transforms $n$ copies of an $\epsilon$--noisy Strange state into $m$ copies of an $\epsilon'$--noisy Strange state, with depolarising errors $\epsilon' \leq \epsilon \leq 3/7$. We also allow pre/post-processing by local Clifford unitaries.
	
	Let $T =(k\beta)^{-1}$ be any finite temperature for the physical system and let $H= \sum_{k \in \mathbb{Z}_3} E_k |E_k\>\<E_k|$ be the Hamiltonian of each qutrit subsystem in its eigen-decomposition.
Assume that in the thermodynamic limit ($n,m \gg 1$), the protocol applied to the equilibrium state $\tau^{\otimes n} = (e^{-\beta H}/\Z)^{\otimes n}$ maps $\tau^{\otimes n} \longrightarrow \tau^{\prime \otimes m}$, where we write $\tau' = e^{-\beta H'}/\Z'$ for some Hermitian $H'$.

Then the asymptotic magic distillation rate $R = m/n$ is bounded as
\begin{equation}\label{eq:rate_bounds_proof}
	R \leq \dfrac{\log \big( 1-\frac{4}{3}\epsilon \big) + \beta (\phi - F)}{\log \big( 1-\frac{4}{3}\epsilon' \big) + \beta (\phi' - F')},
\end{equation}
where $F$ is the free energy of $\tau$,  and 
\begin{equation}\label{eq:phi}
	\phi = -\beta^{-1} \log \zeta
\end{equation}
with $\zeta$ given by the expressions
\begin{align}
	\zeta &= \sum_{k\in \mathbb{Z}_3} \alpha_k e^{-\beta E_k}, \\
	\alpha_k &= \<E_k| A_{\z_\star} |E_k\>,
\end{align}
and $W_\tau(\z)$ attaining a minimum at $\z=\z_\star$. The primed variables are defined similarly for the output system.
\end{theorem}
\noindent The proof of this is provided in Supplementary Note~4, and follows a similar line to the unital protocol bound.

The above bound depends on: 
\begin{itemize}
	\item Quantum computational measures $\epsilon, \epsilon'$,
	\item Thermodynamic quantities $\beta, F, F'$,
	\item Intermediate terms $\phi, \phi'$. 
\end{itemize}
The intermediate terms specify how the energy eigenbasis of the system $\{|E_k\>\}$ relates to the computational stabilizer basis $\{|k\>\}$.  In particular, the term $\phi$ quantifies to what degree a sharp energy value can be associated to the negativity in the Wigner representation. Its form is similar to $F$, and so $\phi$ can loosely speaking be viewed as a `magic free energy' term. 

The coefficients $\alpha_k$ may be negative for some $k$, when the Hamiltonian has non-stabilizer eigenstates, but the quantity $\phi$ is always well-defined, since the function $\zeta$ is always positive for $\tau$ in the interior of $\F$. The quantity $\phi$ can diverge if the state $\tau$ acquired zero Wigner components, which occurs for $\tau$ on the boundary of the set of Wigner-positive states, and is not defined outside of $\F$. Finally, if $H$ has a stabilizer eigenbasis then $\phi = E_i$ for some $i$, independent of the temperature $T$. Details of this are provided in Supplementary Note~4.

\subsection*{Technical aspects and special cases of free energy bounds}

Similar to the unital protocol bounds, the above result is based on only part of the Lorenz curves and so can certainly be tightened with further analysis. The primary role of the pre/post-processing by Clifford unitaries in Theorem~\ref{thm:free-energy} is to simplify the form of the bound, and allow it to be expressed in terms of the free energy per particle in a way that does not depend on the parameters $\epsilon, \epsilon', \beta$ in a complex form. In Supplementary Note~4, we give bounds in which one does not include these Clifford changes of basis, thus having a more non-trivial dependency on the parameters.

\begin{figure}[t!]
    \centering
    \includegraphics[scale=0.39]{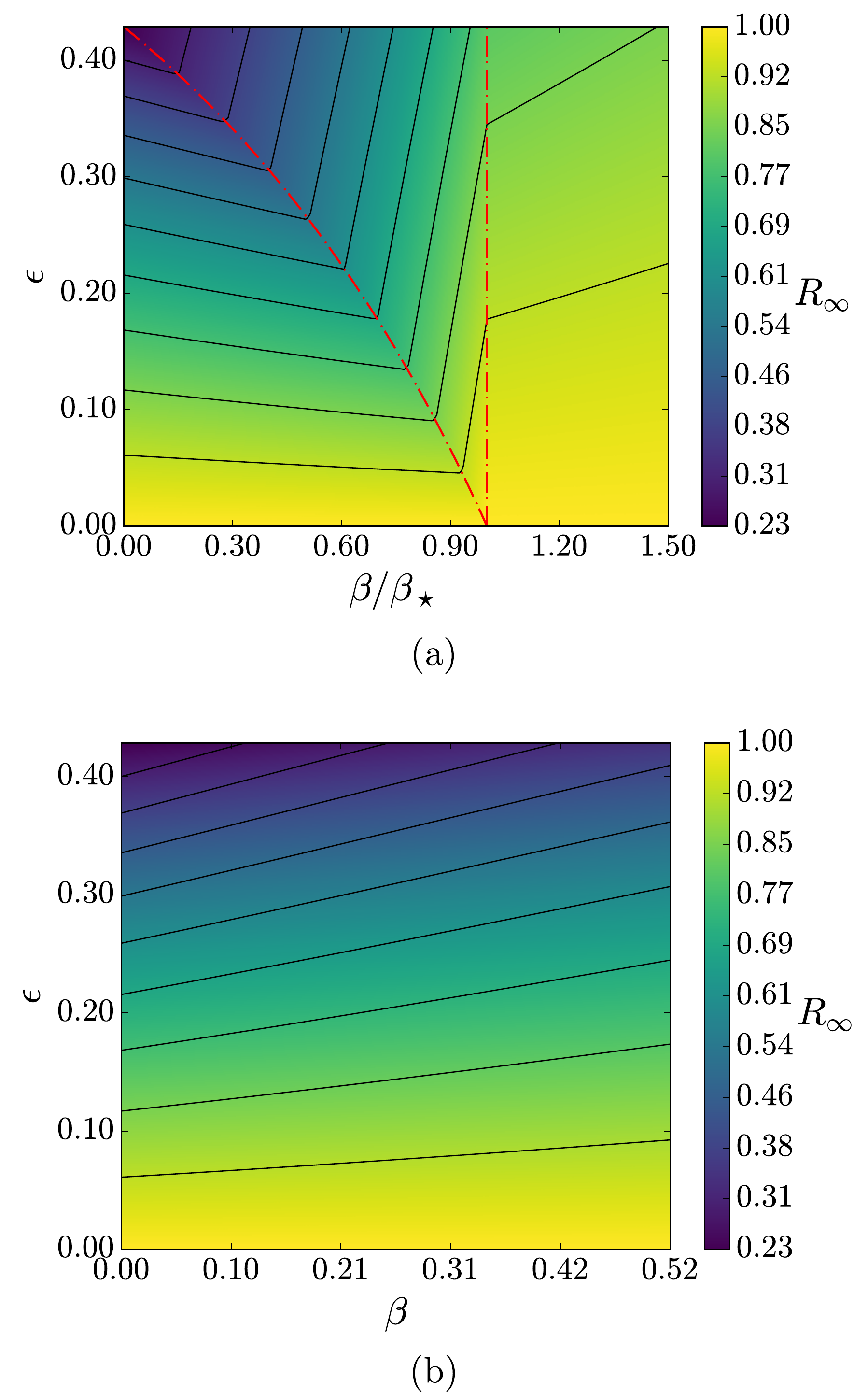}
    \caption{\textbf{Temperature-dependent bounds for magic distillation.} Shown are two contour plots of the bound on $R(\epsilon, \beta)$ for the case $H = H' = \sum_{k \in \mathbb{Z}_3} k \ketbra{k}$ and $\epsilon' = 0$, where $\beta$ is the inverse temperature and $\epsilon$ is the depolarising error of the input magic state. 
The top figure~(a) does not use any changes of Clifford basis, and the form of the bound depends on both the error parameter and temperature. The curved dashed line is $\epsilon_\star(\beta)$ and given by Eq.~(86) in Supplementary Note~4 and $\beta_\star = (kT_\star)^{-1}$ is given by $E_2-\phi = kT_\star \ln 2$. 
In the bottom figure~(b) Clifford processing is used resulting in a smoother bound. In both figures the $\beta = 0$ line correspond to the unital bounds.}
    \label{fig:rate_contour}
\end{figure}

The analysis makes other simplifying assumptions that could easily be dropped, at the price of more complex expressions. We could perform similar analysis for general qudits, and different choices of magic states, for example. It might also be of interest to consider other choices of reference states $\tau$ that are more appropriate to the hardware physics, for example for photonic set-ups~\cite{bombin2021interleaving}.

The assumption which is non-trivial is that we neglect correlations in the reference state in the thermodynamic limit. However, for more general scenarios one could make use of variational tools such as the Bogoliubov inequality~\cite{bogolyubov_1966} for approximating the free energy of a system via product states, to obtain similar bounds.

The simplest special case to consider is where the Hamiltonian $H$ is diagonal in the computational basis, and $H=H'$, implying that the protocol leaves the reference equilibrium state unchanged, hence it corresponds to a Gibbs-preserving map~\cite{faist_2015}. For the limiting case of $\epsilon' \rightarrow 0$ we obtain Fig.~\ref{fig:rate_contour} which is a contour plot of the bound on $R$ as a function of inverse temperature $\beta$ and the depolarising error $\epsilon$ for the noisy input magic states. In this figure we show both the bounds with Clifford basis changes and without them. In the more general case, $\Delta F \coloneqq F' - F \ne 0$, so the protocol, when applied to the reference state, adds/extracts free energy from the system.
We demonstrate change of the bounds with respect to a varying output Hamiltonian in Fig.~\ref{fig:rvsa}.

\begin{figure}
    \centering
    \includegraphics[scale=0.35]{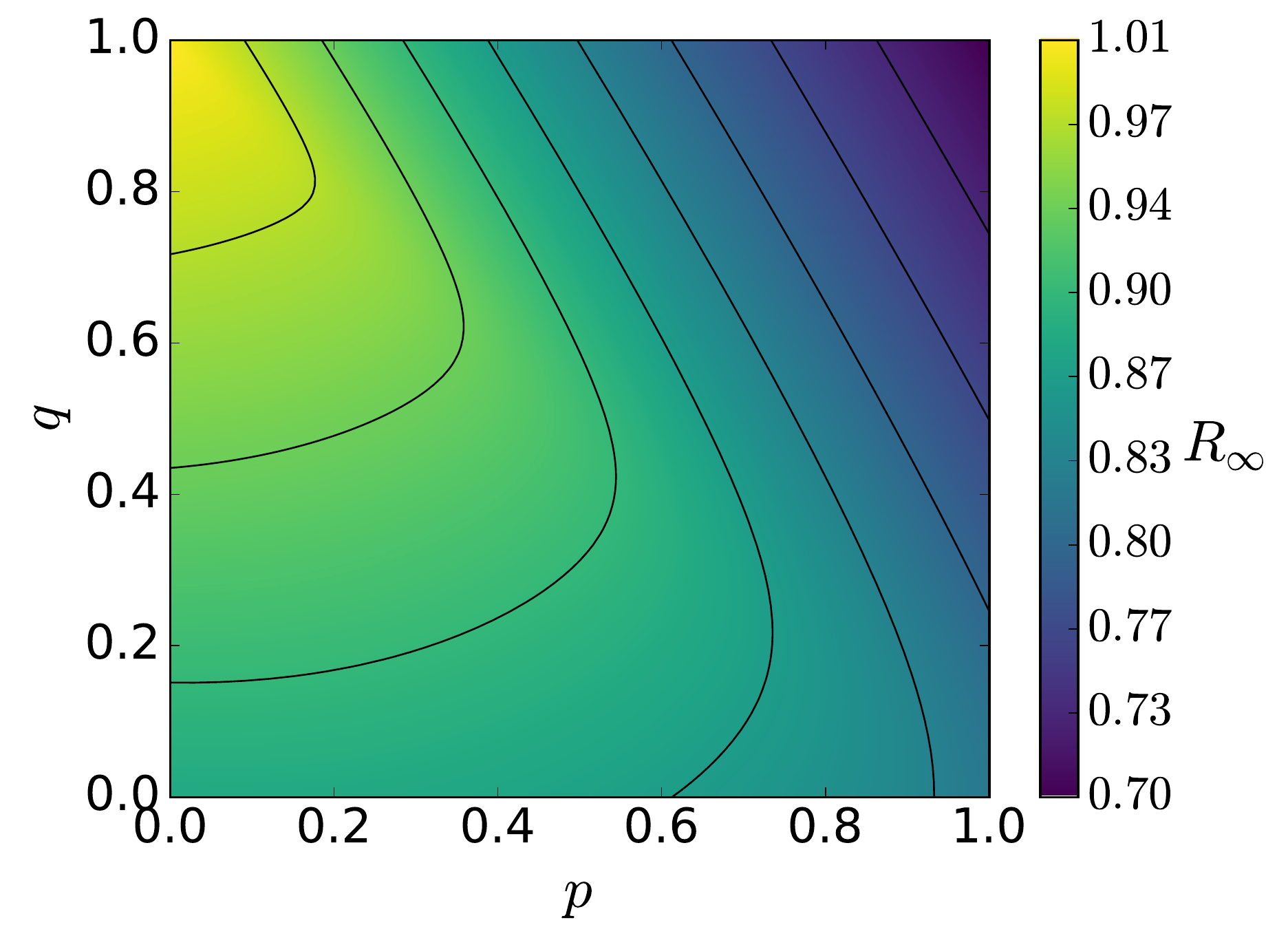}
    \caption{\textbf{Variation of distillation bounds with Hamiltonian.}  We illustrate variation of the bound in Eq.~(\ref{eq:rate_bounds_proof}) with respect to the system Hamiltonian. We fix $(\epsilon, \epsilon', \beta) = (0.1, 0.0, 0.2)$ and look at a family of Hamiltonians for the qutrit systems, given by $H = A_{\bmo}$ and $H' = (1-p-q)A_{\bmo}+pA_{(1,2)}+q{\rm{diag}}(0,1,2)$, with varying parameters $p, q$, and where $A_{\bmo}$ and $A_{(1,2)}$ are phase-point operators. 
    }
    \label{fig:rvsa}
\end{figure}

\subsection*{Extension of single-shot entropies to quasi-distributions}

We have considered Wigner distributions of magic states within a statistical mechanical setting in which it was argued that magic can be viewed as a non-classical form of free energy. This was done at the level of majorization constraints that appear due to stabilizer operations being described by stochastic maps in the Wigner representation. In this representation, non-equilibrium states with substantial free energy have Lorenz curves that deviate noticeably from the line $y=x$. We also found that when considered in a distillation protocol, magic could also be linked to physical free energies in a non-trivial way. We can therefore ask: is it possible to make a more direct link between magic and free energy or entropy?

One obstacle to linking with macroscopic, equilibrium free energy is that the Wigner distribution is generally a quasi-distribution, so the entropy $H(\w) = -\sum_i w_i \log w_i$ of the Wigner distribution is not defined outside the set $\F$, let alone have a physically meaningful interpretation.
 
However, it turns out that the analysis for the temperature-dependent bound given by Eq.~(\ref{eq:rate_bounds_proof}) implicitly made use of the single-shot R\'{e}nyi divergence~\cite{renyi_1960} $D_\infty (\p \hspace{1pt}||\hspace{1pt} \r)$, which is given by
\begin{equation}
	D_\infty (\p \hspace{1pt}||\hspace{1pt} \r) = \log \max_i \frac{p_i}{r_i}.
\end{equation}
It is clear that this does remain well-defined on quasi-distributions, and from the majorization relation also has a clear physical meaning, and in Supplementary Note~4, we show that this divergence has the following properties that are derived via relative majorization.
\begin{theorem}\label{thm:Dinfty}
	Let $\tau$ be in the interior of $\F$. Then $D_\infty(W_\rho || W_\tau)$ is well-defined for all $\rho$, and the following hold:
\begin{enumerate}
\item $D_\infty(W_\rho \hspace{1pt}||\hspace{1pt} W_\tau) \ge 0$ for all quantum states $\rho$.
\item  $D_\infty(W_\rho \hspace{1pt}||\hspace{1pt} W_\tau) = 0$ if and only if $\rho =\tau$.
\item $D_\infty(W_{\rho^{\otimes 2n}} \hspace{1pt}||\hspace{1pt} W_{\tau^{\otimes 2n}}) = n D_\infty(W_{\rho^{\otimes 2}} \hspace{1pt}||\hspace{1pt} W_{\tau^{\otimes 2}})$ for any $n \in \mathbb{N}$.
\item $D_\infty(W_\rho \hspace{1pt}||\hspace{1pt} W_\tau) \geq D_\infty(W_{\E(\rho)} \hspace{1pt}||\hspace{1pt} W_{\E(\tau)})$ for any free operation $\E$ such that $\E(\tau)$ is in the interior of $\F$.
\end{enumerate}
\end{theorem}
The reason for the power of $2$ in the third property is that the maximization is sensitive to the presence of negativity in the Wigner distribution.

We can go further, and show that in fact a range of R\'{e}nyi entropies remain both well-defined and meaningful in the Wigner representation. The $\alpha$--R\'{e}nyi entropy $H_\alpha(\p)$ is defined as~\cite{renyi_1960}
\begin{equation}\label{eq:H}
	H_\alpha(\p) \coloneqq \frac{1}{1-\alpha} \log \sum_i p_i^\alpha.
\end{equation}
To extend to quasi-distributions we must ensure the following: (a) the entropy is mathematically well-defined on quasi-distributions, and (b) obeys monotonicity under the majorization pre-order. In Supplementary Note~5 we prove the following.
\begin{theorem}\label{thm:HSchur} 
	If $\alpha = \frac{2a}{2b-1}$ for positive integers $a,b$ with $a \geq b$, then $H_\alpha(\w)$ is well-defined on the set of quasi-distributions, and if $\w \succ \w'$ for two quasi-distributions $\w, \w'$ then $H_\alpha (\w) \leq H_\alpha(\w')$.
\end{theorem}
This monotonicity provides a meaningful extension of entropies to quasi-distributions. The unusual form of the $\alpha$ parameter is required so that the entropy function is well-defined and real-valued for all $\w$. The allowed values of $\alpha$ are dense in $\alpha >1$, however, while R\'{e}nyi entropies on probabilities can also be used for $0 < \alpha <1$, it turns out that this parameter regime is no longer monotonic in terms of the majorization pre-order for quasi-distributions (see Supplementary Note~5 for details).

We can prove the following result, which establishes the equivalence between Wigner negativity and the existence of negative R\'{e}nyi entropy.
\begin{theorem}\label{thm:Magic}
	A quantum state $\rho$ has Wigner negativity if and only if $H_\alpha(W_\rho) <0$ for some $\alpha = \frac{2a}{2b-1}$, with positive integers $a \ge b$.
\end{theorem}
Therefore, the statistical mechanical description of stabilizer operations provides a setting in which negative entropy is fully meaningful, and quantifies the fact that magic states are more ordered than a perfectly sharp, deterministic classical state $\w = (1,0, \dots , 0)$ with zero entropy. We note that prior work has shown that negative conditional entropy~\cite{rio_thermodynamic_2011} arises in the context of quantum correlations, but is unrelated to the present negative entropy.

The $\alpha \rightarrow 1$ limit of the R\'{e}nyi entropy diverges if negativity is present, however it diverges in a well-defined way. We can write the entropy as 
\begin{equation}
	H_{1+\epsilon}(W_\rho) = -\epsilon^{-1} \log \sum_{\z} |W_\rho (\z)|^{1+\epsilon}
\end{equation} 
with $\epsilon \rightarrow 0^+$ through a sequence of rational values. We recall that mana can be written as $\M(\rho) = \log \sum_{\z} |W_\rho(\z)|$, so we obtain a second perspective on mana:
\begin{equation}\label{eq:mana-divergence}
	\M(\rho) = -\lim_{\epsilon \rightarrow 0^+} \epsilon H_{1+\epsilon}(W_\rho).
\end{equation}
Therefore, the mana of a state $\rho$ is minus the residue of the pole at $\alpha=1$ for $H_\alpha(W_\rho)$, and quantifies the divergence of the R\'{e}nyi-entropy as we approach the limiting Shannon entropy $H(W_\rho)$. 

We also note that the R\'{e}nyi entropy can be described as a $q$--deformation of the Shannon entropy~\cite{baez2011renyi}, and that $H_\alpha(\p)$ can be related to the $q$--derivative of free energy.  For this, the R\'{e}nyi parameter becomes a temperature term $\alpha = \frac{T_0}{T}$ where $T_0$ is a reference temperature and derivatives are considered via the limit $T \rightarrow T_0$. Therefore, if we try to push this to the present setting, the presence of Wigner negativity would correspond to a divergence in the first derivative of an effective free energy as $T \rightarrow T_0$. It is interesting to speculate whether this non-classicality could be interpreted in terms of a phase transition~\cite{domb2000phase}.

\subsection*{R\'{e}nyi divergences and general distillation bounds}
Every entropy can be obtained from a relative entropy~\cite{Gour_2020} so we now state the extensions for R\'{e}nyi divergences of Wigner distributions.
The $\alpha$--R\'{e}nyi divergence $D_\alpha(\p||\r)$ is defined as
\begin{equation}\label{eq:D}
	D_\alpha(\p \hspace{1pt}||\hspace{1pt} \r) \coloneqq \frac{1}{\alpha-1} \log \sum_i p_i^\alpha r_i^{1-\alpha},
\end{equation}
where $\p$ can be a quasi-distribution and $\r$ a probability distribution with positive components. In Supplementary Note~5, we also show that the divergence satisfies the following.

\begin{theorem}\label{thm:Da_props} 
	Let $\tau$ be in the interior of $\F$. 
	If $\alpha = \frac{2a}{2b-1}$ for positive integers $a,b$ with $a \geq b$, then the $\alpha$-R\'{e}nyi divergence $D_\alpha(W_\rho || W_\tau)$ is well-defined for all states $\rho$, and the following hold:
\begin{enumerate}
	\item $D_\alpha(W_\rho \hspace{1pt}||\hspace{1pt} W_\tau) \ge 0$ for all quantum states $\rho$.
	\item  $D_\alpha(W_\rho \hspace{1pt}||\hspace{1pt} W_\tau) = 0$ if and only if $\rho =\tau$.
	\item $D_\alpha(W_{\rho^{\otimes n}} \hspace{1pt}||\hspace{1pt} W_{\tau^{\otimes n}}) = n D_\alpha(W_\rho \hspace{1pt}||\hspace{1pt} W_\tau)$ for any $n \in \mathbb{N}$.
	\item $D_\alpha(W_\rho \hspace{1pt}||\hspace{1pt} W_\tau) \geq D_\alpha (W_{\E(\rho)} \hspace{1pt}||\hspace{1pt} W_{\E(\tau)})$ for any free operation $\E$ such that $\E(\tau)$ is in the interior of $\F$.
\end{enumerate}
\end{theorem}

We can now apply this to obtain a family of magic distillation bounds based on R\'{e}nyi divergences.

\begin{theorem}\label{thm:entropic_bounds}
	Consider a general magic state distillation protocol on odd prime dimension qudits, that converts a magic state $\rho^{\otimes n} \longrightarrow \E(\rho^{\otimes n})=\rho'^{\otimes m}$ and let $\tau$ be any full-rank stabilizer reference state on a qudit. Then, the distillation rate $R \coloneqq m/n$ is upper bounded as
	\begin{equation}
		R \leq \frac{D_{\alpha}(W_\rho \hspace{1pt}||\hspace{1pt} W_\tau)}{\rlap{\raisebox{-0.2ex}{$\widetilde{\phantom{D}}$}}D_\alpha( \rho', \tau')},
	\end{equation}
	where $\alpha = \frac{2a}{2b-1}$ for any positive integers $a,b$ with $a \geq b$ and the average divergence per qudit
	\begin{equation}
\widetilde{D}_\alpha( \rho', \tau') \coloneqq \frac{1}{m} D_\alpha (W_{\rho'^{\otimes m}} \hspace{1pt}||\hspace{1pt} W_{\tau'_m}),
\end{equation}
between the output magic state $\rho'^{\otimes m}$ and $\tau'_m = \E(\tau^{\otimes n})$.
\end{theorem}
The proof is given in Supplementary Note~5. 

While this bound is abstract in its present form, it does suggest future applications by viewing $\widetilde{D}_\alpha( \rho', \tau')$ within the context of hypothesis testing~\cite{tomamichel2013}, as a measure of distinguishability between $\rho'^{\otimes m}$ and the $m$ qudit state $\tau'_m$. 
However, in order to properly link with hypothesis-testing one would have to first extend such results to quasi-distributions.

If correlations in $\tau'_m$ between subsystems can be neglected, then $\tau'_m = \tau'^{\otimes m}$ for some qudit state $\tau'$ and
\begin{equation}
	\widetilde{D}_\alpha( \rho', \tau') = D_\alpha (W_{\rho'} \hspace{1pt}||\hspace{1pt} W_{\tau'}),
\end{equation}
which gives a generalized version of Theorem~\ref{thm:free-energy}, with a term such as $D_\alpha(W_\rho ||W_\tau)$ that behaves essentially as an $\alpha$--free energy difference~\cite{Brandao_2015} in the Wigner representation.

We can illustrate a basic application of the above bound for magic protocols that give unital maps. In this case it is easy to see that the above bound reduces to the simple form
\begin{equation}\label{eq:Ra}
	R \leq R_\alpha \coloneqq \frac{2\log d - H_{\alpha}(W_\rho)}{2\log d - H_{\alpha}(W_{\rho'})}.
\end{equation}
In particular for the case of noisy Strange states on qutrits and $\epsilon'=0$, this bound becomes

\begin{equation}
	R \leq R_\alpha = \frac{2(1-\alpha)\log 3 - \log \left [8(\frac{1}{6} - \frac{\epsilon}{18})^\alpha + (-\frac{1}{3} + \frac{4\epsilon}{9})^\alpha\right ] }{2(1-\alpha)\log 3 + \alpha \log 6 - \log (8 + 2^\alpha )}.
\end{equation}
This upper bound can be numerically minimized over $\alpha$, for any $\epsilon$. The case $R_{10}$ is shown in Fig.~\ref{fig:distill_bounds} and numerical evidence shows it is close to the optimal bound from majorization.

\section*{Discussion}

We have described how relative majorization can be used to establish upper bounds on magic distillation protocols that take into account additional physics of the system. Our bounds exploited relatively simple aspects of the Lorenz curves of the quasi-distributions, so it would be of interest to sharpen these bounds and obtain a better handle on the Lorenz curve structure in the $n\rightarrow \infty$ limit. It would also be of interest to analyse what features of single-shot entropies can be extended to quasi-distributions in a sensible form~\cite{renner_phd, tomamichel2013}. This raises interesting questions as we no longer have a notion of typicality and the central limit theorem does not apply. That said, for special states such as the Strange state, the asymptotic behaviour is relatively simple, so exact asymptotics for this should be possible. In Supplementary Notes~2,~3 and~4 we give additional results that may be of use for later work.

Because Theorem \ref{thm:frag}  provides a completeness statement for majorization, we can ask if this approach could additionally provide lower bounds for distillation protocols. In contrast to the upper bounds, it is now essential to include the symplectic constraints on the phase space in the majorization relations. In principle the resultant majorization (and relative majorization) constraints should provide exact specifications on what can be obtained via free operations. 
Interestingly, this route does not require the use of error correcting codes, but instead is built on group-theoretic features of the phase space. 
The Clifford group action on a quantum system corresponds to the action of the affine symplectic group $G \coloneqq SL(2,\mathbb{Z}_d) \ltimes \mathbb{Z}_d^2$ on the discrete phase space~\cite{Appleby_2005, Bengtsson_2006}. Therefore, any convex mixture of Clifford unitaries will correspond to a convex mixture of these group actions. The Hardy-Littlewood theorem~\cite{hardy_1952} tells us that majorization is obtained from convex mixtures of arbitrary permutations, a result that has also been generalized to convex mixtures of elements of a general group $G$. $G$--majorization has been studied in the classical literature and a range of results are known about it~\cite{giovagnoli_1985, steerneman_1990, giovagnoli_cyclic_1996, eaton_1977}. 
One can therefore consider sub-groups of the affine symplectic group and analyse the resulting majorization/relative majorization that follows from its action on quasi-distributions. However if the sub-group considered is too simple then it may lead to trivial distillation rates. An example of this is when we restrict to the set of Weyl covariant channels~\cite{fukuda2006}, which are represented by a convex mixture of displacements on the phase space. Initial work by the authors have shown that the majorization conditions for Weyl covariant channels can be solved exactly in terms of discrete Fourier transforms, however it is found that the resultant distillation rates are trivial.  

One notable structure that could be exploited for $G$--majorization is when $G$ is a finite reflection group. For this,  the $G$--majorization pre-order is guaranteed to be described by a finite list of conditions~\cite{giovagnoli_1985}, just as is the case for the relative majorization ordering. Therefore one route to concrete lower bounds on distillation rates is to consider stochastic maps obtained by reflection sub-groups of $SL(2,\mathbb{Z}_d) \ltimes \mathbb{Z}_2$, or other simple sub-groups, for which the majorization relation reduces to a finite set of conditions. Beyond this, another promising route is to formulate the majorization conditions for a sub-set of operations as a semi-definite program and then attempt to estimate realisable distillation rates by exploiting dualities and appropriate coarse-grainings. This approach has recently been applied by one of the authors to estimate the set of quantum states accessible by covariant channels in the resource theory of asymmetry~\cite{alexander2021} where one allows depolarizing noise that is then minimized. 
It would be of interest to see if a similar approach could be followed for magic state distillation.
 
Beyond magic state distillation, the topic of $G$--majorization has been extensively studied, but to our knowledge there has not been work on relative $G$--majorization. This would correspond to transformations that are not unital. Physically, this regime would correspond to a form of thermo-majorization obtained from looking at the action of the Clifford group at a micro-canonical level, and then reducing to a small subsystem~\cite{Pathria_1997}. While this seems like a painful thing to consider, there is motivation for this beyond the aim of magic protocols: in the case of classical statistical mechanics on a phase space this is precisely the situation, albeit in the continuum limit. Statistical mechanics of actual systems obey Hamiltonian dynamics, thus they automatically respect a symplectic form~\cite{Arnold_2000, Pathria_1997}. Therefore, the pre-order of statistical mechanical states with respect to phase space dynamics preserving a Gibbs state must correspond to a symplectic majorization condition. Of course, technical features arise in the continuum limit when considering distributions on an unbounded phase space, but this could be remedied by either considering a compact phase space (e.g.~for a particle on a ring) or by first studying the discrete case. Such scenarios also arise in the Quantum Hall Effect~\cite{Klitzing_1980}, which is another regime where these techniques could be of use.

Lastly, we note that it would be of interest to consider the possibility of applying the approach taken here to other scenarios in which one wishes to distinguish classical from non-classical behaviour based on quasi-probability representations~\cite{Ferrie_2008, barnett_1997,Allahverdyan_2018, arvidsson_2020, halpern_2018, Lostaglio_2018, Levy_2020}.

\section*{Methods}

\subsection*{Mathematical methods used for proofs}

In order to derive our results we make use of three broad subfields.

The first one is the phase space representation of quantum systems as described in the first section of our Results.
Generally, there are various possible representations~\cite{Ferrie_2008}, but we make use of the distinguished Wigner representation~\cite{Gross2006} for odd-dimensional systems.

We then develop a theory of majorization~\cite{cit:marshall,  Veinott_1971, Ruch_1976} for quasi-distributions so as to impose a partial order between magic states.

Finally, we extend the theory of classical R\'{e}nyi entropies and divergences~\cite{renyi_1960} for probability distributions to quasi-distributions and show that they possess sufficient properties for the derivation of our entropic bounds. 

Precise definitions of all the above concepts and explanations of their incorporation to the study of magic distillation are provided in the Results section as the concepts become relevant.

\section*{Data availability}
No datasets were generated or analysed during the current study.

\section*{Code availability}
The code used to produce the figures in this work is available from the corresponding author upon reasonable request.

\section*{Acknowledgements}
We thank Earl Campbell, Matteo Lostaglio, Philippe Faist and Nicole Yunger Halpern for helpful comments and discussions. NK is supported by the EPSRC Centre for Doctoral Training in Controlled Quantum Dynamics. DJ is supported by the Royal Society and a University Academic Fellowship.

\section*{Competing interests}
The authors declare no competing interests.

\section*{Author contributions}
DJ proposed the initial idea. 
Derivations were done by NK. 
Analysis and writing of the manuscript were jointly done by NK and DJ.

\providecommand{\noopsort}[1]{}\providecommand{\singleletter}[1]{#1}

\newpage

\section*{Supplementary Note 1: Properties of Wigner representations}

Here we present basic properties of the phase-point operators and the Wigner distribution that are used throughout the main text.

\begin{proposition}\label{thm:aproperties}
    For any dimension $d$, the phase-point operators satisfy:
    \begin{enumerate}
        \item[(i)]\label{en:a1} Hermiticity and unitarity: $A_{\z}^\dagger = A_{\z} = A_{\z}^{-1}$;
	    \item[(ii)]\label{en:a2} Closure under transposition: $A_{(q, p)}^T = A_{(q, -p)}$;
	    \item[(iii)]\label{en:a3} Unit trace for odd $d$: $\tr[A_{\z}] = 1$;
	    \item[(iv)]\label{en:a4} Completeness relation: $\sum_{\z \in \P_d} A_{\z} = d\id$;
	    \item[(i)]\label{en:a5} Orthogonality: $\tr[A_{\z}^\dagger A_{\y}] = d \delta_{\z,\y}$.
	\end{enumerate}
\end{proposition}
All properties follow from the definition in Eq.~(6) in the main text along with properties of the displacement operators $D_{\z}$ and can be found in the literature, e.g.~\cite{app:veitch,app:Vourdas_2004,app:Gross2006}

\begin{proposition}\label{thm:wstate}
  The state Wigner distribution is
  \begin{enumerate}
    \item[(i)]\label{en:w1} Real valued: $W_\rho(\z) \in \mathbb{R}^{d^2}$;
    \item[(ii)]\label{en:w2} Normalized: $\sum_{\z \in \P_d} W_\rho(\z) = 1$;
    \item[(iii)]\label{en:w3} Bounded: $\abs{W_\rho(\z)} \leq \frac{1}{d}$.
    \item[(iv)]\label{en:w4} Additive under mixing: \vspace{2pt}\\
    $W_{p\rho_1 + (1-p)\rho_2}(\z) = p W_{\rho_1}(\z) + (1-p) W_{\rho_2}(\z)$;
    \item[(v)]\label{en:w5} Multiplicative under tensor products: \vspace{2pt}\\
    $W_{\rho_A \otimes \rho_B}(\z_A \oplus \z_B) = W_{\rho_A}(\z_A)W_{\rho_B}(\z_B)$.
	\end{enumerate}
\end{proposition}
\begin{proof}
	Proof of all properties can be found in the literature~\cite{app:veitch,app:Vourdas_2004,app:Gross2006,app:Wang_2019} except for property (iii) which we prove here.
	
Let $\{\lambda_i\}_{i \in \mathbb{Z}_d}$ be the (non-negative) eigenvalues of $\rho$, summing to 1.
Let $\{\alpha_{\z,i}\}_{i \in \mathbb{Z}_d}$ be the eigenvalues of $A_{\z}$. For any $\z, \alpha_{\z,i} \in \{-1, 1\}$, due to the hermiticity and unitarity of the phase-point operators. 
Then,
\begin{align}
	\abs{W_{\rho}(\z)} &= \frac{1}{d}\abs{\tr[A_{\z} \rho]} \leq \frac{1}{d} \abs{\sum_i \alpha_{\z,i} \lambda_i} \nonumber\\ &\leq \frac{1}{d}\sum_i \lambda_i = \frac{1}{d}.
\end{align}
The first inequality follows from Theorem 1 of~\cite{app:mirsky} for the trace of complex matrices, while the second is the Cauchy-Schwarz inequality.
\end{proof}

\begin{proposition}
    \label{thm:wchannel}
    The Wigner distribution of a quantum channel $\E: {\cal B}(\H_{d_A}) \rightarrow {\cal B}(\H_{d_B})$ is
    \begin{enumerate}
        \item[(i)]\label{en:wo1} Real-valued: $W_{\E}(\z|\y) \in \mathbb{R}^{d_A^2} \times \mathbb{R}^{d_B^2}$;
        \item[(ii)]\label{en:wo2} Normalized: $\sum_{\y \in \P_{d_{\scalebox{.9}{$\scriptscriptstyle B$}}}} W_{\E}(\z|\y) = 1$ \\ 
        for any $\z \in \P_{d_A}$;
        \item[(iii)]\label{en:wo3} Bounded: $\abs{W_{\E}(\z|\y)} \leq \frac{d_A}{d_B}$;
	    \item[(iv)]\label{en:wo4} $W_{\E(\rho)}(\z) = \sum\limits_{\y \in \P_{d_{A}}} W_{\E}(\z|\y) W_\rho(\y)$ for any $\z \in \P_{d_B}$.
    \end{enumerate}
\end{proposition}
\begin{proof}
	Proof of all properties are provided in~\cite{app:Wang_2019} except for property (iii) which is a direct consequence of the definition of $W_{\E}$ and property (iii) in Proposition~\ref{thm:wstate}.
\end{proof}

Within the context of the Wigner representation, we define the rescaled quasi-distribution
\begin{equation}
	W_{\rho|\tau}(\z) \coloneqq \frac{W_\rho (\z) }{ W_\tau(\z)},
\end{equation}
which is well-defined if $\tau$ is a full-rank stabilizer state.

The following result shows that the rescaled Lorenz curves behave in a natural way under tensor products. 
\begin{proposition}\label{prop:rescaled_multi}
	Let $\tau_A, \tau_B$ be full rank stabilizer states on systems $A$ and $B$, and let $\rho_A, \rho_B$ be arbitrary states on $A,B$. Then, the rescaled quasi-distribution obeys
\begin{equation}
W_{\rho_A \otimes \rho_B | \tau_A \otimes \tau_B}(\z_A \oplus \z_B) =W_{\rho_A | \tau_A}(\z_A) W_{\rho_B | \tau_B}(\z_B)
\end{equation}
\end{proposition}
\begin{proof} This follows from the multiplicativity of the Wigner distribution,
\begin{align}
W_{\rho_A \otimes \rho_B | \tau_A \otimes \tau_B}(\z_A \oplus \z_B) &= \frac{W_{\rho_A \otimes \rho_B} (\z_A \oplus \z_B)}{W_{\tau_A \otimes \tau_B} (\z_A \oplus \z_B)}\nonumber \\
 = \frac{W_{\rho_A } (\z_A )W_{\rho_B } (\z_B )}{W_{\tau_A } (\z_A )W_{\tau_B } (\z_B )} &= W_{\rho_A | \tau_A}(\z_A) W_{\rho_B | \tau_B}(\z_B). \nonumber
\end{align}
\end{proof}

\section*{Supplementary Note 2: Properties of relative majorization and Lorenz curves}

\subsection*{Relative majorization for quasi-distributions}

We prove the following result connecting majorization to relative majorization~\cite{app:horodecki2013, app:Brandao_2015, app:lostaglio}.
\begin{proposition}\label{relmaj2maj}
	Let $\w \in \mathbb{R}^N, \w' \in \mathbb{R}^{N'}$ be quasi-distributions and $\r \in \mathbb{R}^N, \r' \in \mathbb{R}^{N'}$ probability distributions with positive rational entries given by $r_i = a_i/K$ and $r_i' = a_i'/K$ for positive integers $a_i, a_i'$ and $K = \sum_{i=1}^N a_i = \sum_{i=1}^{N'} a_i'$. 
Then,
\begin{equation}
(\w, \r) \succ (\w', \r') \mbox{ if and only if } \Gamma_{\bma} (\w) \succ \Gamma_{\bma'} (\w'),
\end{equation}
where the embedding map $\Gamma_{\bma} :\mathbb{R}^N \rightarrow \mathbb{R}^K$ is given by
\begin{equation}
	\Gamma_{\bma}(\z) \coloneqq \bigoplus_{i=1}^N z_i \bmeta_{a_i},
\end{equation}
with $\bmeta_{a_i} = (1/a_i, 1/a_i, \dots, 1/a_i)$ the uniform distribution on $a_i$ elements.
\end{proposition}
\begin{proof}
	We first note that the map $\Gamma_{\bma}$ is stochastic and has a well-defined left-inverse $\Gamma^{-1}_{\bma}: \mathbb{R}^K \rightarrow \mathbb{R}^N$ given by
\begin{equation}
	\Gamma^{-1}_{\bma}(\x) = \left( \sum_{i=1}^{a_1} x_i,  \sum_{i = a_1 +1}^{a_1 + a_2} \hspace{-5pt} x_i, \dots , \hspace{-15pt} \sum_{i=a_1 + \dots + a_{N-1}+1}^{K} \hspace{-10pt} x_i \right),
\end{equation}
that obeys $(\Gamma^{-1}_{\bma} \circ \Gamma_{\bma}) (\z) = \z$ for all $\z \in \mathbb{R}^N$. 

The claim is equivalent to the statement that there exists a bistochastic map $B$ sending $\Gamma_{\bma}(\w)$ to $\Gamma_{\bma'}(\w')$ if and only if there exists a stochastic map $A$ sending $\w$ to $\w'$ and $\r$ to $\r'$.

Suppose there is a stochastic map $A$ such that $A\w=\w'$ and $A\r = \r'$. 
We define $B:\mathbb{R}^K \mapsto \mathbb{R}^K$ by $B \coloneqq \Gamma_{\bma'} \circ A \circ \Gamma_{\bma}^{-1}$ so that it is stochastic as a composition of stochastic maps and it preserves the uniform distribution in $\mathbb{R}^K$, since
\begin{align}
	&B(1/K, 1/K, \dots, 1/K) = \big( \Gamma_{\bma'} \circ A \circ \Gamma_{\bma}^{-1} \big) \big( \Gamma_{\bma'}(\r) \big) = \nonumber\\
	&\Gamma_{\bma'} (\r') = (1/K, 1/K, \dots, 1/K),
\end{align}
therefore $B$ is bistochastic.
Finally, $B$ maps the embedded distributions as follows,
\begin{equation}
	B \Gamma_{\bma}(\w) = \big( \Gamma_{\bma'} \circ A \circ \Gamma_{\bma}^{-1} \big) \big( \Gamma_{\bma}(\w) \big) = \Gamma_{\bma'} (\w').
\end{equation}

Conversely, suppose a bistochastic map $B$ exists sending $\Gamma_{\bma}(\w)$ to $\Gamma_{\bma'}(\w')$. Again, define $A: \mathbb{R}^N \mapsto \mathbb{R}^{N'}$ by $A \coloneqq \Gamma_{\bma'}^{-1} \circ B \circ \Gamma_{\bma}$ so that it is stochastic as a composition of stochastic maps and $A \w = \big( \Gamma_{\bma'}^{-1} \circ B \circ \Gamma_{\bma} \big) \w = \w'$, as well as $A \r = \r'$.
\end{proof}
We note that the vectors $\r, \r'$ with rational components form a dense subset of the positive probability distributions, and do not consider further technicalities that have no impact on actual physical measurements, which always have a finite resolution.

Restricting the above result in the context of proper probability distributions, we can prove the Lorenz curve condition for probability distributions.
We define a Lorenz curve elbow as a point where the slope of the curve changes, expressed explicitly in Eq.~(15) for Lorenz curve $L_{\w | \r}(x)$
\begin{proposition}\label{thm:lc_equiv_prob}
    Given probability distributions $\w, \r \in \mathbb{R}^{N}$ and $\w', \r' \in \mathbb{R}^{N'}$ with $\r,\r'$ having positive components, then
\begin{equation*}
	(\w, \r) \succ (\w', \r') \mbox{ if and only if } L_{\w |\r}(x) \ge L_{\w' |\r'}(x),
\end{equation*}
for all $x \in [0,1]$.
\end{proposition}
\begin{proof}
	In the case of $N=N'$ and $\r = \r' = \bmeta$, where $\bmeta = (1/N,1/N,\dots, 1/N)$ is the uniform distribution on $N$ elements, the statement reduces to the Lorenz curve condition $L_{\w}(x) \ge L_{\w'}(x)$ for all $x$, for simple majorization which follows immediately from the defining set of inequalities for majorization~\cite{app:marshall}. Namely, the Lorenz curve $L_{\w}(x)$ for $\w$ is obtained from the partial sums of $\w$ sorted in non-increasing order. It is also clear from the definition that the Lorenz curve of $\w$ is given by $L_{\w}(x) = L_{\w | \bmeta}(x)$.

To prove the general statement we reduce the relative majorization Lorenz curve condition to standard majorization. Using the notation and assumptions of Proposition~\ref{relmaj2maj} for distributions with rational components, we can define $\Gamma_{\bma} (\w), \Gamma_{\bma'} (\w')$ and the uniform distribution $\bmeta = (1/K, 1/K, \dots, 1/K)$.

The key ingredient in the proof is that the Lorenz curve of $\w$ relative to $\r$ coincides with the Lorenz curve of $\Gamma_{\bma} (\w)$, namely
\begin{equation}
	L_{\w | \r}(x) = L_{\Gamma_{\bma} (\w) | \bmeta}(x) = L_{
	\Gamma_{\bma}(\w)}(x) \mbox{ for all } x \in [0,1].
\end{equation}
To see this we consider the elbows of $L_{\w | \r}(x)$,
\begin{equation}
	(x_k, L_{\w|\r}(x_k)) = \left( \sum_{i=1}^k r_{\pi(i)}, \sum_{i=1}^k w_{\pi(i)} \right),
\end{equation}
where the permutation $\pi$ sorts $(w_i/r_i)$ in non-increasing order.
Expressing $\Gamma_{\bma} (\w)$ as
\begin{equation}
	\Gamma_{\bma} (\w) = \frac{1}{K} \bigoplus_{i=1}^N \left( \frac{w_i}{r_i}, \dots, \frac{w_i}{r_i} \right),
\end{equation}
where $(w_i/r_i, \dots, w_i/r_i)$ has $a_i$ elements, it is clear that permutation $\pi$ sorts $\Gamma_{\bma}(\w)$ in non-increasing order too.
The Lorenz curve elbows $(y_k, L_{\Gamma_{\bma} (\w) | \bmeta}(y_k))$ occur at
\begin{equation}
	y_k = \sum_{j=1}^{a_{\pi(1)} + \dots + a_{\pi(k)}} \frac{1}{K} = \sum_{i=1}^k \frac{a_{\pi(i)}}{K} = x_k
\end{equation}
and take values
\begin{equation}
	L_{\Gamma_{\bma} (\w) | \bmeta}(y_k) = \sum_{i=1}^k a_{\pi(i)}\frac{1}{K}\frac{w_{\pi(i)}}{r_{\pi(i)}} = L_{\w|\r}(x_k).
\end{equation}
Therefore, the Lorenz curve of the embedded distribution coincides with the Lorenz curve in the relative majorization setting.

Finally, we have $(\w, \r) \succ (\w', \r')$  if and only if $\Gamma_{\bma} (\w) \succ \Gamma_{\bma'} (\w')$, which holds if and only if $L_{\Gamma_{\bma} (\w) | \bmeta}(x) \geq L_{\Gamma_{\bma'} (\w') | \bmeta}(x)$ for all $x \in [0,1]$, which in turn holds if and only if $L_{\w | \r}(x) \geq L_{\w' | \r'}(x),\ x \in [0,1]$,  which concludes the proof.
\end{proof}

We now verify useful majorization properties in the context of quasi-probability distributions, starting from the property that the Lorenz curve is concave.
\begin{proposition}\label{L-concave} 
	Let $\w$ be a quasi-distribution and let $\r$ be a probability distribution with strictly non-zero components. Then $L_{\w|\r}(x)$ is a concave function on $[0,1]$.
\end{proposition}
\begin{proof} 
	Let $x_\star$ be the point where $L_{\w|\r}(x)$ first attains it maximum. Therefore, on $[0,x_\star]$ the function rises monotonically to $L_{\w|\r}(x_\star)$, via the sum of all positive entries of $\w$, taken in decreasing order. Likewise on $[x_\star, 1]$, the function decreases monotonically from its maximum via the partial sums of the negative entries of $\w$ in decreasing order. Let $f_1(x)$ be equal to $L_{\w|\r}(x)$ on $[0, x_\star]$ and $0$ otherwise. Also let $f_2(x)$ be equal to $L_{\w|\r}(x)$ on $[x_\star,1]$ and and $0$ otherwise. By inspection both $f_1$ and $f_2$ are concave functions, and $f_1(x) + f_2 (x) = L_{\w|\r}(x)$ for all $x\in [0,1]$. However, the sum of two concave functions is also concave which concludes the proof.
\end{proof}

The following result is used in the main text to extend the Lorenz curve condition of Proposition~\ref{thm:lc_equiv_prob} into the context of quasi-distributions.
\begin{proposition}\label{Lorenz_linearity}
	Let $\w$ be a quasi-distribution and let $\r$ be a probability distribution with strictly non-zero components. 
	Then, $L_{a\w + b \r | \r} (x) = a L_{\w |\r}(x) + b x$ for any constants $a > 0$ and $b \in \mathbb{R}$.
\end{proposition}
\begin{proof} 
	The Lorenz curve of $a\w + b \r$ relative to $\r$ passes through $(0,0)$ and the points $(\sum_{i=1}^k{r_{\pi(i)}}, \sum_{i=1}^k(a \w + b \r)_{\pi(i)})$ where $\pi$ is the permutation that puts $(a w_i/r_i + b)$ in non-increasing order. Since $a > 0$, the permutation $\pi$ puts  $(w_i/r_i)$ in non-increasing order too. We thus have
\begin{align*}
&\left( \sum_{i=1}^k r_{\pi(i)}, \sum_{i=1}^k(a \w + b \r)_{\pi(i)} \right) = \\ 
&\left( \sum_{i=1}^k r_{\pi(i)},a \sum_{i=1}^k  w_{\pi(i)} + b\sum_{i=1}^k r_{\pi(i)} \right) \nonumber,
\end{align*}
so the value of the Lorenz curve at each potential elbow point $x_k = \sum_{i=1} ^kr_{\pi(i)}$ is given by
\begin{align}
&L_{a \w +b \r|\r} (x_k) = a L_{\w|\r} (x_k) + b L_{\r|\r}(x_k) = \nonumber\\
&a L_{\w|\r} (x_k) + b x_k,
\end{align}
so we have $L_{a\w  + b\r|\r} (x) = a L_{\w |\r}(x) + b x$ for any $x \in [0,1]$ due to linearity.
\end{proof}
	
The following provides equivalent formulations of relative majorization on quasi-distributions. 

\begin{proposition}\label{prop:rmajor}
Given quasi-distributions $\w, \w'$, $\r , \r'$, such that the components of $\r$ and $\r'$ are positive, the following statements are equivalent:
  \begin{enumerate}
    \item[(i)] $\w' = A\w$ and $\r' = A\r$ for a stochastic map $A$;
    \item[(ii)] $L_{\w|\r}(t) \geq L_{\w'|\r'}(t)$ for $t\in [0,1)$ and \vspace{5pt}\\ $L_{\w|\r}(1) = L_{\w'|\r'}(1)$;
    \item[(iii)] $\sum\limits_{i=1}^n \abs{w_i - r_i t} \geq \sum\limits_{i=1}^n \abs{w'_i - r'_i t}$ for all $t \in \mathbb{R}$.
  \end{enumerate}
\end{proposition}
\begin{proof}
	The proofs for these properties on proper probability distributions can be found in~\cite{app:marshall,app:ruch_mixing_1978,app:Renes_2016,app:Buscemi_2017} and references therein.
	
	The equivalence between (i) and (ii) on quasi-distributions was proven in Theorem~1 in the main text.
	
	The equivalence between (i) and (iii) follows from a similar argument where we mask negativity with a probability distribution.
	Namely, let $\epsilon > 0$ be such that $\w(\epsilon) \coloneqq \epsilon \w + (1-\epsilon) \r$ and $\w'(\epsilon) \coloneqq \epsilon \w + (1-\epsilon) \r'$ are genuine probability distributions. This is guaranteed by picking a sufficiently small $\epsilon > 0$, since $\r$ has positive components. We now have that $(\w , \r) \succ (\w', \r')$ if and only if $(\w(\epsilon) , \r) \succ (\w'(\epsilon), \r')$.
	Moreover, we have that for all $c \in \mathbb{R}$, $\sum_{i=1}^n \abs{w(\epsilon)_i - r_i c} \geq \sum_{i=1}^n \abs{w'(\epsilon)_i - r'_i c}$, which leads to
	\begin{align*}
		\sum\limits_{i=1}^n \abs{w_i - r_i \frac{\epsilon + c - 1}{\epsilon}} \geq \sum\limits_{i=1}^n \abs{w'_i - r'_i \frac{\epsilon + c - 1}{\epsilon}}
	\end{align*}
Replacing $t = (\epsilon + c - 1)/\epsilon$, we see that $t$ can attain any real value for $c \in \mathbb{R}$, so we deduce the required $L_1$--norm condition on quasi-distributions $\w, \w'$.
\end{proof}

\subsection*{Component-multiplicity pairs}

In general, a $1$--copy $d$--dimensional state $\rho$ is described by its $d^2$--dimensional Wigner distribution $W_\rho$. 
The distribution $W_\rho$ is defined on the phase space, but it can be convenient to re-express this using vector notation.
We discuss this in terms of Wigner distributions, but there is nothing to prevent the discussion from applying to rescaled Wigner distributions as well.

To each Wigner distribution $W_\rho(\z)$ we can associate component-multiplicity pairs $\{(w_i, m_i)\}$ where the value $w_i$ occurs in the distribution $W_\rho(\z)$ with multiplicity $m_i$.

As an example, for the Strange state $\rho_S \coloneqq \rho_S(\epsilon=0)$ with
\begin{equation}
\hspace{-5pt} W_{\rho_S} = (-1/3, 1/6,  1/6,  1/6,  1/6,  1/6,  1/6,  1/6,  1/6)
\end{equation}
we have the component-multiplicity pairs: $\{( -1/3, 1), ( 1/6, 8)\}$. However, we might also wish more freedom and not require that the different $w_i$ values are all distinct. For example, the component-multiplity pairs $\{(-1/3, 1), (1/6, 2), (1/6, 3), (1/6, 3)\}$ also describe $W_{\rho_S}$.

This representation is more compact when a Wigner distribution has a lot of multiplicities, and allows for simple handling of multiple copies via the following fact.
\begin{proposition}
Consider two Wigner distributions $W_{\rho_A}(\z_A), W_{\rho_B}(\z_B)$ with component-multiplicity pairs 
\begin{equation}
	\{(w_i, m_i)\} \text{ and } \{(w_j', m_j')\},
\end{equation}
respectively. Then, $\{(w_i w_j', m_i m_j')\}$ gives component-multiplicity pairs for $W_{\rho_A \otimes \rho_B}(\z_A \oplus \z_B)$.
\end{proposition}
\begin{proof}
	This result is true because all components of $W_{\rho_A \otimes \rho_B}(\z_A \oplus \z_B)$ are of the form $w_i w_j'$ and 
\begin{equation*}
	\sum_{i}\sum_{j} m_i m_j' = \sum_{i} m_i \sum_{j} m_j' = d_A^2 d_B^2,
\end{equation*}
where $d_A, d_B$ are the dimensions of $\rho_A, \rho_B$ respectively, and
so the set $\{(w_i w_j', m_i m_j')\}$ contains exactly the Wigner components of $W_{\rho_A \otimes \rho_B}(\z_A \oplus \z_B)$.
\end{proof}
The above result also applies when we consider rescaled Wigner distributions, crucially due to the multiplicativity result shown in Proposition~\ref{prop:rescaled_multi}.

In the case of distillation protocols, we are interested in copies of distributions. 
For this, we have the following result, which follows from combinatorics.
\begin{proposition}\label{ncopycomponents}
	Suppose $W_\rho$ has a set of $D$ component-multiplicity pairs $\{(w_i, m_i)\}$. Then, $W_{\rho^{\otimes n}}$ has component-multiplicity pairs $\{(W_{\q}, M_{\q})\}$, with index $\q$ running through all vectors $(q_1, \dots, q_D)$, where $q_1, \dots, q_D$ are non-negative integers that sum to $n$, and
\begin{align}
	W_{\q} &= \prod_{i=1}^D w_i^{q_i} \label{eq:W}\\
	M_{\q} &= \binom{n}{q_1,q_2, \dots, q_d} \prod_{i=1}^D m_i^{q_i}. \label{eq:M}
\end{align}
The term outside the product in the expression for $M_{\q}$ is the generalized binomial coefficient,
\begin{equation}
	\binom{n}{q_1,q_2, \dots, q_d} \coloneqq \frac{n!}{q_1!\dots q_D!}.
\end{equation}
\end{proposition}
\begin{proof}
	Denote by $C_D^n \coloneqq \{\k\}$ the set of all vectors $\k \coloneqq (k_1, \dots, k_D)$ with non-negative integer components that sum to $n$, i.e.
	\begin{equation*}
	0 \leq k_1, \dots, k_D \leq n \text{ and } k_1 + \dots + k_D = n.
	\end{equation*}
	
	We proceed by induction.	
	Assume $n = 1$ and let $\k_i$ be the vector with its $i$-th component equal to 1 and 0's elsewhere.
	The set $C_D^1$ consists of all vectors of this form, i.e. 
\begin{equation*}
	C_D^1 = \{ \k_i \}_{i=1,\dots,D}
\end{equation*}
	It is also true by direct calculation that
\begin{equation*}
	\left( W_{\k_i}, M_{\k_i} \right) = (w_i, m_i).
\end{equation*}
Therefore, $\{ (W_{\k}, M_{\k}) \}_{\k \in C_D^1}$ is a complete set of component-multiplicity pairs for $W_\rho$.

	Assume that $\{(W_{\k}, M_{\k})\}_{\k \in C_D^n}$ as given in Eqs.~(\ref{eq:W}, \ref{eq:M}) is a complete set of component-multiplicity pairs for the $n$--copy distribution $W_{\rho^{\otimes n}} = W_{\rho}^{\otimes n}$.
	By construction, the distribution $W_{\rho}^{\otimes (n+1)} = W_{\rho}^{\otimes n} \otimes W_{\rho}$, so it admits the complete set of component multiplicity pairs
\begin{equation}
	\{(W_{\k} w_i, M_{\k} m_i)\},\ \k \in C_D^n \text{ and } i=1,\dots,D.
\end{equation}
	
	Consider the component sum of distribution $W_{\rho}^{\otimes (n+1)}$,
\begin{align*}
	&\sum_{\k \in C_D^n}\sum_{i=1}^D M_{\k} m_i W_{\k} w_i = \sum_{\k \in C_D^n} M_{\k}W_{\k} \sum_{i=1}^D m_i w_i =\\
	&\sum_{\k \in C_D^n} \frac{n!}{k_1!\dots k_D!} \prod\limits_{i=1}^D {m_i}^{k_i}{w_i}^{k_i} \sum_{i=1}^D m_i w_i =\\
	&\left( \sum_{i=1}^D m_i w_i \right)^n \left( \sum_{i=1}^D m_i w_i \right) = \left( \sum_{i=1}^D m_i w_i \right)^{n+1} =\\
	&\sum_{\q \in C_D^{n+1}} M_{\q}W_{\q},
\end{align*}
where in the last expression, vectors $\q = (q_1, \dots, q_D)$ have non-negative integer components that sum to $(n+1)$ and 
\begin{align*}
	M_{\q} &= \frac{(n+1)!}{q_1!\dots q_D!} \prod\limits_{i=1}^D {m_i}^{q_i},\\
	W_{\q} &= \prod\limits_{i=1}^D {w_i}^{q_i}.
\end{align*}
We have used the multinomial theorem to proceed between lines 2-3 and lines 3-4 of the derivation.

We have achieved a regrouping of the distribution components.
Every component $W_{\q}$ is of the form $W_{\k} w_i$ with $q_i = k_i + 1$ and $q_j = k_j$ for $j\neq i$ and 
\begin{align*}
	\sum_{\q \in C_D^{n+1}}  \hspace{-6pt} M_{\q} &=  \hspace{-10pt} \sum_{\q \in C_D^{n+1}} \frac{(n+1)!}{q_1!\dots q_D!} \prod\limits_{i=1}^D {m_i}^{q_i} = \left( \sum_{i=1}^D m_i \right)^{n+1} \hspace{-10pt} \\
	&= d^{2(n+1)},
\end{align*}
which is the dimension of $W_{\rho}^{\otimes (n+1)}$.

Therefore, $\{ (W_{\q}, M_{\q}) \}_{\q \in C_D^{n+1}}$ contains exactly the components of $W_{\rho}^{\otimes n}$, completing the proof.
\end{proof}
Once again, the above result applies on rescaled Wigner distributions as well.

\subsection*{Results on Lorenz curve constraints}

We first prove the following result on the relation between majorization in a particular sub-theory $\R_\sigma$ and mana.
\begin{customthm}{5}
	Given a magic state $\rho$, the maximum $L_\star$ of its Lorenz curve $L_{\rho|\sigma}(x)$ is independent of the $\R_\sigma$ and equal to $1+{\rm sn}(\rho)$. Moreover, the majorization constraint is stronger than mana in every $\R_\sigma$.
\end{customthm}
\begin{proof}
	We denote the Wigner distributions of the states single-component vectors $\w(\rho)=W_\rho(\z)$ and $ \w(\sigma)=W_\sigma(\z)$. Likewise, we write $\w(\rho|\sigma) = W_{\rho|\tau}(\z)$.
	We choose the component indexing so that $\w(\rho|\sigma)^\downarrow = \w(\rho|\sigma)$, so the components are sorted in non-increasing order.

Note that all components of $\w(\sigma)$ are positive, so $w(\rho|\sigma)_i \geq 0$ if and only if $w(\rho)_i \geq 0$ for any $i=1,\dots,d^2$.
	
	Let $i_\star$ be the index of the smallest non-negative component of $\w(\rho|\sigma)^\downarrow$.
	Then, $w(\rho)_i < 0$ if and only if $i > i_\star$, so the maximum of Lorenz curve $L_{\rho|\sigma}(x)$ takes the value 
	\begin{equation}
		L_{\rho|\sigma}(x_{i_\star}) = \sum_{i=1}^{i_\star} w(\rho)_i,
	\end{equation}
	and is achieved at $x_{i_\star} \coloneqq \sum_{i=1}^{i_\star} w(\sigma)_i$. The location of this maximum varies with $\R_\sigma$, but its value is independent of $\sigma$,
	\begin{align}
	L_\star &\coloneqq	L_{\rho|\sigma}(x_{i_\star}) 
		= \sum\limits_{\z: W_{\rho}(\z) \geq 0}\hspace{-5pt} W_{\rho}(\z) \nonumber \\
		&= 1 + {\rm sn}(\rho).
	\end{align}
	
Since mana is a monotonic function of sum-negativity, $\M(\rho) = \log \hspace{1pt}(2\hspace{1pt}{\rm sn}(\rho)+1)$, we see that mana determines the peak of the Lorenz curve $L_{\rho|\sigma}(x)$. However, mana is one of $d^{2n}$ constraints, so majorization is strictly a stronger constraint in any $\R_\sigma$.
\end{proof}

We now prove a simple majorization constraint, arising by considering only the part of the Lorenz curves between the origin $(0,0)$ and the first elbow.
\begin{proposition}[First elbow constraint]\label{prop:first_elb}
	Consider a magic state process $\rho \longrightarrow \rho'$ with input and output Lorenz curves $L_{\rho|\sigma}(x), L_{\rho'|\sigma'}(x)$ and denote the coordinates of the first elbow of $L_{\rho|\sigma}(x)$ by $(X_0, L_0)$ and the coordinates of the first elbow of $L_{\rho' |\sigma'}(x)$ by $(X'_0, L'_0)$.
	
Then, given any coordinates $(x, L)$ and $(x', L')$ on the input and output Lorenz curves respectively, where $0 < x \leq X_0$ and $0 < x' \leq X'_0$, the process is possible only if
\begin{equation}\label{eq:first_elb_bound1}
	\frac{L}{x} \geq \frac{L'}{x'}.
\end{equation}
\end{proposition}
\begin{proof} 
	If the transformation is possible, then $L_{\rho|\sigma}(x) \geq L_{\rho'|\sigma'} (x)$ for all $x$ in $[0,1]$. 
	Restricting to the initial linear segment of $L_{\rho|\sigma}(x)$ joining $(0,0)$ to $(X_0,L_0)$, we then require that it be above the initial linear segment of $L_{\rho'|\sigma'}(x)$ joining $(0,0)$ to $(X_0', L_0')$. 
	Since both these linear segments start at the origin, this condition is equivalent to the slope of the initial segment of the input Lorenz curve being greater than or equal to the slope of the initial segment of the output Lorenz curve. 
	This slope constraint is then equivalent to Eq.~(\ref{eq:first_elb_bound1}). \end{proof}

The following result reduces the Lorenz curve condition on the interval $[0,1]$ to a finite set of independent inequalities and can therefore simplify the computation of relative majorization constraints in practical scenarios.
\begin{proposition}\label{thm:elbows}
	Let $\rho, \rho'$ be any two quantum states with Lorenz curves $L_{\rho|\sigma}(x), L_{\rho'|\sigma'}(x)$, where $\sigma, \sigma'$ are states with positive Wigner distributions. Assume that $L_{\rho'|\sigma'}(x)$ has $t$ elbows at locations $x_1, \dots, x_t$. Then, $L_{\rho|\sigma}(x) \geq L_{\rho'|\sigma'}(x)$ for all $x \in [0,1]$ iff $L_{\rho|\sigma}(x_{i}) \geq L_{\rho'|\sigma'}(x_{i})$ for all $i =1,\dots,t$.
\end{proposition}
\begin{proof}	
If $L_{\rho|\sigma}(x) \geq L_{\rho'|\sigma'}(x)$ for all $x \in [0,1]$, then the condition on the elbows follows trivially.

Conversely, assume $L_{\rho|\sigma}(x_i) \geq L_{\rho'|\sigma'}(x_i)$ for all $i=1,\dots t$. 
Suppose on the contrary that $L_{\rho|\sigma}(x)$ dips below $L_{\rho'|\sigma'}(x)$ at some point $x=y$ where $x_i < y < x_{i+1}$ for some $i$. Then this implies
\begin{align*}
L_{\rho|\sigma}(x_i) &\ge L_{\rho'|\sigma'}(x_i) \\
L_{\rho|\sigma}(y) & <L_{\rho'|\sigma'}(y) \\
L_{\rho|\sigma}(x_{i+1}) &\ge L_{\rho'|\sigma'}(x_{i+1}).
\end{align*}
However, since $L_{\rho'|\sigma'}(x)$ is linear between $x_i$ and $x_{i+1}$, the above conditions imply that $L_{\rho|\sigma}(x)$ is not concave, which contradicts Proposition~\ref{L-concave}.
Therefore, $L_{\rho|\sigma}(x) \geq L_{\rho'|\sigma'}(x)$ for all $x$ in $[0,1]$.
\end{proof}

\section*{Supplementary Note 3: General aspects of Lorenz curves for noisy Strange states}

Here we provide background detail as to how the Lorenz curve behaves for a noisy Strange state in the large $n$ limit.

\subsection*{Binomial distributions and error bounds}
Consider an experiment consisting of $n$ trials of throwing a coin with probability $p$ of landing on heads and $1-p$ of landing on tails.
We denote the sums over an even number $m$ of successful trials by $\Phi_+(m; n, p)$ and an odd number $m$ of successful trials by $\Phi_-(m; n, p)$, and they are given by
\begin{align}	
	\Phi_+(m; n, p) &\coloneqq \sum\limits_{\ell=0}^{m/2} \binom{n}{2\ell} p^{2\ell} (1-p)^{n-2\ell}, \nonumber\\ 
	&\text{for even integers } m\in[0,n], \label{eq:fp_app} \\
	\Phi_-(m; n, p) &\coloneqq \sum\limits_{\ell=1}^{(m-1)/2} \binom{n}{2\ell+1} p^{2\ell+1} (1-p)^{n-(2\ell+1)}, \nonumber\\ 
	&\text{for odd integers }m\in[0,n]. \label{eq:fn_app}
\end{align}
In the next section we will use $\Phi_+$ and $\Phi_-$ to express the elbow coordinates of Lorenz curves for unital protocols.

We have the Shannon entropy of a $p$--coin and the Shannon relative entropy for a $p$--coin and a $q$--coin given by the expressions
\begin{align}
	S(p) &\coloneqq -p\log{p} -(1-p)\log{(1-p)}, \label{eq:ent}\\
	\ent{p}{q} &\coloneqq p \log{\frac{p}{q}} + (1-p) \log{\frac{1-p}{1-q}}. \label{eq:ent_rel}
\end{align}
They are symmetric about $p=1/2$ and $q=1/2$, namely $S(p) = S(1-p)$ and $\ent{p}{q} = \ent{1-p}{1-q}$. We also have the entropic bounds
\begin{proposition}\label{comb_bounds}
	For all $\ell\in [1,n-1]$,
	\begin{align}
		&\left[ 8\ell\left(1-\frac{\ell}{n}\right) \right]^{-\frac{1}{2}} 2^{n S\left(\frac{\ell}{n}\right)} \leq \binom{n}{\ell} \leq \\
		&\left[ 2\pi \ell\left(1-\frac{\ell}{n}\right) \right]^{-\frac{1}{2}} 2^{n S\left(\frac{\ell}{n}\right)}.
	\end{align}
\end{proposition}
The proof provided in~\cite{app:ash} proceeds with direct calculation for the edge cases $\ell = 1,2, n-1, n-2$ and use of Stirling's approximation for the remaining cases. The bounds in Proposition~\ref{comb_bounds} can be directly inserted into the expressions of functions $\Phi_+$ and $\Phi_-$ to derive strict upper and lower bounds.

It is also possible to derive more manageable bounds on $\Phi_+, \Phi_-$.
To this end, we rewrite the functions as
\begin{equation}
	\Phi_{\pm}(m; n, a) = \frac{1}{2}(\Phi(m; n, a) \pm (1+a)^{-n} S(m; n, a)),
\end{equation}
where $a = p/(1-p)$, and $\Phi$ is the standard cumulative function for the binomial distribution
\begin{equation}
	\Phi(m; n, a) = (1+a)^{-n} \sum_{k=0}^m \binom{n}{k} a^k,
\end{equation}
and we define the remainder term
\begin{equation}
	S(m; n, a) \coloneqq \sum_{k=0}^m \binom{n}{k} (-a)^k.
\end{equation}
The cumulative function $\Phi$ obeys the following bounds, which are proven in~\cite{app:ash}.
\begin{proposition}\label{phi_bounds}
	Given fixed $n>0$ and $p$, $\Phi$ satisfies the following bounds:
	\begin{align*}
		\begin{split}
		&\text{1. } \Phi(m; n, p) \geq \left[ 8m\left(1-\frac{m}{n}\right) \right]^{-\frac{1}{2}} 2^{-n\ent{\frac{m}{n}}{p}}, \\
		&\hspace{14pt} m\in [1,n-1] \\
		&\text{2. } \Phi(m; n, p) \geq 1 - 2^{-n\ent{\frac{m+1}{n}}{p}},\ m\in [np+1,n-2] \\
		&\text{3. } \Phi(m; n, p) \leq 1 - \left[ 8(m+1)\left(1-\frac{m+1}{n}\right) \right]^{-\frac{1}{2}}\times \\
		&\hspace{14pt} 2^{-n\ent{\frac{m+1}{n}}{p}},\ m\in [0,n-2]
		\end{split}
		\\
		&\text{4. } \Phi(m; n, p) \leq 2^{-n\ent{\frac{m}{n}}{p}},\ m\in [0,np]
	\end{align*}
\end{proposition}

We can now put estimates on $S(m; n, a)$, and hence the functions $\Phi_{\pm}$, for different parameter regimes. 
We can consider the function $f(x) = (1+x)^n$ and note that $S(m; n, a)$ is the $m$'th partial sum of this expansion at the point $x=-a$. The truncated Maclaurin series is given by
\begin{equation}
	f(x) = f(0) + x f'(0) + \dots \frac{x^m}{m!}f^{(m)}(0) + R_m(x)
\end{equation}
with a remainder term
\begin{align}
	R_m (x)&= \int_{0}^x dt f^{(m+1)}(t) \frac{(x-t)^m}{m!} \\
	&= \frac{x^{m+1}}{(m+1)!} f^{(m+1)}(x_*),
\end{align}
where in the second expression, $x_*$ is a point that lies between $0$ and $x$ that comes from the Mean Value Theorem.

Applying this to the function $f(x) = (1+x)^n$ gives
\begin{equation}
	(1+x)^n = \sum_{k=0}^m \binom{n}{k} x^k + R_m.
\end{equation}
Therefore, if we evaluate at $x=-a$ we obtain
\begin{equation}
	S(m; n, a) = (1-a)^n - R_m(-a),
\end{equation}
where the remainder term is given by
\begin{align}
	R_m(-a) &= \int_0^{-a} dt f^{(m+1)}(t) \frac{(-a-t)^m}{m!} \\
&= \frac{(-a)^{m+1}}{(m+1)!} f^{(m+1)}(x_*).
\end{align}
We can compute the derivative $f^{(m+1)}(x)$ explicitly,
\begin{equation}
	f^{(m+1)}(x) = (m+1)!\binom{n}{m+1}(1+x)^{n-m-1}.
\end{equation}
Therefore, we have that
\begin{align}
	R_m(-a) &= (-1)^{m+1}(m+1)\binom{n}{m+1}\times \nonumber\\
	&\hspace{12pt} \int_{-a}^0 dt (1+t)^{n-m-1}(a+t)^m \\
&= \binom{n}{m+1}(-a)^{m+1}(1+x_*)^{n-m-1},
\end{align}
where in the latter expression $x_* \in [-a,0]$. 
The first integral expression can be estimated via the Cauchy-Schwarz or the H{\"o}lder inequality. 
Therefore, one has an explicit form with an unknown (but bounded) parameter $x_*$, or the integral form.

A very simple estimate, based on $x_*$ lying in the interval $[-a,0]$ gives
\begin{equation}
(1-a)^{n-m-1} \leq \frac{R_m(-a)}{\binom{n}{m+1}(-a)^{m+1}} \leq 1,
\end{equation}
which in turn leads to the following bounds on $\Phi_+(m; n, a)$:
\begin{align}
	\hspace{-2cm}2\Phi_+(m; n, a) \leq\ &\Phi(m; n, a) + (1-a)^n - \frac{(-a)^{m+1}}{m!} \times \nonumber\\
	 &\binom{n}{m+1} (1-a)^{n-m-1} \text{ and} \\
	2\Phi_+(m; n, a) \geq\ &\Phi(m; n, a) + (1-a)^n - \frac{(-a)^{m+1}}{m!}  \times \nonumber\\
	 &\binom{n}{m+1} .
\end{align}
Better bounds can be obtained with a finer analysis, and it would be of interest to obtain tighter asymptotic behaviour on $\Phi_{\pm}$ that could in turn be used to improve magic distillation bounds.

\subsection*{Lorenz curve coordinates for unital protocols}

The Wigner distribution of the $n$--copy qutrit maximally mixed state $\left(\id/3\right)^{\otimes n}$ is the uniform probability distribution over the phase space, consisting of $9^n$ components equal to $9^{-n}$.
The Wigner distribution of the 1-copy $\epsilon$--noisy Strange state $\rho_{\rm{S}}(\epsilon)$ for unital protocols consists of some permutation of a single negative component
\begin{equation}
	- v(\epsilon) \coloneqq - \left( \frac{1}{3} -\frac{4}{9}\epsilon \right),
\end{equation} 
and $8$ positive components each with value
\begin{equation}
	u(\epsilon) \coloneqq \frac{1}{6} -\frac{1}{18}\epsilon.
\end{equation}

For unital protocols we need the condition $0 \leq \epsilon < 3/4$, so that the state contains some Wigner negativity ($-v < 0$).
It is also clear that $v \geq u$ in the interval $0 \leq \epsilon \leq 3/7$, while $u > v$ in the interval $3/7 < \epsilon < 3/4$.

The Wigner distribution of the $n$--copy $\epsilon$--noisy Strange state $\rho_{\rm{S}}(\epsilon)^{\otimes n}$ is given by $W_{\rho_{\rm{S}}(\epsilon)^{\otimes n}} = W_{\rho_{\rm{S}}(\epsilon)}^{\otimes n}$ and, using Proposition~\ref{ncopycomponents} for unital protocols, we can associate it with a set $\{(w_i, m_i)\}_{i=0,\dots,n}$ of component-multiplicity pairs, the elements of which are presented in Supplementary Table~\ref{tab:lcsu} for all different combinations of physical parameters $n, \epsilon$.
We split the values of index $i$ into two intervals, defined precisely as
\begin{align}
&\text{LHS: } 0 \leq i \leq \left\lfloor \frac{n}{2} \right\rfloor \text{ and} \\
&\text{RHS: } \left\lfloor \frac{n}{2} \right\rfloor +1 \leq i \leq n.
\end{align}
The labels $\text{LHS}, \text{RHS}$ denote that the components contribute to the left hand side or the right hand side of the Lorenz curve maximum.
The maximum is included in the $\text{LHS}$ interval.
\begin{table}[h]
  \renewcommand{\tablename}{Supplementary Table}
  \def\arraystretch{1.5}
  \centering
  \begin{tabular}{c|c|c|r|r}
    \multicolumn{3}{c|}{Case} & \multicolumn{1}{c}{$m_{i}$} & \multicolumn{1}{|c}{$w_{i}$} \\[0.5ex]\hline
    \multirow{4}{*}{\raisebox{-4ex}{\rotatebox[origin=c]{90}{$0\leq \epsilon < \frac{3}{7}$}}} & \hspace{0.8ex}\multirow{2}{*}{\raisebox{-1ex}{\rotatebox[origin=c]{90}{$n$ even}}}\hspace{0.8ex} & LHS & $8^{2i}\binom{n}{2i}$ & $u^{2i}v^{n-2i}$ \\
    & & RHS & $8^{n-2i}\binom{n}{2i}$ & $-u^{n-2i}v^{2i}$ \\ \cline{2-5}
    & \multirow{2}{*}{\raisebox{-2ex}{\rotatebox[origin=c]{90}{$n$ odd}}} & LHS & $8^{2i+1}\binom{n}{2i+1}$ & $u^{2i+1}v^{n-2i-1}$ \\
    & & RHS & $8^{n-2i-1}\binom{n}{2i+1}$ & $-u^{n-2i-1}v^{2i+1}$ \\ \hline
    \multirow{4}{*}{\raisebox{-4ex}{\rotatebox[origin=c]{90}{$\frac{3}{7}\leq \epsilon < \frac{3}{4}$}}} & \multirow{2}{*}{\raisebox{-1ex}{\rotatebox[origin=c]{90}{$n$ even}}} & LHS & $8^{n-2i}\binom{n}{2i}$ & $u^{n-2i}v^{2i}$ \\
    & & RHS & $8^{2i}\binom{n}{2i}$ & $-u^{2i}v^{n-2i}$ \\ \cline{2-5}
    & \multirow{2}{*}{\raisebox{-2ex}{\rotatebox[origin=c]{90}{$n$ odd}}} & LHS & $8^{n-2i}\binom{n}{2i}$ & $u^{n-2i}v^{2i}$ \\
    & & RHS & $8^{2i}\binom{n}{2i}$ & $-u^{2i}v^{n-2i}$ \\ \hline
  \end{tabular}
  \caption{Wigner components $w_{i}(n, \epsilon)$ of $\rho_{\rm{S}}(\epsilon)^{\otimes n}$ along with their multiplicities $m_{i}(n, \epsilon)$, with $0 \leq i \leq n$.
  The expressions change depending on the error $\epsilon$, the parity of the number of copies $n$ and whether the index $i$ is lower or higher than the index of the Lorenz curve maximum (LHS or RHS).
  The term $2i$ is considered $\hspace{-6pt}\mod{\hspace{-2pt}(n+1)}$, so that $2i\ \hspace{-6pt}\mod{\hspace{-2pt}(n+1)} = 2i - n - 1$ if $i > \left\lfloor n/2 \right\rfloor$.}
  \label{tab:lcsu}
\end{table}

For unital protocols, the rescaled distribution is simply proportional to the Strange state distribution,
\begin{equation}
	W_{\rho_{\rm{S}}(\epsilon)^{\otimes n}|\left(\id/3\right)^{\otimes n}} = 9^n W_{\rho_{\rm{S}}(\epsilon)^{\otimes n}},
\end{equation}
therefore expressions for the Lorenz curve coordinates follow from summing up the Wigner components in decreasing order.
As a result, every Strange state Lorenz curve for a unital protocol contains $n$ elbows. In the case of accidental degeneracies among Wigner components, some of these elbows may be straightened out.
Including the boundary points $(x_{-1}, L_{-1}) \coloneqq (0,0)$ and $(x_{n}, L_{n}) \coloneqq (1,1)$, we label these elbows as 
\begin{equation*}
\{(x_{i}, L_{i})\}_{i=-1,0,\dots,n},
\end{equation*}
where the coordinates are given by
\begin{equation}
	(x_{i}, L_{i}) = \left( \frac{1}{9^n}\sum_{\ell=0}^i m_{\ell}, \sum_{\ell=0}^i m_{\ell} w_{\ell} \right).
\end{equation}
The Lorenz curve maximum is the $\lfloor n/2 \rfloor$-th elbow and its coordinates are calculated by collecting all the positive Wigner components,
\begin{align}
	x_{\lfloor n/2 \rfloor} &= \frac{1}{2}\left(1 + \left(\frac{7}{9}\right)^n\right), \\
	L_{\lfloor n/2 \rfloor} &= \frac{1}{2}\left (1 + \left(\frac{15 - 8\epsilon}{9}\right)^n \right).
\end{align}

In Supplementary Table~\ref{tab:lcsu_coord_elb_app}, we present the elbow coordinates of the $n$-copy, $\epsilon$--noisy Strange state Lorenz curve for a unital protocol for any combination of physical parameters $n, \epsilon$.
\begin{table}[h]
  \renewcommand{\tablename}{Supplementary Table}
  \def\arraystretch{1.5}
  \centering
  \begin{tabular}{c|c|c|r|r}
\hline 
    \multirow{4}{*}{\raisebox{-5ex}{\rotatebox[origin=c]{90}{$0\leq \epsilon < \frac{3}{7}$}}} & \hspace{0.8ex}\multirow{2}{*}{\raisebox{-3ex}{\rotatebox[origin=c]{90}{$n$ even}}}\hspace{0.8ex} & LHS & $\Phi_+\left(2i;n,\frac{8}{9}\right)$ & $\left( \frac{5}{3} - \frac{8}{9}\epsilon\ \right)^n \Phi_+\left(2i;n,\frac{12-4\epsilon}{15-8\epsilon}\right)$ \\
    & & RHS & $\Phi_-\left(2i;n,\frac{1}{9}\right)$ & $- \left( \frac{5}{3} - \frac{8}{9}\epsilon\ \right)^n\Phi_-\left(2i;n,\frac{3-4\epsilon}{15-8\epsilon}\right)$ \\ \cline{2-5}
    & \multirow{2}{*}{\raisebox{-3ex}{\rotatebox[origin=c]{90}{$n$ odd}}} & LHS & $\Phi_-\left(2i;n,\frac{8}{9}\right)$ & $\left( \frac{5}{3} - \frac{8}{9}\epsilon\ \right)^n \Phi_-\left(2i;n,\frac{12-4\epsilon}{15-8\epsilon}\right)$ \\
    & & RHS & $\Phi_-\left(2i;n,\frac{1}{9}\right)$ & $- \left( \frac{5}{3} - \frac{8}{9}\epsilon\ \right)^n\Phi_-\left(2i;n,\frac{3-4\epsilon}{15-8\epsilon}\right)$ \\ \hline
    \multirow{4}{*}{\raisebox{-5ex}{\rotatebox[origin=c]{90}{$\frac{3}{7}\leq \epsilon < \frac{3}{4}$}}} & \multirow{2}{*}{\raisebox{-3ex}{\rotatebox[origin=c]{90}{$n$ even}}} & LHS & $\Phi_+\left(2i;n,\frac{1}{9}\right)$ & $\left( \frac{5}{3} - \frac{8}{9}\epsilon\ \right)^n \Phi_+\left(2i;n,\frac{3-4\epsilon}{15-8\epsilon}\right)$ \\
    & & RHS & $\Phi_-\left(2i;n,\frac{8}{9}\right)$ & $- \left( \frac{5}{3} - \frac{8}{9}\epsilon\ \right)^n\Phi_-\left(2i;n,\frac{12-4\epsilon}{15-8\epsilon}\right)$ \\ \cline{2-5}
    & \multirow{2}{*}{\raisebox{-3ex}{\rotatebox[origin=c]{90}{$n$ odd}}} & LHS & $\Phi_+\left(2i;n,\frac{1}{9}\right)$ & $\left( \frac{5}{3} - \frac{8}{9}\epsilon\ \right)^n \Phi_+\left(2i;n,\frac{3-4\epsilon}{15-8\epsilon}\right)$ \\
    & & RHS & $\Phi_+\left(2i;n,\frac{8}{9}\right)$ & $- \left( \frac{5}{3} - \frac{8}{9}\epsilon\ \right)^n\Phi_+\left(2i;n,\frac{12-4\epsilon}{15-8\epsilon}\right)$ \\ \hline
  \end{tabular}
  \caption{Strange state Lorenz curve elbow coordinates for a unital protocol, with $0 \leq i \leq n$. 
  For the LHS cases, the $i$ in the first column labels the elbow coordinate $x_i$ and the second column gives $L_i=L(x_i)$. For RHS cases, the $i$ in the first column labels $x_i - x_{\lfloor n/2 \rfloor}$ and the second column gives $L_i - L_{\lfloor n/2 \rfloor}$.
  The coordinate expressions depend on the error $\epsilon$, the parity of the number of copies $n$ and the location of the elbow relative to the maximum (LHS or RHS).
  The term $2i$ is considered $\hspace{-6pt}\mod{\hspace{-2pt}(n+1)}$, so that $2i\ \hspace{-6pt}\mod{\hspace{-2pt}(n+1)} = 2i - n - 1$ if $i > \left\lfloor n/2 \right\rfloor$.
  The point $(x_{-1}, L_{-1}) \coloneqq (0,0)$ is not included in the table.
  }
  \label{tab:lcsu_coord_elb_app}
\end{table}

We can get explicit expressions for all $9^{n}$ points of the Lorenz curve $L_{\rho_{\rm{S}}(\epsilon)^{\otimes n}|(\id/3)^{\otimes n}}$, in terms of the elbow coordinates:
\begin{align}
    x_{ij} &= \left( 1-\frac{j}{m_{i}} \right) x_{i-1} + \frac{j}{m_{i}} x_{i}, \label{eq:x}\\
    L_{ij} &= \left( 1-\frac{j}{m_{i}} \right) L_{i-1} + \frac{j}{m_{i}} L_{i} \label{eq:l}
\end{align}
for $j = 1,\dots,m_{i}$ and $i=0,\dots,n$, where multiplicities $m_i$ are given in Table~\ref{tab:lcsu}.

One can consider the state 
\begin{equation*}
\rho_{\rm{S}}(\epsilon')^{\otimes n'} \otimes \left( \frac{1}{3}\id \right)^{\otimes (n-n')},
\end{equation*}
where tensoring with the maximally mixed state keeps the Lorenz curve unchanged, but increases the resolution of (the uniformly distributed) points.
The new point coordinates are given by:
\begin{align}
    &x_{ijk} = \left( 1-p_{ijk}\right) x_{i-1} + p_{ijk} x_{i} \label{eq:lcsu_xcoord}\\
    &L_{ijk} = \left( 1-p_{ijk} \right) L_{i-1} + p_{ijk} L_{i}, \label{eq:lcsu_lcoord}\\
    &\text{where } p_{ijk} = \frac{k + (j-1)9^{n-n'}}{9^{n-n'} m_{i}} \nonumber\\
    &\text{for } i=0,\dots,n',\ j = 1,\dots,m_{i}(n', \epsilon'),\ k = 1,\dots,9^{n-n'}. \nonumber
\end{align}
We can also combine the indices, by introducing a single index
\begin{equation}
    I(i,j,k) \coloneqq k + \left[ (j-1) + \sum_{\ell=0}^{i-1} m_{\ell}(n', \epsilon') \right]9^{n-n'},
\end{equation}
so that $I=1,2,\dots, 9^{n}$.
The elbow coordinates correspond to 
\begin{equation}
	I(i, m_{i}(n', \epsilon'), 9^{n-n'}) = \sum_{\ell=0}^{i} m_{\ell}(n', \epsilon'),\ i= 0,\dots,n'.
\end{equation}
The indexing $I$ is bijective, i.e.
\begin{equation}
	(i,j,k) = (i',j',k') \text{ iff } I(i,j,k) = I(i',j',k').
\end{equation}

\section*{Supplementary Note 4: Temperature-dependent distillation bounds from relative majorization}

Here we prove our main theorem, which bounds magic distillation rates in terms of free energies.
\begin{customthm}{7}\label{thm:free-energy}
	Consider a magic distillation protocol on qutrits that transforms $n$ copies of an $\epsilon$--noisy Strange state into $m$ copies of an $\epsilon'$--noisy Strange state, with depolarising errors $\epsilon' \leq \epsilon \leq 3/7$. We also allow pre/post-processing by local Clifford unitaries.
	
	Let $T =(k\beta)^{-1}$ be any finite temperature for the physical system and let $H= \sum_{k \in \mathbb{Z}_3} E_k |E_k\>\<E_k|$ be the Hamiltonian of each qutrit subsystem in its eigen-decomposition.
Assume that in the thermodynamic limit ($n,m \gg 1$), the protocol applied to the equilibrium state $\tau^{\otimes n} = (e^{-\beta H}/\Z)^{\otimes n}$ maps $\tau^{\otimes n} \longrightarrow \tau^{\prime \otimes m}$, where we write $\tau' = e^{-\beta H'}/\Z'$ for some Hermitian $H'$.

Then the asymptotic magic distillation rate $R = m/n$ is bounded as
\begin{equation}\label{eq:rate_bounds_proof}
	R \leq \dfrac{\log \big( 1-\frac{4}{3}\epsilon \big) + \beta (\phi - F)}{\log \big( 1-\frac{4}{3}\epsilon' \big) + \beta (\phi' - F')},
\end{equation}
where $F$ is the free energy of $\tau$,  and 
\begin{equation}\label{eq:phi}
	\phi = -\beta^{-1} \log \zeta
\end{equation}
with $\zeta$ given by the expressions
\begin{align}
	\zeta &= \sum_{k\in \mathbb{Z}_3} \alpha_k e^{-\beta E_k}, \\
	\alpha_k &= \<E_k| A_{\z_\star} |E_k\>,
\end{align}
and $W_\tau(\z)$ attaining a minimum at $\z=\z_\star$. The primed variables are defined similarly for the output system.
\end{customthm}

\begin{proof}
	For the sake of clarity, we write $\rho_n \coloneqq \rho_S(\epsilon)^{\otimes n}$, $\rho'_m \coloneqq \rho_S(\epsilon')^{\otimes m}$, $\tau_n \coloneqq \tau^{\otimes m}$ and $\tau'_m \coloneqq \tau'^{\otimes m}$.  We also assume, without loss of generality on the asymptotic distillation rate, that $n$ and $m$ are both even.

To establish the distillation bound we consider the distillation protocol that gives $\E( \rho_n) = \rho'_m$ for the magic states. 
We then consider that the protocol transforms the reference equilibrium state as $\E( \tau_n) = \tau'_m$.
Since we have a finite temperature we have that $\tau$ and $\tau'$ are full rank stabilizer states and have a strictly positive Wigner distribution. The input and output magic states generally have quasi-probability Wigner distributions. 
For any such protocol we therefore have that
\begin{equation}
	( W_{\rho_n}(\z), W_{\tau_n}(\z) ) \succ ( W_{\rho'_m}(\z), W_{\tau'_m}(\z) ),
\end{equation}
or, equivalently, in terms of the relevant Lorenz curves,
\begin{equation}
	L_{\rho_n |\tau_n}(x) \ge L_{\rho'_m |\tau'_m}(x) \mbox{ for all } x \in [0,1].
\end{equation}

We now consider the rescaled Wigner distribution $W_{\rho | \tau}(\z) \coloneqq W_\rho(\z)/W_\tau(\z)$, which is well-defined since $\tau$ is full-rank. 
Due to the multiplicative property of the Wigner distribution, the rescaled distribution is also multiplicative in the sense that
\begin{equation}
	W_{\rho \otimes \rho' | \tau \otimes \tau'} (\z_1 \oplus \z_2) = W_{\rho | \tau}(\z_1)W_{\rho' | \tau'}(\z_2),
\end{equation}
for any states $\rho, \rho'$ and any full-rank stabilizer states $\tau, \tau'$.
Therefore, we have that
\begin{align}
	W_{\rho_n |\tau_n} (\z) &= \prod_{i=1}^n W_{\rho|\tau}(\z_i)\\
	W_{\rho_n } (\z) &= \prod_{i=1}^n W_\rho(\z_i)
\end{align}
where $\z = \oplus_{i=1}^n \z_i \in \mathbb{Z}_3^{2n}$ is the phase space point for the full system in terms of those of the individual subsystems.

The points defining the Lorenz curve $L_{\rho_n |\tau_n}(x)$ are obtained from sorting the components of $W_{\rho_n |\tau_n}(\z)$ in non-increasing order and then computing the partial sums of $W_{\rho_n |\tau_n}(\pi(\z))$ where $\pi$ is the permutation that realizes the sorting. 
However, similar to the unital protocol analysis, we use the slope constraint (Proposition~\ref{prop:first_elb}) obtained by considering the line segments connecting the origin to the first elbow of both Lorenz curves.

The Wigner distribution of a single noisy Strange state consists of 1 negative component $W_{\rho}(\mathbf{0}) = -v(\epsilon)$ and 8 positive components $W_{\rho}(\z) = u(\epsilon)$ for $\z \neq \bmo$.
The Wigner distribution of the full-rank, stabilizer equilibrium state $\tau$ is $W_{\tau}(\z) > 0$ for all $\z \in \mathbb{Z}_3^2$.

Assume that the smallest component of the distribution $W_\tau(\z)$ is at $\z=\z_\star$. Since the magic content of the state is unchanged under a Clifford unitary $C$, we can instead consider the state
\begin{equation}
\rho_S(\epsilon) \rightarrow C \rho_S(\epsilon) C^\dagger = D_{\z_\star} \rho_S(\epsilon) D_{\z_\star}^\dagger,
\end{equation}
which has its single negative Wigner component $-v(\epsilon)$ at the point $\z_\star$.

The components of the rescaled distribution $W_{\rho_n|\tau_n}$ for this transformed state are given by
\begin{equation}
	\left(\frac{-v}{W_{\tau}(\z_\star)}\right)^{i_{\z_\star}} \prod_{\z \neq \z_\star} \left(\frac{u}{W_{\tau}(\z)}\right)^{i_{\z}},
\end{equation}
where the integer indices obey the following conditions:
\begin{align}
&0 \leq i_{\z} \leq n \mbox{ for all } \z \in \mathbb{Z}_3^2, \nonumber \\
&\sum_{\z \in \mathbb{Z}_3^2} i_{\z} = n.
\end{align}
We now compute the largest rescaled component.
Firstly, note that $n$ is even, so we require that $i_{\z_\star} \in \{0,2,\dots,n\}$ for the component to be positive.
Then, we have that $v \geq u$ because $\epsilon \leq 3/7$, and we have already ensured that $W_{\tau}(\z_\star) \leq W_{\tau}(\z)$ for all $\z \in \mathbb{Z}_3^2$.
Therefore, the largest rescaled component occurs when $i_{\z_{\star}} = n$ and $i_{\z} = 0$ for $\z \neq \z_\star$ and is equal to $(v/W_{\tau}(\z_\star))^n$.
The coordinates of the first Lorenz curve point after the origin are given by
\begin{equation}
	(x_0, L_0) = ((W_{\tau}(\z_\star))^n, v^n)
\end{equation}

Given a Hamiltonian decomposition,
\begin{equation}
	H = \sum_{k \in \mathbb{Z}_3} E_k |E_k\>\<E_k|,
\end{equation}
we now express the coordinates of the first point in terms of free energy quantities. We expand as follows
\begin{align}
	W_{\tau}(\z_\star) &= \frac{1}{3\Z}\tr\left[ A_{\z_\star}e^{-\beta H} \right] \nonumber\\
	&= \frac{e^{\beta F}}{3} \tr\left[ A_{\z_\star}\sum_{k \in \mathbb{Z}_3} e^{-\beta E_k} |E_k\>\<E_k| \right] \nonumber\\
	&= \frac{e^{\beta F}}{3} \sum_{r \in \mathbb{Z}_3} \alpha_k e^{-\beta E_k}
	= \frac{e^{\beta F}}{3} \zeta \nonumber\\
	&= \frac{e^{\beta (F - \phi)}}{3},
\end{align}
where we define
\begin{align}
\alpha_k \coloneqq \<E_k |A_{\z_\star} |E_k\> \\
\zeta \coloneqq  \sum_{k \in \mathbb{Z}_3} \alpha_k e^{-\beta E_k},
\end{align}
and a `magic free energy' term $\phi \coloneqq -\beta^{-1} \log \zeta$.
Therefore, the coordinates can be expressed as
\begin{equation}
	(x_0, L_0) = \left( \frac{e^{n\beta (F - \phi)}}{3^n}, v(\epsilon)^n \right).
\end{equation}
The output states obey the same conditions, so we also have the first point coordinates of the output Lorenz curve given by
\begin{equation}
	(x'_0, L'_0) = \left( \frac{e^{m\beta (F' - \phi')}}{3^m}, v(\epsilon')^m \right).
\end{equation}

If the largest rescaled component of a state is distinct with no multiplicities, then these coordinates correspond to the first elbow of the corresponding Lorenz curve, whereas if it appears multiple times, then the coordinates derived correspond to a point on the interior of the line segment connecting the origin to the first elbow.
In both cases, the distillation bound remains the same, as is clear by its derivation in Proposition~\ref{prop:first_elb}, given by $L_0/x_0 \geq L'_0/x'_0$, leading to
\begin{equation}
	\left( 3v(\epsilon)e^{-\beta (F - \phi)} \right)^{n} = \frac{L_0}{x_0}
	\geq \frac{L'_0}{x'_0} = \left( 3v(\epsilon')e^{-\beta (F' - \phi')} \right)^{m}.
\end{equation}
We note that $3v(\epsilon')e^{-\beta (F' - \phi')} > 1$ always, as $\rho_{\rm{S}}(\epsilon')$ is not a free state, so the initial slope of its Lorenz curve exceeds $1$.
Taking the natural logarithm on both sides and rearranging gives the bound
\begin{equation}
	\frac{m}{n} \leq \dfrac{\log \big( 1-\frac{4}{3}\epsilon \big) + \beta (\phi - F)}{\log \big( 1-\frac{4}{3}\epsilon' \big) + \beta (\phi' - F')},
\end{equation}
which completes the proof.
\end{proof}

The following result shows that if the energy eigenbasis is a stabilizer basis then the above bounds simplify further.
\begin{proposition} \label{sharp-phi}
	Let $H$ be a Hamiltonian $H$ with spectrum $\{E_k\}_{k \in \mathbb{Z}_d}$.
	If $H$ has a stabilizer eigenbasis, then $\phi = E_{k_{\star}}$ for some $k_{\star} \in \mathbb{Z}_d$.
\end{proposition}
\begin{proof}
	We first write $\alpha_k$ as
\begin{equation}
	\alpha_k = \<E_k | A_{\z_{\star}} |E_k\> = d W_{|E_k\>\<E_k|}(\z_{\star})
\end{equation}
Since each $|E_k\>\<E_k|$ is a stabilizer state, their Wigner distributions have orthogonal supports in the phase space, so only one of them, say $k = k_{\star}$ contains point $\z_{\star}$ in its support. 
Additionally, its distribution is uniform on its support, so $W_{|E_{k_\star}\>\<E_{k_\star}|}(\z_{\star}) = 1/d$.

Therefore, only the coefficient $\alpha_{k_{\star}}$ is non-zero and
\begin{align}
\phi &= -\beta^{-1}\log (\sum_k \alpha_k e^{-\beta E_k}) \nonumber \\
&=E_{k_{\star}} \hspace{-2pt}-\hspace{1pt} kT \log (d W_{|E_{k_\star}\>\<E_{k_\star}|}(\z_{\star})) =E_{k_{\star}}.
\end{align}

\end{proof}

\subsection*{Dropping pre/post-processing of the magic states}

If we drop the freedom to pre/post-process the magic states via Clifford unitaries, then the final expression for the bounds change and become a little more complex, since the largest component in the rescaled distribution depends on a number of factors. We illustrate this by the following result.

\begin{proposition}\label{prop:no-processing}
	Consider a magic distillation protocol on qutrits that transforms
\begin{equation*}
	\rho_S(\epsilon)^{\otimes n} \longrightarrow \E(\rho_S(\epsilon)^{\otimes n})=\rho_S(\epsilon')^{\otimes m} 
\end{equation*}
with $n, m \gg 1$.

Let each qutrit have a Hamiltonian $H$ with energies $E_0, E_1, E_2$ and stabilizer eigenstates, and define $E_{\rm max} = \max\{E_0, E_1, E_2\}$ and $E_s$ as the eigenvalue of the Hamiltonian eigenstate whose Wigner distribution overlaps the negative component of $\rho_S(\epsilon)$ in the phase space.
Let $T =(k\beta)^{-1}$ be any finite temperature and assume that in the thermodynamic limit ($n,m \gg 1$), the protocol applied to the equilibrium state $\tau^{\otimes n} = (e^{-\beta H}/\Z)^{\otimes n}$ maps $\tau^{\otimes n} \longrightarrow \tau^{\prime \otimes m}$, where we express state $\tau'$ as $\tau' = e^{-\beta H'}/\Z'$ for some Hermitian $H'$.

Define $\beta_\star = (k T_\star)^{-1}$ through the relation
\begin{equation}
	E_{\rm max} - E_s \eqqcolon kT_\star \ln 2,
\end{equation}
and define a threshold error,
\begin{equation}\label{eq:threshold}
	\epsilon_{\star}(\beta) \coloneqq 
	\begin{cases}
		3 - \dfrac{9}{4-2^{\beta/\beta_\star - 1}}, &\text{ for } \beta \leq \beta_\star \\
		0, &\text{ for } \beta > \beta_\star.
	\end{cases}
\end{equation}
Then, the distillation rate $R = m/n$ of the magic protocol is bounded as:
\begin{equation}
	R \leq
	\begin{cases}
		\dfrac{\ln \big( 1-\frac{4}{3}\epsilon \big) + \beta (E_s - F)}{\ln \big( 1-\frac{4}{3}\epsilon' \big) + \beta (E'_s - F')},\ &\epsilon \leq \epsilon_\star, \epsilon' \leq \epsilon'_\star, \vspace{10pt}\\
		\dfrac{\ln \big( 1-\frac{4}{3}\epsilon \big) + \beta (E_s - F)}{\ln \big( \frac{1}{2}-\frac{1}{6}\epsilon' \big) + \beta (E'_{\rm{max}} - F')},\ &\epsilon \leq \epsilon_\star, \epsilon' > \epsilon'_\star, \vspace{10pt}\\
		\dfrac{\ln \big( \frac{1}{2}-\frac{1}{6}\epsilon \big) + \beta (E_{\rm{max}} - F)}{\ln \big( 1-\frac{4}{3}\epsilon' \big) + \beta (E'_s - F')},\ &\epsilon > \epsilon_\star, \epsilon' \leq \epsilon'_\star, \vspace{10pt}\\
		\dfrac{\ln\big( \frac{1}{2}-\frac{1}{6}\epsilon \big) + \beta (E_{\rm{max}} - F)}{\ln \big( \frac{1}{2}-\frac{1}{6}\epsilon' \big) + \beta (E'_{\rm{max}} - F')},\ &\epsilon > \epsilon_\star, \epsilon' > \epsilon'_\star,
	\end{cases}
\end{equation}
where $F$ is the free energy of state $\tau$ and all primed quantities are similarly defined for the ouptut system.
\end{proposition}
\begin{proof}
The proof proceeds in the same manner as in Theorem~\ref{thm:free-energy} up to the point where we need to evaluate $W_{\rho_n |\tau_n} (\z)$ and $W_{\rho_n}(\z)$ explicitly, beyond which we do not have the freedom of additional Clifford processing.

The equilibrium state at inverse temperature $\beta$ on a single qutrit is given by $\tau = e^{-\beta H} / \Z$. Moreover, we have that $\tau$ is a full-rank stabilizer state, where $\beta \geq 0$ and $H = E_0 |\varphi_0\>\<\varphi_0| + E_1 |\varphi_1\>\<\varphi_1| + E_2 |\varphi_2\>\<\varphi_2|$ is an eigendecomposition of $H$.
The state $\tau$ can now be written as 
\begin{equation}
	\tau = \frac{e^{-\beta E_0}}{\Z} |\varphi_0\>\<\varphi_0| + \frac{e^{-\beta E_1}}{\Z} |\varphi_1\>\<\varphi_1| + \frac{e^{-\beta E_2}}{\Z} |\varphi_2\>\<\varphi_2|,
\end{equation}
where the eigenstates $\{\ket{\varphi_k}\<\varphi_k|\}$ are pure, orthonormal stabilizer states, which can be represented in terms of generalized Paulis. We are free to redefine the computational basis so it coincides with the basis of $H$. Abstractly, let $C$ be the unitary transforming each $|\varphi_k\>\<\varphi_k|$ to $|k\>\<k|$. Since the Clifford group is the normalizer of the Heisenberg-Weyl group, $C$ is a Clifford unitary. Therefore, $C$ maps $\tau$ to another stabilizer state that is diagonal in the computational basis, and we can assume without loss of generality that $\tau$ is diagonal in $\{|0\>,|1\>, |2\>\}$. This re-definition of the computational basis means that the coordinates on the discrete phase space are also permuted, so the negative Wigner component $-v(\epsilon)$ of the Strange state is in a new position. We denote by $E_s$ the eigenvalue of $H$ where the associated eigenvector has Wigner distribution overlapping the negative component of the state $C\rho_S(\epsilon)C^\dagger$ in the original phase space, hence also of the state Strange state $\rho_S(\epsilon)$ in the re-defined phase space. The eigenvalue $E_s$ is unique, since the eigenstates form an orthonormal basis.

The Wigner distribution of state $\tau$ is then given by
\begin{align}
	W_{\tau}(\z) &= \sum\limits_{k=0}^2 \frac{e^{-\beta E_k}}{\Z}W_{|k\>\<k|}(q, p) \nonumber\\
	&= \sum\limits_{k=0}^2 \frac{e^{-\beta E_k}}{\Z} \delta_{q, k} = \frac{e^{-\beta E_q}}{3\Z},
\end{align}
where $x$ labels one of the three vertical lines in the phase space.
The rescaled Wigner distribution $W_{\rho|\tau}(\z)$ can now be computed. It has $9$ components, but these come with multiplicities. In total, there are four distinct values on the phase space, as illustrated in Supplementary Figure~\ref{fig:pd_split}.
\begin{figure}[h]
    \renewcommand{\figurename}{Supplementary Figure}
    \centering
    \includegraphics[scale=0.35]{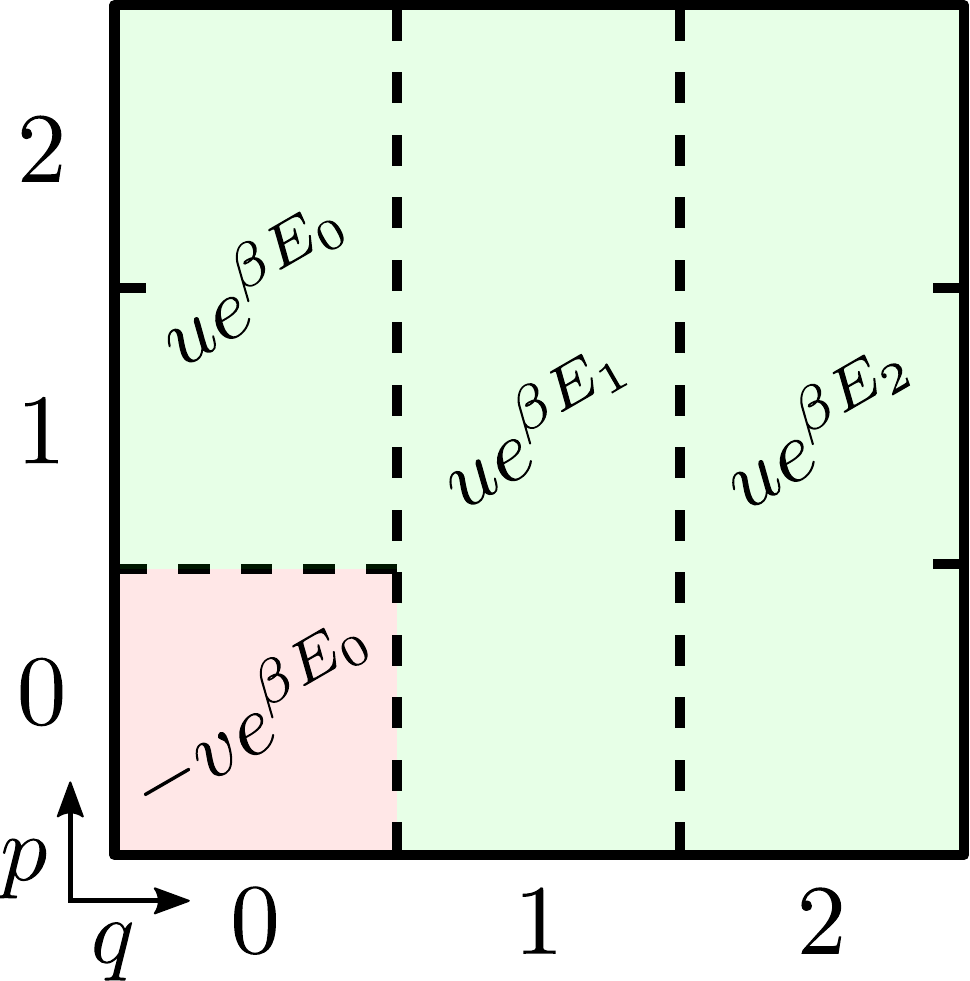}
    \caption{\textbf{Qutrit phase space regions for $W_{\rho | \tau}(\z)$.}
    Here, the negative component of the magic state overlaps the Wigner distribution of $|0\>$. The rescaled distribution attains a single value in each of the four regions, proportional to the value depicted in the region, see Eq.~(\ref{eq:bmw_rescaled}).
    }
    \label{fig:pd_split}
\end{figure}

We define vectors $\w(\rho), \w(\tau), \w(\rho|\tau)$ based on the values occurring in $W_\rho, W_\tau, W_{\rho|\tau}$ respectively and $\m$ as the vector of associated multiplicities of each value in $W_\rho(\z)$. Specifically, the component values and multiplicities of the relevant distributions in the four distinct regions are given by
\begin{align}
	\m &\coloneqq (1,2,3,3), \\
	\w(\rho) &\coloneqq (-v, u, u, u), \\
	\w(\tau) &\coloneqq \frac{1}{3\Z} \left( e^{-\beta E_0}, e^{-\beta E_0}, e^{-\beta E_1}, e^{-\beta E_2} \right), \\
	\w(\rho | \tau) &\coloneqq 3\Z \left( -v e^{\beta E_0}, u e^{\beta E_0}, u e^{\beta E_1}, u e^{\beta E_2} \right). \label{eq:bmw_rescaled}
\end{align}

Using this notation, the values and multiplicities of the $n$--copy distribution $\w(\rho_n |\tau_n)$ are computed in Proposition~\ref{ncopycomponents}. The values are given by 
\begin{align}\label{eq:ncopy_w_rescaled}
	[\w(\rho_n | \tau_n)]_{ijk} &= (3\Z)^{n} (-v)^{n-\alpha} u^{\alpha} e^{\beta (n-\alpha)E_s} e^{\beta ( i E_0 + j E_1 + k E_2 )},
\end{align}
where the indices $i,j,k$ are non-negative integers that obey the constraint $\alpha \coloneqq i+j+k \leq n$.
The multiplicity of this above value is $m_{ijk}$ with
\begin{equation}
	m_{ijk} = \frac{n!}{i!j!k!(n-\alpha)!} 2^i 3^j 3^k.
\end{equation}
The associated components of $\w(\rho_n)$ are given by
\begin{align}
	[\w(\rho_n)]_{ijk} &= (-v)^{n-\alpha} u^{\alpha}, \label{eq:ncopy_wrho}\\
	[\w(\tau)]_{ijk} &= (3\Z)^{-n} e^{-\beta (n-\alpha)E_s} e^{-\beta ( i E_0 + j E_1 + k E_2 )}. \label{eq:ncopy_wsigma}
\end{align}

In order to construct the $n$--copy Lorenz curve $L_{\rho_n|\tau_n}(x)$ we need to order the components of the distribution, $w(\rho | \tau)_{ijk}$ in decreasing order, and identify the sequence of indices that give us $W_{\rho_n}(\pi(\z))$.

However, in order to obtain a non-trivial bound, it again suffices to use the constraint from the first elbow $(x_0, L_0)$ of $L_{\rho_n|\tau_n}(\z)$. We therefore compute the largest component 
\begin{equation}
	w_{\rm max} \coloneqq \max_{i,j,k} [\w(\rho_n | \tau_n)]_{ijk},
\end{equation}
and determine the indices at which this occurs.
Putting in the values we obtain
\begin{align}
	&(3\Z)^{-n}w_{\rm max} = \nonumber\\
	&\max\limits_{i,j,k}\Big\{ (-v)^{n-\alpha} u^{\alpha} e^{\beta (n-\alpha)E_s} e^{\beta ( i E_0 + j E_1 + k E_2 )} \Big\}, \label{eq:max_slope}
\end{align}
where $0 \leq i,j,k \leq n$ and $\alpha \coloneqq i+j+k \leq n$.
Now for $0 \leq \epsilon \leq 3/7$, we have $v \geq u$. Since we assume that $n$ is even, we need the sum $\alpha = i+j+k$ to be even too, so that the expression is positive. 

Given an even value for $\alpha$, the term $v^{n-\alpha} u^{\alpha} e^{-\beta (n-\alpha)E_s}$ is fixed, so the expression is maximized by setting the coefficient of the highest energy $E_{\rm{max}}$ equal to $\alpha$.
Hence, we have
\begin{align}
	&w_{\rm max} = \nonumber\\
	&(3\Z)^{n} v^n e^{n\beta E_s}\max\limits_{\substack{\alpha = 0,2, \\ \dots,n-2,n}}{\Big\{ \left( \frac{u}{v} e^{\beta (E_{\rm{max}} - E_s)} \right)^{\alpha} \Big\}}.
\end{align}
If the expression $\frac{u(\epsilon)}{v(\epsilon)} e^{\beta (E_{\rm{max}} - E_s)}$ is less than $1$ then the maximum occurs at $\alpha=0$, otherwise it occurs at $\alpha = n$. For a fixed state $\tau$, this transition is determined by the value of the depolarising error parameter $\epsilon$ of the noisy magic state. The transition occurs at $\epsilon = \epsilon_\star$ where
\begin{equation}\label{eq:noise_transition}
	\frac{u(\epsilon_\star)}{v(\epsilon_\star)} e^{\beta (E_{\rm{max}} - E_s)} = \frac{3-\epsilon_\star}{6-8\epsilon_\star} e^{\beta (E_{\rm{max}} - E_s)} = 1.
\end{equation}
If $E_{\rm{max}} = E_s$, namely if the state negativity lies in the same phase space region as the highest energy, this threshold is constant in temperature and given by $\epsilon_{\star} = 3/7$. However, the condition that $\epsilon_\star \ge 0$ also implies a constraint on the effective temperature of the stabilizer state. Specifically, there is a threshold temperature value $\beta_\star$ given by
\begin{equation}
	\beta_{\star} \coloneqq \frac{1}{E_{\rm{max}} - E_s} \ln2,
\end{equation}
such that for the regime $0 \leq \beta \leq \beta_\star$ a threshold error $\epsilon_\star$ exists, and for $\beta > \beta_\star$ no such transition exists, so we choose $\epsilon_\star = 0$. 
Therefore, the transition value for the error is given by
\begin{equation}
	\epsilon_{\star}(\beta) \coloneqq 
	\begin{cases}
		3 - \dfrac{9}{4-2^{\beta/\beta_\star - 1}}, &\text{ for } \beta \leq \beta_\star \\
		0, &\text{ for } \beta > \beta_\star.
	\end{cases}
\end{equation}
The quantity $w(\rho_{\rm{S}} | \sigma)_{\rm{max}}$ is now given by
\begin{equation*}
w_{\rm max} =
	\begin{cases}
		(3\Z)^{n} v^n e^{n\beta E_s}, &\mbox{if }\epsilon \leq \epsilon_{\star},\ \hspace{3pt}\rm{(C1)}\\
		(3\Z)^{n} u^n e^{n\beta E_{\rm{max}}}, &\mbox{if }\epsilon > \epsilon_{\star}.\ \hspace{5pt}\rm{(C2)} 
	\end{cases}
\end{equation*}
Case $\rm{(C1)}$ corresponds to $(i,j,k) = (0,0,0)$, so the multiplicity is $m_{000} = 1$, while
Case $\rm{(C2)}$ corresponds to
\begin{equation}
	(i,j,k) = 
	\begin{cases}
	(0,n,0), &\text{if } E_{\rm{max}} = E_1, \\
	(0,0,n), &\text{if } E_{\rm{max}} = E_2,
	\end{cases}
\end{equation}
so the multiplicity in both cases is $3^n$.

Using that $F = -\beta^{-1} \log \Z$, the first elbow coordinates in the two cases are now given by
\begin{equation}\label{eq:first_elb_coords}
	(x_0, L_0) =
	\begin{cases}
		\left(\frac{1}{3^n} e^{-n\beta (E_s - F)}, v^n \right), &\epsilon \leq \epsilon_\star \vspace{10pt}\\
		\left( e^{-n\beta (E_{\rm{max}}-F)}, (3u)^n \right). &\epsilon > \epsilon_\star
	\end{cases}
\end{equation}

Similarly, considering the output magic state with respect to state $\sigma'$, the image of equilibrium state $\sigma$ under the magic protocol, we get output Lorenz curve coordinates,
\begin{equation}\label{eq:transformed_first_elb_coords}
	(x'_0, L'_0) =
	\begin{cases}
		\left(\frac{1}{3^{n'}} e^{-n\beta (E'_s - F')}, v(\epsilon')^{n'} \right), &\epsilon' \leq \epsilon'_\star \vspace{10pt}\\
		\left( e^{-n'\beta (E'_{\rm{max}}\hspace{-2.5pt}-F')}, (3u(\epsilon'))^{n'} \right), &\epsilon' > \epsilon'_\star
	\end{cases}
\end{equation}
There are four combinations of coordinates, depending on the error parameters $\epsilon, \epsilon'$ for the input and output states.
In each of these combinations, we simply use the first elbow constraint, as described in Proposition~\ref{prop:first_elb}, and manipulate the coordinates as in the proof of Theorem~\ref{thm:free-energy}, leading to the bounds in the statement of this theorem.
\end{proof}

\section*{Supplementary Note 5: Single-shot entropies on quasi-distributions and magic distillation bounds}

In this section, we analyse R\'{e}nyi entropies on quasi-distributions and use them to get magic distillation bounds.

\subsection*{R\'{e}nyi min-divergence from Lorenz curves}
We first demonstrate that the initial slope analysis on the Wigner quasi-distributions makes use of the min-relative divergence. We define 
\begin{equation}
	D_\infty(W_\rho || W_\tau) \coloneqq \log  \max_{\z} \frac{W_\rho(\z)}{W_\tau(\z)},
\end{equation}
which is well-defined since both the nominator and denominator in the above expression are strictly positive when $\tau$ is in the interior of $\F$.

\begin{customthm}{8}
	Let $\tau$ be in the interior of $\F$. Then $D_\infty(W_\rho || W_\tau)$ is well-defined for all $\rho$, and the following hold:
\begin{enumerate}
\item $D_\infty(W_\rho \hspace{1pt}||\hspace{1pt} W_\tau) \ge 0$ for all quantum states $\rho$.
\item  $D_\infty(W_\rho \hspace{1pt}||\hspace{1pt} W_\tau) = 0$ if and only if $\rho =\tau$.
\item $D_\infty(W_{\rho^{\otimes 2n}} \hspace{1pt}||\hspace{1pt} W_{\tau^{\otimes 2n}}) = n D_\infty(W_{\rho^{\otimes 2}} \hspace{1pt}||\hspace{1pt} W_{\tau^{\otimes 2}})$ for any $n \in \mathbb{N}$.
\item $D_\infty(W_\rho \hspace{1pt}||\hspace{1pt} W_\tau) \geq D_\infty(W_{\E(\rho)} \hspace{1pt}||\hspace{1pt} W_{\E(\tau)})$ for any free operation $\E$ such that $\E(\tau)$ is in the interior of $\F$.
\end{enumerate}
\end{customthm}
\begin{proof}
	Since $\tau$ is in the interior of $\F$, its Wigner function obeys $W_\tau(\z) >0$ for all points $\z$ in the phase space. 
In general, $W_\rho(\z)$ is a quasi-distribution, but the given form of $\alpha$ ensures that $W_{\rho}(\z)^\alpha \geq 0$. 
Therefore $D_\infty (W_\rho || W_\tau)$ is always well-defined.

1. We have that 
\begin{equation}
	D_\infty(W_\rho || W_\tau) =\log  \max_{\z} \frac{W_\rho(\z)}{W_\tau(\z)},
\end{equation}
so $2^{D_\infty(W_\rho || W_\tau)}$ equals the slope of the Lorenz curve $L_{\rho|\tau}(x)$ at $x=0$. Since $L_{\rho|\tau}(x)$ is a concave function passing through $(0,0)$ and $(1,1)$ this implies that $L_{\rho |\tau}(x) \ge x$ for all $x \in [0,1]$, and in particular its slope at $x=0$ is always greater than or equal to $1$, which implies $D_\infty(W_\rho || W_\tau) \geq 0$ for all $\rho$.

2. $D_\infty(W_\rho || W_\tau) = 0$ implies that $\max_{\z} \frac{W_\rho(\z)}{W_\tau(\z)}=1$ and the initial slope of $L_{\rho|\tau}(x)$ is $1$. From the concavity of the function and the fact that $L_{\rho|\tau}(1)=1$, this implies that the slope of $L_{\rho|\tau}(x)$ must equal $1$ throughout the interval $[0,1]$ and this, together with the definition of the Lorenz curve, implies that $W_\rho(\z)/W_\tau(\z) = 1$ for all $\z$. Since the Wigner representation $\rho \mapsto W_\rho(\z)$ is bijective this implies that $\rho = \tau$ only. The converse holds by inspection.

3. Given a vector $\w \in \mathbb{R}^N$, it is generally the case that $\max_k w_k \ne \max_k |w_k|$. 
However, we do have that
\begin{equation}
	\max_{k_1,k_2} (w_{k_1} w_{k_2}) = \max_{k_1,k_2} |w_{k_1} w_{k_2}|,
\end{equation}
and if additionally $w_k \geq 0$ for all $k$, then we also have for $2n$ copies that
\begin{equation}
	\max_{k_1, k_2, \dots, k_{2n}} \left( w_{k_1}w_{k_2}\cdots w_{k_{2n}} \right) = (\max_k |w_k|)^{2n},
\end{equation}
and also,
\begin{align}
	&\max_{k_1, k_2, \dots, k_{2n}}\left( w_{k_1}w_{k_2}\cdots w_{k_{2n}} \right) = \nonumber\\
	&\max_{k_1, k_2, \dots, k_{2n}}\left( |w_{k_1}w_{k_2}|\cdots |w_{k_{2n-1}}w_{k_{2n}}| \right) = \nonumber \\
	&\left( \max_{k_1 ,k_2}\left( w_{k_1} w_{k_2} \right) \right)^n.
\end{align}

If we now let $\w = W_{\rho|\tau}(\z) \coloneqq W_{\rho}(\z) /W_{\tau}(\z)$ we have
\begin{equation}
\max_{\z_1, \z_2, \dots ,\z_{2n}} \prod_{k=1}^{2n} W_{\rho|\tau}(\z_k) = \left( \max_{\z_1, \z_2} W_{\rho|\tau}(\z_1)W_{\rho|\tau}(\z_2) \right)^n
\end{equation}
and therefore taking logarithms we have that $D_\infty(W_{\rho^{\otimes 2n}} || W_{\tau^{\otimes 2n}}) = n D_\infty(W_{\rho^{\otimes 2}} || W_{\tau^{\otimes 2}})$ as required.

4. The free operation $\E$ in the Wigner representation corresponds to a stochastic map, which sends $W_\tau(\z)$ to $W_{\E(\tau)}(\z)$, which is another strictly positive probability distribution on the phase space, due to $\E(\tau)$ being in the interior of $\F$. Moreover, $(W_\rho, W_\tau ) \succ (W_{\E(\rho)} , W_{\E(\tau)})$. As shown in the main text this condition holds if and only if $L_{\rho|\tau}(x) \ge L_{\E(\rho)|\E(\tau)}(x)$ for all $x$. In particular, this implies the slope at the origin of the input curve is never less than the slope at the origin of the output curve, and hence the result follows.
\end{proof}

\subsection*{Well-defined R\'{e}nyi entropies on quasi-distributions}

We now make use of known results from majorization theory to establish the Schur-concavity of the R\'{e}nyi entropy $H_\alpha$ on Wigner quasi-distributions for a subset of $\alpha$ values. We consider the set of all quasi-distributions that is a hyper-plane in $\mathbb{R}^N$ given by
\begin{equation}
	\mathcal{Q}_N = \left\{ \w \in \mathbb{R}^N: \sum_{i=1}^N w_i = 1 \right\}.
\end{equation}

We define the notion of Schur-convex and Schur-convave functions on a subset $\D \subseteq \mathbb{R}^N$.
\begin{definition} 
	A continuous, real-valued function $f$ defined on $\mathcal{D} \subseteq \mathbb{R}^N$ is called Schur-convex on the domain $\mathcal{D}$ if $\p \prec \q$ on $\mathcal{D}$ implies that $f(\p) \le f(\q)$. A continuous function $f:\mathcal{D} \rightarrow \mathbb{R}$ is called Schur-concave if $(-f)$ is Schur-convex.
\end{definition}
For any subset $\D \subseteq \mathbb{R}^N$ we define
\begin{equation}
	\mathcal{D}^\downarrow \coloneqq \{ \w \in \mathcal{D} : w_1 \ge w_2 \ge \cdots \ge w_N\}.
\end{equation}

We now have the following fundamental result in majorization theory.
\begin{proposition}[Schur-Ostrowski theorem~\cite{app:Schur_1923, app:Ostrowski_1952, app:marshall}] 
	Let $\D\subseteq \mathbb{R}^N$ be a convex set with non-empty interior, and invariant under permutations of vector components. Let $f: \D \rightarrow \mathbb{R}$ be a continuously differentiable function on the interior of $\D$ and continuous on the whole of $\D$. Then $f$ is Schur-convex on $\D$ if and only if $f$ is a symmetric function on $\D$ and 
\begin{equation}
	\partial_1 f(\w) \ge \partial_2 f(\w) \ge \cdots \ge \partial_N f(\w),
\end{equation}
for all $\w$ in the interior of $\mathcal{D}^\downarrow$.
\end{proposition}
Since the set $\mathcal{Q}_N$ of quasi-distributions obey the necessary conditions, we get the following result:
\begin{proposition}
	Let $f$ be a real-valued continuous function defined on $\mathcal{Q}_N$ that is continuously differentiable on the interior of $\mathcal{Q}_N$. 
Then, $f$ is Schur-convex on $\mathcal{Q}_N$ if and only if it is a symmetric function and
\begin{equation}
	\partial_1 f(\w) \ge \partial_2 f(\w) \ge \cdots \ge \partial_N f(\w),
\end{equation}
for all quasi-distributions $\w$ in the interior of $\mathcal{Q}^\downarrow_N$.
\end{proposition}

The classical $\alpha$-R\'{e}nyi entropy on a probability distribution $\p = (p_i)$ is given by
\begin{equation}
	H_\alpha(\p) \coloneqq \frac{1}{1-\alpha} \log \sum_i p_i^\alpha,
\end{equation}
where $\alpha \ge 0$. This means the function is well-defined on $\mathcal{Q}_N^+$, namely the set of quasi-distributions with non-negative components. We wish to extend the domain to the whole of $\mathcal{Q}_N$, but this requires that $w_i^\alpha$ be a real number for any $w_i \in \mathbb{R}$. 

The exponent function $x \mapsto x^\alpha$ is well-defined and real-valued if $\alpha$ is a rational number of the form $\alpha = r / s$ where $s$ is an odd integer and $r \in \mathbb{N}$. However, we also require that $\sum_i w_i^\alpha > 0$ in order for the logarithm to return a real value. To ensure this, we restrict to $\alpha = r/s$ with $r=2a$ being an even integer and $s=2b-1$ being an odd integer, for some integers $a,b\in  \mathbb{N}$, which implies that the following are all well-defined, real-valued expressions,
\begin{equation*}
	0 \le w_i^\alpha = w_i^{\frac{2a}{2b-1}} = \left (w_i^{\frac{1}{2b-1}}\right )^{2a} =\left (w_i^{2a}\right )^{\frac{1}{2b-1}} = |w_i|^{\frac{2a}{2b-1}}.
\end{equation*}
We note that previous work~\cite{app:Manfredi_2000} has considered the $\alpha=2$ entropy of Wigner distributions, and other work exists that uses the Wehrl entropy, based on the Hussimi function of a quantum state~\cite{app:Gnutzmann_2001}.

We next show that $H_\alpha (\w)$ for $\alpha = 2a/(2b-1)$ is a Schur-concave function on $\mathcal{Q}_N$ provided that $a \ge b$.
\begin{customthm}{9}\label{thm:HSchur}
	If $\alpha = \frac{2a}{2b-1}$ for positive integers $a,b$ with $a \geq b$, then $H_\alpha(\w)$ is well-defined on the set of quasi-distributions $\mathcal{Q}_N$, and moreover if $\w \succ \w'$ for two quasi-distributions $\w, \w'\in \mathcal{Q}_N$ then $H_\alpha (\w) \leq H_\alpha(\w')$.
\end{customthm}
\begin{proof}
We have $\alpha = \frac{2a}{2b-1} \geq \frac{2b}{2b-1} > 1$ and $\alpha$ is a rational with even numerator and odd denominator, so $\sum_i w_i^\alpha > 0$ and $H_\alpha(\w)$ is well-defined for all quasi-distributions $\w$.

We also have that $H_\alpha(\w)$ is well-defined, continuous, differentiable on the interior of $\mathcal{Q}_N$ and symmetric in the components of $\w$.

Consider the partial derivatives
\begin{equation}
	\frac{\partial H_\alpha}{\partial w_i} = \frac{\alpha}{(\alpha-1)\sum_k{w_k^\alpha}}\left( -w_i^{\alpha-1} \right).
\end{equation}
For $\alpha= \frac{2a}{2b-1}$,  with $a\ge b$, we have that
\begin{equation}
	\frac{\alpha}{(\alpha-1)\sum_k{w_k^\alpha}} > 0.
\end{equation}
The first derivative of the function $g(w) \coloneqq -w^{\alpha - 1}$ is given by
\begin{equation}
	g'(w) = -(\alpha - 1) w^{\alpha-2} = -(\alpha - 1)w^{\frac{2a-4b+2}{2b-1}} \leq 0,
\end{equation}
because $\alpha > 1$ and
\begin{equation}
	w^{\frac{2a-4b+2}{2b-1}} = \left(w^{\frac{a-2b+1}{2b-1}}\right)^{2} \geq 0,
\end{equation}
so $g(w)$ is non-increasing in $w$.

Therefore, whenever $w_i \geq w_j$, we have that 
\begin{equation}
	-w_i^{\alpha-1} \leq -w_j^{\alpha-1},
\end{equation}
which implies that
\begin{equation}
	\frac{\partial H_\alpha}{\partial w_i}(\w) \leq \frac{\partial H_\alpha}{\partial w_j}(\w),
\end{equation}
for any $\w$ in the interior of $\mathcal{Q}_N^\downarrow$.
Therefore, $H_\alpha$ is Schur-concave on $\mathcal{Q}_N$ and
\begin{equation}
	H_\alpha(\w) \leq H_\alpha(\w'),
\end{equation}
for any quasi-distributions $\w, \w'$ that obey $\w \succ \w'$.
\end{proof}
While we have integers $a,b$ such that $0 < \alpha=\frac{2a}{2b-1} < 1$ and $H_\alpha(\w)$ is well-defined, it turns out that monotonicity does not hold if we drop the condition $a \ge b$. If $\alpha =2a/(2b-1)$ with $\alpha < 1$, then $g(w) \coloneqq \frac{-w^{\alpha-1}}{\alpha-1}$ is no longer monotonic for all $w \in \mathbb{R}$, and the problem occurs when comparing $w_i < 0$ and $w_j >0$. As an example of this dependence of monotonicity on the domain of the function, consider the function $g(x) = \frac{1}{x^3}$ which is monotone decreasing on both $x<0$ and $x>0$ separately, however it is not monotone on the full real-line.

We note that if $\alpha =r/s$ with both $r$ and $s$ odd, then $H_\alpha$ is not well-defined for all quasi-distributions, although if the actual set of quasi-distributions has a sufficiently bounded levels negativity, then $\log \sum_{\z} W_\rho(\z)^\alpha$ can still be obtained for $r$ odd, provided the total sum is never negative. This occurs for Wigner distributions, where we have $|W_\rho(\z)| \leq 1/d$ for a $d$--dimensional quantum system.

We also note that the set $F \coloneqq \{2a/(2b-1): a,b \in \mathbb{N} \text{ and } a \geq b\}$ is dense in the reals $\mathbb{R}_{>1}$, as any rational $c/d$ with $c>d$ can be approximated by  $c2^n / (d2^n-1) \in F$ for sufficiently large $n$.

We have that $H_{\alpha}$ is additive on products of quasi-distributions, which we state for completeness.
\begin{proposition}\label{H_add}
	For any $\w \in \mathcal{Q}_N$ and $\w' \in \mathcal{Q}_{N'}$, we have
	\begin{equation}
		H_{\alpha}(\w\otimes \w') = H_{\alpha}(\w) + H_{\alpha}(\w'),
	\end{equation}
	where $\alpha = 2a/(2b-1)$ with positive integers $a\ge b$.
\end{proposition}
\begin{proof}
	\begin{align}
		H_{\alpha}(\w \otimes \w') &= \frac{1}{1-\alpha}\log \sum_{i,j} \left[\w_i \w'_j\right]^{\alpha} \nonumber\\
		&= \frac{1}{1-\alpha}\log  \left [ \sum_i\w_i^\alpha \sum_j\w_j^{\prime \alpha} \right ] \nonumber\\
				&= \frac{1}{1-\alpha}\log  \left [ \sum_i\w_i^\alpha\right ] + \frac{1}{1-\alpha}\log \left [ \sum_j\w_j^{\prime \alpha} \right ] \nonumber\\
		&= H_{\alpha}(\w) + H_{\alpha}(\w').
	\end{align}
\end{proof}

Applied to Wigner distributions $W_\rho (\z)$ for a quantum state $\rho$ we have
\begin{equation}
	H_\alpha(W_\rho) \coloneqq \frac{1}{1-\alpha} \log \sum_{\z} W_{\rho}(\z)^\alpha,
\end{equation}
where $\alpha$ can take values of the form $2a / (2b-1)$, for non-negative integers $a,b$. For the noisy Strange state $\rho = \rho_{\rm{S}}(\epsilon)$, this becomes
\begin{align}
	H_{\alpha}\left(W_{\rho_{\rm{S}}(\epsilon)}\right) = \frac{1}{1-\alpha}\log \left[ 8\left( \frac{1}{6} - \frac{1}{18}\epsilon \right)^\alpha + \left(-\frac{1}{3} + \frac{4}{9}\epsilon \right)^\alpha \right] \nonumber\\
\end{align}
\begin{figure}[t!]
    \renewcommand{\figurename}{Supplementary Figure}
    \centering
    \includegraphics[scale=0.35]{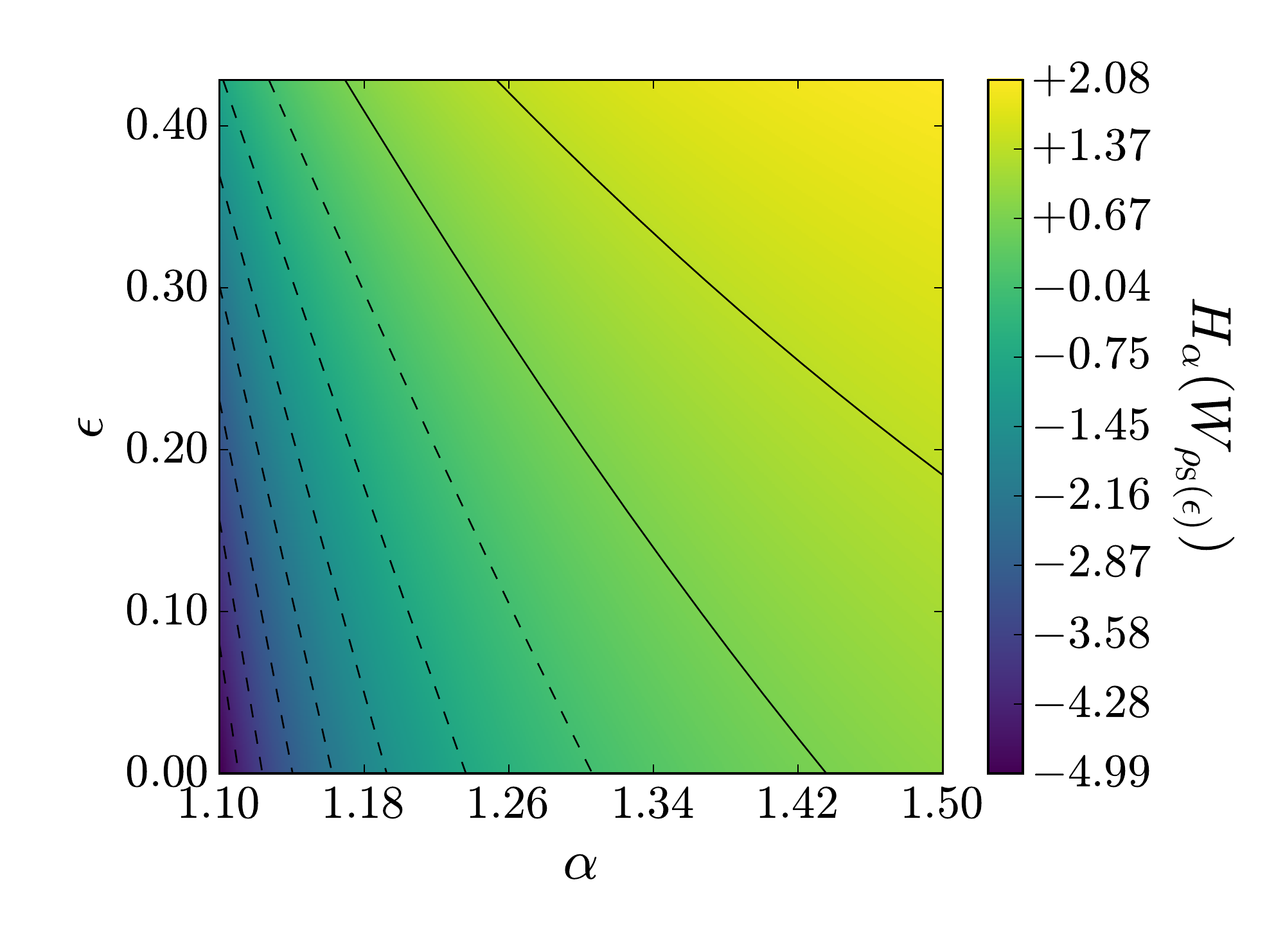}
    \caption{\textbf{R\'{e}nyi entropy $H_\alpha$ of the $\epsilon$-noisy Strange state.} $H_\alpha$ is increasing in $\alpha$ and $\epsilon$. The rightmost dashed contour line corresponds to $H_\alpha(W_{\rho_{\rm{S}}(\epsilon)}) = 0$. In particular, $H_\alpha(W_{\rho_{\rm{S}}(0)}) = 0$ occurs at $\alpha \approx 1.31$.
    }
    \label{fig:H}
\end{figure}
We now show the following equivalence between Wigner negativity and negative R\'{e}nyi entropies.
\begin{customthm}{10}
	A quantum state $\rho$ has Wigner negativity if and only if $H_\alpha(W_\rho) < 0$ for some $\alpha =  \frac{2a}{2b-1}$, with positive integers $a \geq b$.
\end{customthm}
\begin{proof} 
If we have $H_\alpha (W_\rho) <0 $ for $\alpha >1$, then it follows that $\log \sum_{\z} W_\rho (\z)^\alpha > 1$. 
However, it is known that $H_\alpha$ is always non-negative on probability distributions, so $W_\rho(\z)$ must be a strict quasi-distribution with negativity.

Conversely, suppose $\rho$ has negativity in its Wigner representation. This in particular implies that $\sum_{\z} |W_\rho (\z)| >1$. We now consider $\sum_{\z} W_\rho (\z)^{\frac{2a}{2b-1}}$ for positive integers $a\ge b$. We have that
\begin{align}
\sum_{\z} W_\rho (\z)^{\frac{2a}{2b-1}} &= \sum_{\z} |W_\rho (\z)|^{\frac{2a}{2b-1}} \nonumber \\
&= \sum_{\z} |W_\rho (\z)|^{1+ \epsilon},
\end{align}
where $\epsilon = \alpha -1=\frac{2a}{2b-1} - 1 >0$. By choosing the positive integers $a$ and $b$ sufficiently large we can make $\epsilon$ arbitrarily close to $0$. This implies that there exists a sequence $\epsilon_n= \alpha_n-1 =\frac{2a_n}{2b_n-1} - 1$ with integer pairs $a_n, b_n$ such that
\begin{equation}
\sum_{\z} |W_\rho (\z)|^{1+ \epsilon_n} \rightarrow \sum_{\z} |W_\rho (\z)| >1,
\end{equation}
as $n$ increases.
Since $\frac{1}{1-\alpha_n} <0$ for any $n$ it follows that $H_{\alpha_n} (W_\rho) = \frac{1}{1-\alpha_n} \log \sum_{\z} W_\rho (\z)^{\alpha_n}  <0 $ at some point in the sequence, which completes the proof.
\end{proof}
Therefore Wigner negativity coincides with negativity of a R\'{e}nyi entropy.

\subsection*{R\'{e}nyi divergences on quasi-distributions}
We now define the $\alpha$-R\'{e}nyi divergence for a quasi-distribution $\w \in \mathcal{Q}_N$ relative to a full-rank probability distribution $\r\in \mathcal{Q}_N^+$ as
\begin{equation}
	D_\alpha(\w||\r) \coloneqq \frac{1}{\alpha - 1} \log \sum_{i} w_i^\alpha r_i^{1-\alpha},
\end{equation}
where $\alpha=2a/(2b-1)$ for positive integers $a \ge b$. We now have the following result that relates the R\'{e}nyi divergence to the R\'{e}nyi entropy on a dense subset of probability distributions.
\begin{proposition}\label{H2D}
	Let $\w\in \mathcal{Q}_N$ be a quasi-distribution and $\r \in \mathcal{Q}_N^+$ a probability distribution with positive rational entries given by $r_i = a_i/K$ for positive integers $a_i$ and $K = \sum_{i=1}^N a_i$.
	Then,
	\begin{equation}
		H_\alpha(\Gamma_{\bma}(\w)) = K - D_\alpha(\w \hspace{1pt}||\hspace{1pt} \r),
	\end{equation}
	where $\alpha = 2a / (2b-1)$ for positive integers $a \ge b$.
\end{proposition}
\begin{proof}
	For the given form of $\alpha$ and positive $r_i,\ i=1,\dots,N$, $D_\alpha(\w || \r)$ is well-defined for all $\w$. From the definition of $\Gamma_{\bma}$ we have
	\begin{equation}
		\Gamma_{\bma}(\w) = \bigoplus_{i=1}^N w_i (1/a_i, 1/a_i, \dots, 1/a_i).
	\end{equation}
	This directly leads to
	\begin{align}
		H_\alpha(\Gamma_{\bma}(\w)) &= \frac{1}{1-\alpha} \log \sum_{i=1}^N a_i \left( \frac{w_i}{a_i} \right)^\alpha \nonumber\\
		&= \frac{1}{1-\alpha} \log \sum_{i=1}^N w_i^{\alpha} a_i^{1-\alpha} \nonumber\\
		&= \frac{1}{1-\alpha} \log K^{1-\alpha} \sum_{i=1}^N w_i^{\alpha} r_i^{1-\alpha} \nonumber\\
		&= K + \frac{1}{1-\alpha} \log \sum_{i=1}^N w_i^{\alpha} r_i^{1-\alpha} \nonumber\\
		&= K - D_\alpha(\w ||\r).
	\end{align}
\end{proof}
With this we now establish monotonicity for R\'{e}nyi relative divergences.
\begin{proposition}\label{DSchur}
	Let $\alpha = \frac{2a}{2b-1}$ for any positive integers $a,b$ with $a \geq b$. Let $\w \in \mathcal{Q}_N$, $\w' \in \mathcal{Q}_{N'}$ and $\r \in  \mathcal{Q}^+_N$, $\r' \in \mathcal{Q}^+_{N'}$ with positive rational components $a_i/K$ and $a'_i/K$ respectively.
	Then $D_\alpha(\w \hspace{1pt}||\hspace{1pt} \r) \geq D_\alpha(\w' \hspace{1pt}||\hspace{1pt} \r')$, whenever $(\w, \r) \succ (\w', \r')$.
\end{proposition}
\begin{proof}
	The statement $(\w, \r) \succ (\w', \r')$ is equivalent to $\Gamma_{\bma}(\w) \succ \Gamma_{\bma'}(\w')$.
	Therefore, due to Theorem~\ref{thm:HSchur} we have
	\begin{equation}
		H_\alpha(\Gamma_{\bma}(\w)) \leq H_\alpha(\Gamma_{\bma'}(\w'))
	\end{equation}
	which, due to Proposition~\ref{H2D}, is equivalent to
	\begin{equation}
		D_\alpha(\w \hspace{1pt}||\hspace{1pt} \r) \geq D_\alpha(\w' \hspace{1pt}||\hspace{1pt} \r').
	\end{equation}
\end{proof}
Since the rationals are dense in the real numbers we can assume any Wigner distribution considered has rational components without affecting results.
We now specialize to Wigner distributions of quantum states, and quantum operations that map the set of free states $\F$ into itself.
\begin{customthm}{11}\label{thm:Da_props}
	Let $\tau$ be in the interior of $\F$. 
	If $\alpha = \frac{2a}{2b-1}$ for positive integers $a,b$ with $a \geq b$, then the $\alpha$-R\'{e}nyi divergence $D_\alpha(W_\rho \hspace{1pt}||\hspace{1pt} W_\tau)$ is well-defined for all states $\rho$, and the following hold:
\begin{enumerate}
\item $D_\alpha(W_\rho \hspace{1pt}||\hspace{1pt} W_\tau) \ge 0$ for all quantum states $\rho$.
\item  $D_\alpha(W_\rho \hspace{1pt}||\hspace{1pt} W_\tau) = 0$ if and only if $\rho =\tau$.
\item $D_\alpha(W_{\rho^{\otimes n}} \hspace{1pt}||\hspace{1pt} W_{\tau^{\otimes n}}) = n D_\alpha(W_\rho \hspace{1pt}||\hspace{1pt} W_\tau)$ for any $n \in \mathbb{N}$.
\item $D_\alpha(W_\rho \hspace{1pt}||\hspace{1pt} W_\tau) \geq D_\alpha (W_{\E(\rho)} \hspace{1pt}||\hspace{1pt} W_{\E(\tau)})$ for any free operation $\E$ such that $\E(\tau)$ is in the interior of $\F$.
\end{enumerate}
\end{customthm}
\begin{proof}
	Since $\tau$ is in the interior of $\F$, its Wigner function obeys $W_\tau(\z) >0$ for all points $\z$ in the phase space. 
In general, $W_\rho(\z)$ is a quasi-distribution, but for $\alpha = 2a/2b-1$ we that $W_{\rho}(\z)^\alpha \geq 0$. 
Therefore $D_\alpha (W_\rho || W_\tau)$ is always well-defined and real-valued.

1. $D_\alpha$ is Schur-convex and every pair $(W_{\rho}, W_{\tau})$ majorizes the pair $(W_{\id/d}, W_{\id/d})$, so
\begin{align}
	D_\alpha(W_{\rho} \hspace{1pt}||\hspace{1pt} W_{\tau}) \geq D_\alpha(W_{\id/d} \hspace{1pt}||\hspace{1pt} W_{\id/d}) &= \nonumber\\
	\frac{1}{\alpha-1} \log \sum_{\z} W_{\id/d}(\z) &= 0.
\end{align}

2. In the inequality of property 1, equality holds iff $L_{\rho|\tau}(x) = L_{\id/d|\id/d}(x) = L_{\tau|\tau}(x)$ for all $x\in [0,1]$ due to Proposition~\ref{prop:rmajor} which in turn holds iff $\rho = \tau$.

3. This property follows from the multiplicativity of the Wigner distribution.
In particular,
\begin{align}
	&D_\alpha(W_{\rho^{\otimes n}} \hspace{1pt}||\hspace{1pt} W_{\tau^{\otimes n}}) = \nonumber\\
	&\frac{1}{\alpha - 1} \log \sum_{\z} W_{\rho^{\otimes n}}(\z)^\alpha W_{\tau^{\otimes n}}(\z)^{1-\alpha} = \nonumber\\
	&\frac{1}{\alpha - 1} \log \sum_{\z} \prod_{i=1}^n W_{\rho}(\z_i)^\alpha W_{\tau}(\z_i)^{1-\alpha} = \nonumber\\
	&\frac{1}{\alpha - 1} \log \prod_{i=1}^n \sum_{\z_i} W_{\rho}(\z_i)^\alpha W_{\tau}(\z_i)^{1-\alpha} = \nonumber\\
	&\sum_{i=1}^n \frac{1}{\alpha - 1} \log \sum_{\z'} W_{\rho}(\z')^\alpha W_{\tau}(\z')^{1-\alpha} = \nonumber\\
	&= n D_{\alpha}(W_\rho \hspace{1pt}||\hspace{1pt} W_\tau).
\end{align}

4. This follows immediately from the fact that $(W_\rho, W_\tau) \succ (W_{\E(\rho)}, W_{\E(\tau)})$ for any free quantum channel $\E$ that sends $\tau$ into the interior of $\F$, and the Schur-convexity of $D_\alpha$.
\end{proof}

We are now in a position to derive general magic state distillation bounds.
\begin{customthm}{12}
	Consider a general magic state distillation protocol on odd prime dimension qudits, that converts a magic state $\rho^{\otimes n} \longrightarrow \E(\rho^{\otimes n})=\rho'^{\otimes m}$ and let $\tau$ be any full-rank stabilizer reference state on a qudit. Then, the distillation rate $R \coloneqq m/n$ is upper bounded as
	\begin{equation}
		R \leq R_\alpha \coloneqq \frac{D_{\alpha}(W_\rho \hspace{1pt}||\hspace{1pt} W_\tau)}{\rlap{\raisebox{-0.2ex}{$\widetilde{\phantom{D}}$}}D_\alpha( \rho', \tau')},
	\end{equation}
	where $\alpha = \frac{2a}{2b-1}$ for any positive integers $a,b$ with $a \geq b$ and the average divergence per qudit
	\begin{equation}
\widetilde{D}_\alpha( \rho', \tau') \coloneqq \frac{1}{m} D_\alpha (W_{\rho'^{\otimes m}} \hspace{1pt}||\hspace{1pt} W_{\tau'_m}),
\end{equation}
between the output magic state $\rho'^{\otimes m}$ and $\tau'_m = \E(\tau^{\otimes n})$.
\end{customthm}
\begin{proof}
	The bound is a direct consequence of the properties of the $\alpha$--R\'{e}nyi divergence in Theorem~\ref{thm:Da_props}.
	
Due to the action of the magic protocol channel, we get
\begin{equation}	
	D_\alpha(W_{\rho^{\otimes n}} \hspace{1pt}||\hspace{1pt} W_{\tau^{\otimes n}}) \geq D_\alpha (W_{\rho'^{\otimes m}} \hspace{1pt}||\hspace{1pt} W_{\tau'_m}).\vspace{10pt}
\end{equation}

We can use the additivity to rewrite this as
\begin{equation}
	n D_\alpha(W_\rho \hspace{1pt}||\hspace{1pt} W_\tau) \geq m \frac{1}{m}D_\alpha(W_{\rho'^{\otimes m}} \hspace{1pt}||\hspace{1pt} W_{\tau'_m}).
\end{equation}

Since $\rho'^{\otimes m} \neq \tau'_m$, we have $D_\alpha(W_{\rho'^{\otimes m}} \hspace{1pt}||\hspace{1pt} W_{\tau'_m}) > 0$, which directly leads to the bound
\begin{equation}
	\frac{m}{n} \leq \frac{D_{\alpha}(W_\rho \hspace{1pt}||\hspace{1pt} W_\tau)}{\rlap{\raisebox{-0.1ex}{$\widetilde{\phantom{D}}$}}D_\alpha( \rho', \tau')}.
\end{equation}
\end{proof}

\onecolumngrid

\def\bibsection{\section*{Supplementary References}}

\providecommand{\noopsort}[1]{}\providecommand{\singleletter}[1]{#1}

\end{document}